\documentclass[11pt,reqno]{amsart}
\usepackage{amssymb, amsbsy, amsmath, amsfonts, amsthm, mathrsfs}
\usepackage{graphicx}
\usepackage{algorithm, algpseudocode, xcolor, hyperref}
\usepackage[margin=2.5cm]{geometry}
\usepackage{enumitem}

\newtheorem{theorem}{Theorem}
\newtheorem{prop}[theorem]{Proposition}
\newtheorem{lemma}[theorem]{Lemma}

\DeclareMathOperator*{\argmin}{arg\,min}

\newcommand{\blue}[1]{{\color{black}{#1}}}

\newcommand{\trho}{\tilde{\rho}}
\newcommand{\tx}{\tilde{x}}
\newcommand{\erf}{\text{erf}}
\newcommand{\WProx}{\text{WProx}}
\newcommand{\KL}{\mathrm{D}_\mathrm{KL}}

\newcommand{\tP}[1]{{\tilde{P}}}
\newcommand{\prox}{\text{prox}}

\graphicspath{{figs/}}

\title[Splitting Regularized Wasserstein Proximal Algorithms]{Splitting Regularized Wasserstein Proximal Algorithms\\ for Nonsmooth Sampling Problems}

\author{Fuqun Han}
\address{Department of Mathematics, University of California, Los Angeles, Los Angeles, CA, USA}
\email{fqhan@math.ucla.edu}

\author{Stanley Osher}
\address{Department of Mathematics, University of California, Los Angeles, Los Angeles, CA, USA}
\email{sjo@math.ucla.edu}

\author{Wuchen Li}
\address{Department of Mathematics, University of South Carolina, Columbia, SC, USA}
\email{wuchen@mailbox.sc.edu}

\begin{document}
\maketitle

\begin{abstract}
Sampling from nonsmooth target probability distributions is essential in various applications, including the Bayesian Lasso. We propose a splitting-based sampling algorithm for the time-implicit discretization of the probability flow for the Fokker-Planck equation, where the score function, defined as the gradient logarithm of the current probability density function, is approximated by the regularized Wasserstein proximal. When the prior distribution is the Laplace prior, our algorithm is explicitly formulated as a deterministic interacting particle system, incorporating softmax operators and shrinkage operations to efficiently compute the gradient drift vector field and the score function. 
We verify the convergence towards target distributions regarding R\'enyi divergences and Wasserstein-2 distance under suitable conditions. Numerical experiments in high-dimensional nonsmooth sampling problems, such as sampling from mixed Gaussian and Laplace distributions, logistic regressions, image restoration with $L_1$-TV regularization, and Bayesian neural networks, demonstrate the efficiency and robust performance of the proposed method. 
\end{abstract}

\section{Introduction}


Solving the Bayesian Lasso problem \cite{bayesian_lasso_2008} involves sampling from the target distribution  
\[
\rho^*(x) = \frac{1}{Z} \exp\left( -\beta(f(x) + g(x)) \right),
\]  
where \( x \in \mathbb{R}^d \), \( f: \mathbb{R}^d \to \mathbb{R} \) is the negative log-likelihood, \( g(x) = \lambda \|x\|_1 \) is the log-density of the Laplace prior for $\lambda >0$, \( \beta > 0 \) is a known parameter, and \( Z \) is an unknown normalization constant. 
The Bayesian Lasso is widely used as it simultaneously conducts parameter estimation and variable selection. It has broad applications in high-dimensional real-world data analysis, including cancer prediction \cite{chu2022application}, depression symptom diagnosis \cite{psycho_bay_lasso}, and Bayesian neural networks \cite{priors_bayesian_NN}.  

Most algorithms for sampling from $\rho^*$ rely on discretizing the overdamped Langevin dynamics. In each iteration, these algorithms evaluate the gradient of the logarithm of target distribution 
once and plus a Brownian motion perturbation to generate diffusion. However, the time-discretized overdamped Langevin dynamics presents several challenges. First, the gradient of \( g \) may not be well-defined, as in the case of $g$ being a \(L_1\) norm. Second, overdamped Langevin dynamics often perform inefficiently in high-dimensional sampling problems due to the fact that the variance of Brownian motion linearly depends on the dimension.

To address the first challenge, many proximal sampling algorithms, often with splitting techniques by regrading the sampling task as a composite optimization problem in Wasserstein space, have been extensively studied. \cite{pereyra2016proximal,durmus2018efficient,salim_splitting} use proximal operators to approximate the gradient of the nonsmooth log-density. Extended works include methods leveraging a restricted Gaussian oracle (RGO) \cite{Lee2021RGO, chen2023improved, liang2022proximal}, incorporating both sub-gradient and proximal operators \cite{habring2024subgradient},  and solving an inexact proximal map at each iteration \cite{benko2024langevin}. For a recent review, see \cite{pock_sampling}. 
In these works, the proximal map is often interpreted as a semi-implicit discretization of the Langevin dynamics with respect to the drift term. The present study also employs the proximal operator to approximate the gradient of nonsmooth terms, however, the proposed algorithm is fully deterministic as described below. 

Furthermore, to handle the second challenge, instead of considering the time discretization of the Langevin dynamic, we will analyze a deterministic interacting particle system obtained by the time-discretized probability flow ODE. Here, the ODE involves the drift function and the gradient logarithm of the current probability density function, named the score function, which induces the diffusion. Since this approach avoids simulating Brownian motion, it is independent of the sample space dimension. However, accurately approximating the score function presents a challenge of its own.

To approximate the evolution of the score function, \cite{BRWP_2023} derived a closed-form formula using the regularized Wasserstein proximal operator (RWPO). The RWPO is defined as the Wasserstein proximal operator with a Laplacian regularization term  (see Section \ref{sec_algo} for details). By applying Hopf–Cole transformations, the operator admits a closed-form kernel formula. It can be showed that the RWPO \blue{is a Fisher information regularized version of the classical JKO scheme \cite{jko1998} for the Fokker-Plank equation} and provides a first-order approximation to the evolution of the Fokker–Planck equation \cite{han2024convergence}, leading to an effective score function approximation.
The sampling algorithm based on RWPO named backward regularized Wasserstein proximal (BRWP), has been implemented in several studies \cite{BRWP_2023, TT_BRWP} with different computational strategies. Its backward nature comes from the implicit time discretization of the probability flow ODE for the score function term. However, a key challenge in implementing the BRWP kernel lies in approximating an integral over $\mathbb{R}^d$ to compute the denominator term.

In this work, we derive a computationally efficient closed-form update for BRWP without evaluating a high-dimensional integral for special nonsmooth functions, such as the $L_1$ norm. 
Following the restricted Gaussian oracle of BRWP with $L_1$ function, we derive an explicit formula of the sampling algorithm, in which samples interact with each other following an interacting kernel function. In particular, this kernel function is constructed by shrinkage operators and the softmax functions. Moreover, we also apply the splitting method and proximal updates for sampling problems with nonsmooth target density. 

\blue{We sketch the algorithm below. For particles $\{x_i^k\}_{i=1}^N$ in the $k$ iteration, when $g(x) = \lambda\|x\|_1$, the proposed iterative sampling scheme is
\begin{equation*}
    x_i^{k+\frac{1}{2}} = x_i^k - h\nabla f(x_i^k)\,,\quad x_i^{k+1} =x_i^{k+\frac{1}{2}} + \frac{1}{2} \left( S_{\lambda h}(x_i^{k+\frac{1}{2}}) -  \sum_{j=1}^N\text{softmax}(U(i,j)_j)x_j^{k+\frac{1}{2}} \right)\,, 
\end{equation*}
where $h>0$ is the time step size. The interacting kernel is defined as
\begin{equation*}
    U(i,j) := -\frac{\beta}{2} \left(\frac{\|x_i^{k+\frac{1}{2}} -x_j^{k+\frac{1}{2}} \|_2^2 - \| S_{\lambda h}(x_i^{k+\frac{1}{2}})-x_j^{k+\frac{1}{2}} \|_2^2}{2h} - \lambda\| S_{\lambda h}(x_j^{k+\frac{1}{2}} )\|_1\right)\,,
\end{equation*}
with
\[
\textrm{softmax}(x) = \left(\frac{\exp(x_j)}{\sum_{\ell=1}^N \exp(x_{\ell})}\right)_{1\leq j \leq N}\,,\quad x\in\mathbb{R}^d\,.
\]  
The shrinkage operator $S_{\lambda h}$ takes the form
\begin{equation*}
  S_{\lambda h}(x) :
 = \operatorname{sign}(x)\operatorname{ReLU}(|x| - \lambda h)\,,
\end{equation*}
with the rectified linear unit (ReLU) function $\operatorname{ReLU}(z) = \max\{z,0\}$ for $z\in\mathbb{R}$. We remark that the shrinkage operator is the proximal map of the $L_1$ norm, i.e., $ S_{\lambda h}(x)= \operatorname{prox}_{\lambda\|x\|_1}^h(x)$. }
 

Compared to algorithms based on splitting the overdamped Langevin dynamics with Brownian motion, as studied in \cite{pereyra2016proximal, durmus2018efficient, salim_splitting, Lee2021RGO, chen2023improved, liang2022proximal}, the proposed deterministic approach generally provides a better approximation to the target density empirically \blue{with similar computational complexity} and with a small number of particles. It also demonstrates faster convergence in high-dimensional sample spaces, benefiting from adapting the deterministic score function, as established in \cite{han2024convergence}.  
Several other works have investigated deterministic interacting particle systems for sampling, including Stein variational gradient descent methods  \cite{svgd_2016} and blob methods \cite{blob_sampling}. The proposed approach, however, leverages a kernel formulation derived directly from the solution of the Fokker–Planck equation, naturally incorporating information about the underlying dynamics, as reflected in the definition of \( U(i,j) \) above. Furthermore, the proposed kernel is closely related to the restricted Gaussian oracle \cite{Lee2021RGO} due to the definition of the kernel formula for RWPO (see section \ref{sec_algo} for more details). Hence, our computational implementation can also be regarded as an approximation to the restricted Gaussian oracle. 

The structure of this paper is as follows. Section \ref{sec_algo} presents the derivation of the BRWP-splitting sampling scheme with a detailed algorithm description. In particular, we introduce several kernels, each corresponding to a different particle-based approximation of the initial density. \blue{Section \ref{sec_analysis} establishes the convergence of the BRWP-splitting algorithm to the target density under two different sets of assumptions. The first scenario relies on the Poincaré inequality and interpolation techniques to demonstrate exponential convergence in Rényi divergence for smooth potentials. The second scenario adopts a convex optimization framework to prove the convergence in Wasserstein-2 distance for nonsmooth potentials. Section \ref{sec_TV_imaging} further extends our approach to incorporate additional regularization terms—specifically, $L_1$-TV regularization—by integrating primal-dual optimization techniques with the BRWP-splitting algorithm.} Finally, Section \ref{sec_NE} presents numerical experiments on mixture distributions, Bayesian logistic regression, several imaging applications, and Bayesian neural network training. Proofs and detailed derivations are included in the supplementary material.   

\section{Regularized Wasserstein Proximal and Splitting Methods for Sampling}
\label{sec_algo}

We are aiming to draw samples from probability distributions of the form  
\begin{equation}
\label{def_rho_L1}
    \rho^*(x) = \frac{1}{Z}\exp(-\beta(f(x)+\lambda\|x\|_1))\,,
\end{equation}  
where $x\in\mathbb{R}^d$, $f:\mathbb{R}^d \rightarrow \mathbb{R}$ is $L$-smooth, \( \beta=(k_B T)^{-1} \) with a temperature constant $T>0$ and the Boltzmann constant $k_B$, \( \lambda \) is a regularization parameter, and \( Z=\int_{\mathbb{R}^d} \exp(-\beta(f(y)+\lambda\|y\|_1))dy<+\infty\) is an unknown normalization constant. 

Sampling from such a distribution is widely used in parameter estimation, particularly under the framework of the Bayesian Lasso problem \cite{bayesian_lasso_2008}, which simultaneously performs estimation and variable selection. However, the nonsmoothness of the \(L_1\) norm poses significant challenges in developing theoretically sound and numerically efficient sampling algorithms. Beyond the Bayesian Lasso setting, the proposed scheme can also be applied to the case where \( g(x) \) is a nonsmooth function whose proximal operator is easy to compute. In this case, we consider sampling from the distribution  
\begin{equation}
\label{def_rho}
\rho^*(x) = \frac{1}{Z}\exp(-\beta(f(x)+g(x)))\,.
\end{equation}

\subsection{Discretization of Langevin dynamics and probability flow ODE}
In this section, we review the time discretization of the overdamped Langevin dynamics and regularized Wasserstein proximal operator to motivate the proposed algorithm. 

Denote $V = f+g$ for simplicity. To sample from $\rho^*$ in \eqref{def_rho}, a classical approach involves the overdamped Langevin dynamics at time $t$
\begin{equation}
\label{def_Langevin}
    dX_t = -\nabla V(X_t) dt + \sqrt{2\beta^{-1}} dB_t\,,
\end{equation}
where $X_t\in\mathbb{R}^d$ is a stochastic process, and $B_t$ is the standard Brownian motion in $\mathbb{R}^d$.
Denote $\rho_t$ as the probability density function of $X_t$. We know that the  Kolmogorov forward equation of the stochastic process $X_t$ satisfies the following Fokker–Planck equation:  
\begin{equation}
\label{def_FP}
\frac{\partial \rho_t}{\partial t} = \nabla\cdot(\rho_t\nabla V) + \beta^{-1}\Delta\rho_t = \beta^{-1}\nabla\cdot \left(\rho_t\nabla\log\frac{\rho_t}{\rho^*}\right)\,
\end{equation}
where we use the fact that $\rho_t\nabla\log\rho_t=\nabla\rho_t$ and $\nabla\log\rho^*=\nabla \log e^{-\beta V}=-\beta\nabla V$. 

\blue{From the stationary solution of the Fokker–Planck equation, we observe that the invariant distribution of the Langevin dynamics coincides with the target distribution \(\rho^*\).} However, directly applying the overdamped Langevin dynamics \eqref{def_Langevin} to sample from \eqref{def_rho_L1} presents several challenges. Firstly, the function $V$ is nonsmooth, which creates difficulties in the gradient computation. Secondly, the variance of the Brownian motion depends on the sample space dimension linearly, which slows down convergence, posing challenges for high-dimensional sampling tasks.  

To address the first issue, for a small stepsize $h>0$, one often utilizes the Moreau envelope  
\begin{equation}
\label{def_moreau}
    g_{h}(x) = \inf_{y\in\mathbb{R}^d} \left\{ g(y) + \frac{1}{2h}\|x-y\|_2^2 \right\}\,,
\end{equation}
which provides a smooth approximation to the nonsmooth function $g$; and the proximal operator  
\begin{equation}
\label{def_prox}
\prox_g^h(x) = \argmin_{y\in\mathbb{R}^d} \left\{ g(y) + \frac{1}{2h}\|x-y\|_2^2 \right\}\,,
\end{equation}
which provides a smooth approximation to the gradient of $g$ based on the relation  
\begin{equation}
\nabla g_h(x) = \frac{x-\prox_g^h(x)}{h}, \quad \text{for a convex function } g\,.
\end{equation}
These tools, along with splitting techniques, have been widely applied in nonsmooth sampling problems \cite{liang2022proximal,pula_wibisono,durmus2018efficient}. In this work, we also employ the proximal operator to approximate the gradient of nonsmooth functions.  

Furthermore, to tackle the second challenge, which arises from the linear dependence of the variance of Brownian motion and the dimension, we aim to avoid the simulation of Brownian motions in the sampling algorithm. Instead, we consider the evolution of particles $x_t\in\mathbb{R}^d$ governed by the probability flow ODE:  
\begin{equation}
\label{def_ODE}
dx_t = -\nabla V(x_t) dt - \beta^{-1} \nabla \log \rho_{t}(x_t) dt\,.
\end{equation}
Here, the diffusion is induced by the score function $\nabla\log\rho_t$. \blue{Although the individual particle trajectories governed by equation \eqref{def_ODE} differ from those of the stochastic dynamics in \eqref{def_Langevin}, the associated Liouville (or continuity) equation for \eqref{def_ODE} coincides with the Fokker–Planck equation \eqref{def_FP} \cite[Prop.,8.1.8]{ambrosio2008gradient}. This implies that simulating \eqref{def_ODE} also yields a particle system that can approximate the target distribution $\rho^*$.}

The primary challenge in discretizing the probability flow ODE \eqref{def_ODE} in time is the accurate approximation of the score function. For each discretized time point, since we can only access $N$ particles obtained from the previous iteration, kernel density estimation-based particle methods can be unstable and sensitive to the choice of bandwidth. To mitigate this, we consider a semi-implicit discretization of \eqref{def_ODE}, where the score function at the next time step is utilized. This results in the following iterative sampling scheme. 

Denote the time steps as \( t_k \) for \( k = 1, 2, \dots \), with a step size \( h = t_{k+1} - t_k > 0 \). Let \( x^k \) represent a particle at time \( t_k \), distributed according to the density \( \rho_k \), i.e., \( x^k \sim \rho_k \). Similarly, let \( x^{k+1} \sim \rho_{k+1} \), where \( \rho_{k+1} \) is the density at the next time step \( t_{k+1} = t_k + h \).
Then the semi-implicit discretization of the probability flow ODE in time is 
\begin{equation}
\label{def_semi_particle}
x^{k+1} = x^k - h \nabla V(x^k) - h\beta^{-1}\nabla\log\rho_{k+1}(x^k)\,.
\end{equation}
To compute $\rho_{k+1}$, one must approximate the evolution of the density function following the Fokker–Planck equation \eqref{def_FP}. 

\subsection{Regularized Wasserstein proximal operator}
A classical approach to approximate the Fokker-Planck equation is to consider the implicit JKO scheme \cite{jko1998}:  
\begin{equation}
\label{JKO}
   \rho_{k+1} = \argmin_{\rho\in \mathcal{P}_2(\mathbb{R}^d)}\,\,\beta^{-1}\KL(\rho\|\rho^*) + \frac{1}{2h}W^2_2(\rho,\rho_k)\,,
\end{equation}
where $\mathcal{P}_2(\mathbb{R}^d)$ is the set of probability measures in $\mathbb{R}^d$ with a finite second-order moment and  $\KL(\rho\|\rho^*)$ denotes the Kullback–Leibler (KL) divergence defined as \[
\KL(\rho\|\rho^*) := \int_{\mathbb{R}^d}\rho\log\frac{\rho}{\rho^*} dx\,. 
\]
Moreover, $W^2_2(\rho,\rho_k)$ represents the squared Wasserstein-2 distance, which can be expressed using Benamou-Brenier formula \cite{ambrosio2008gradient}:
\[
\frac{W^2_2(\rho_0,\rho_h)}{2h} := \inf_v \int_0^h\int_{\mathbb{R}^d}\frac{1}{2}\|v(t,x)\|_2^2\rho(t,x) dx dt\,,
\]
where the minimization is taken over vector fields $v\colon [0, h]\times \mathbb{R}^d\rightarrow \mathbb{R}^d$ subject to the continuity equation with fixed initial and terminal conditions:
\[
\frac{\partial \rho}{\partial t} + \nabla\cdot(\rho v)= 0\,,\quad \rho(0,x) = \rho_0(x)\,,\quad \rho(h,x) = \rho_h(x)\,.
\]
However, solving the JKO-type implicit scheme often requires high-dimensional optimization, which can be computationally expensive. We remark that many existing sampling algorithms exploit certain splitting of the JKO scheme \cite{liang2022proximal,salim_splitting,LMC_JKO_splitting} and employ the implicit gradient descent for the drift vector fields. The current work considers the implicit update regarding both the drift and the score functions simultaneously.

To derive a closed-form update to approximate the evolution of the Fokker–Planck equation, we start with the Wasserstein proximal operator with linear energy, as introduced in \cite{kernel_proximal}. By incorporating a Laplacian regularization term into the Wasserstein proximal operator and applying the Benamou–Brenier formula, we obtain the following regularized Wasserstein proximal operator (RWPO)
\begin{equation}
\label{Reg_Was}
  \WProx_V^{h,\beta}(\rho_k) := \argmin_{q\in \mathcal{P}_2(\mathbb{R}^d)}\inf_{v} \left\{ \int_{0}^{h} \int_{\mathbb{R}^d} \frac{1}{2} \|v(t, x)\|_2^2 \rho(t, x) \,dx\,dt + \int_{\mathbb{R}^d} V(x) q(x) \,dx\right\}\,,
\end{equation}  
where the minimization is taken over all vector fields \( v \) and terminal density \( q \), subject to the continuity equation with an additional Laplacian term and the initial condition:  
\begin{equation}\label{Reg_con_RWPO}
   \frac{\partial \rho}{\partial t}  + \nabla \cdot (\rho v) = \beta^{-1} \Delta \rho\,, \quad
    \rho(0, x) = \rho_k(x)\,, \quad \rho(h, x) = q(x)\,.
\end{equation}  

\blue{To figure out the connection between the JKO scheme and the RWPO operator, similar to \cite{chen2022improved_proximal}, we perform a change of variables by setting $u = v - \beta^{-1} \nabla \log \rho$. This leads to the reformulated problem
\begin{align*}
&\WProx_V^{h,\beta}(\rho_k)\\
= &\argmin_{q \in \mathcal{P}_2(\mathbb{R}^d)} \inf_{u} \Bigg\{\beta^{-1} \KL(q \| \rho^*) + \int_{0}^{h} \int_{\mathbb{R}^d} \frac{1}{2} \|u\|_2^2 \rho(t, x) \, dx \, dt + \beta^{-2} \int_0^h \int_{\mathbb{R}^d} \|\nabla \log \rho\|_2^2 \rho \, dx \, dt \Bigg\},
\end{align*}
subject to the continuity equation:
\begin{equation}\label{Reg_con}
\partial_t \rho + \nabla \cdot (\rho u) = 0, \quad
\rho(0, x) = \rho_k(x), \quad \rho(h, x) = q(x).
\end{equation}

We then define the Fisher-information regularized Wasserstein distance between $\rho_k$ and $q$ as
$$
\frac{W_{2}^{2,\beta}(\rho_k,q)}{2h} = \inf_{u} \left\{ \int_{0}^{h} \int_{\mathbb{R}^d} \frac{1}{2} \|u\|_2^2 \rho(t, x) \, dx \, dt + \beta^{-2} \int_0^h \int_{\mathbb{R}^d} \|\nabla \log \rho\|_2^2 \rho \, dx \, dt \right\},
$$
with $\rho^* = 1/Z\exp(-\beta V)$ and subject to the continuity constraint \eqref{Reg_con}. Consequently, the RWPO operator can be expressed as
\begin{equation}
\label{RWOP_regu_W2}
\WProx_V^{h,\beta}(\rho_k) = \argmin_{q \in \mathcal{P}_2(\mathbb{R}^d)} \beta^{-1}\KL(q \| \rho^*) + \frac{1}{2h} W_{2}^{2,\beta}(\rho_k, q).
\end{equation}

Comparing \eqref{RWOP_regu_W2} with the original JKO formulation \eqref{JKO}, we observe that the operator $\WProx_V^{h,\beta}$ corresponds to a Fisher-information regularized JKO scheme. The convergence properties of RWPO towards the original Fokker–Planck equation will be discussed in detail in Section~\ref{sec_analysis}.

To obtain an explicit form of the RWPO defined in \eqref{Reg_Was}, we introduce a Lagrange multiplier $\Phi: [0, h] \times \mathbb{R}^d \to \mathbb{R}$, leading to the following unconstrained optimization problem for the functional $\mathcal{L}$ 
\begin{align}
\mathcal{L}(q, \rho, v, \Phi) &= \int_0^h \int_{\mathbb{R}^d} \frac{1}{2} \|v(t, x)\|_2^2 \rho(t, x) \, dx \, dt + \int_{\mathbb{R}^d} V(x) q(x) \, dx \\
&\quad + \int_0^h \int_{\mathbb{R}^d} \Phi(t, x) \left( \partial_t \rho + \nabla \cdot (\rho v) - \beta^{-1} \Delta \rho \right) \, dx \, dt. \notag
\end{align}

The first-order optimality conditions yield a system of coupled PDEs comprising a forward Fokker–Planck equation and a backward Hamilton–Jacobi equation
\begin{align}
\label{regu_PDE}
\begin{cases}
\partial_t \rho + \nabla \cdot (\rho \nabla \Phi) = \beta^{-1} \Delta \rho, \\
\partial_t \Phi + \frac{1}{2} \|\nabla \Phi\|_2^2 = -\beta^{-1} \Delta \Phi, \\
\rho(0, x) = \rho_k(x), \quad \Phi(h, x) = -V(x)\,.
\end{cases}
\end{align}

Recalling the Hopf–Cole transformation, we define
$$
\eta(t, x) = \exp\left( \beta\frac{ \Phi(t, x)}{2} \right), \quad \eta(t, x) \hat{\eta}(t, x) = \rho(t, x) \,,
$$
and the system \eqref{regu_PDE} will be equivalent to a system of heat equations for $\eta$ and $\hat{\eta}$:
$$
\begin{cases}
\partial_t \hat{\eta} = \beta^{-1} \Delta \hat{\eta}, \\
\partial_t \eta = -\beta^{-1} \Delta \eta, \\
\eta(0, x) \hat{\eta}(0, x) = \rho_k(x), \\
\eta(h, x) = \exp\left( -\beta \frac{V(x)}{2} \right).
\end{cases}
$$

Since both $\eta$ and $\hat{\eta}$ satisfy forward and backward heat equations, their solutions can be expressed using the heat kernel. This leads to a closed-form expression for the RWPO operator:
\begin{equation}
\label{rho_T_BRWP}
\WProx_V^{h,\beta}(\rho_k)(x) = \int_{\mathbb{R}^d} \frac{\exp\left[ -\frac{\beta}{2} \left( V(x) + \frac{\|x - y\|_2^2}{2h} \right) \right]}{\int_{\mathbb{R}^d} \exp\left[ -\frac{\beta}{2} \left( V(z) + \frac{\|z - y\|_2^2}{2h} \right) \right] dz} \, \rho_k(y) \, dy =: K_V^{h,\beta} \rho_k(x),
\end{equation}
where the kernel $K_V^{h,\beta}$ is explicitly determined by the potential $V$ and time step size $h$.
} 

From \eqref{regu_PDE}, we observe that since \(\Phi(h,x) = -V(x)\) and \(\rho\) satisfies a Fokker–Planck equation with drift vector field \(\nabla \Phi\), the solution of RWPO approximates the evolution of the Fokker–Planck equation \eqref{def_FP} when \(h\) is small. Furthermore, \cite{han2024convergence} justifies that \(K_V^{h,\beta} \rho_k\) approximates \(\rho_{k+1}\) with an error of order \(\mathcal{O}(h^2)\) when \(V\) is smooth and convex. A justification will also be provided in section \ref{sec_analysis}. 
In summary, we use the kernel formula \eqref{rho_T_BRWP} to approximate the evolution of the Fokker–Planck equation \eqref{def_FP} with \(V = f + g\), which further approximates the score function required in \eqref{def_semi_particle}. 

\subsection{Derivation of the algorithm} 
\label{sec_splitting_derivation}
\blue{
Now, we consider sampling from a composite distribution of the form $\rho^* = \frac{1}{Z}\exp(-\beta(f+g))$. For notational purposes, we write
\begin{equation}
\rho_g^*(x) = \frac{1}{Z_g}\exp(-\beta(g(x)))\,,
\end{equation}
where the normalization constant is 
\[
Z_g = \int_{\mathbb{R}^d}\exp(-\beta g(y))\,dy.
\]

First, we recall that the sampling problem can be interpreted as an optimization problem to minimize the KL divergence, which can be expressed as
\begin{align}
\KL(\rho\|\rho^*) &= \int_{\mathbb{R}^d} \rho \log\frac{\rho}{\rho^*} \, dx = \beta\int_{\mathbb{R}^d} f\rho \, dx + \beta\int_{\mathbb{R}^d} g\rho \, dx + \int_{\mathbb{R}^d} \rho \log \rho \, dx \notag \\
&= \beta(\mathcal{E}_f(\rho) + \mathcal{E}_{g}(\rho) + \mathcal{H}(\rho))\,, \label{def_energy}
\end{align}
where $\mathcal{E}_f$ and $\mathcal{E}_{g}$ denote the potential energies associated with $f$ and $g$, respectively, and $\mathcal{H}$ is the entropy term.

We then split the optimization into two steps: the first applies gradient descent for the smooth term $\mathcal{E}_f(\rho)$, and the second uses an implicit Wasserstein proximal step (JKO scheme) to minimize $\mathcal{E}_{g}(\rho) + \mathcal{H}(\rho)$—a variant of the Passty algorithm \cite{salim_splitting}. This procedure generates a sequence of densities $\{\rho_k\}$ via
\begin{equation}
\label{eqn_Passty_density}
\begin{cases}
    \rho_{k+\frac{1}{2}} &= (I - h\nabla f)_{\#} \rho_k, \\
    \rho_{k+1}& =   \displaystyle \argmin_{\rho \in \mathcal{P}_2(\mathbb{R}^d)} \int g \rho dx + \beta^{-1}\int \rho\log \rho dx + \frac{1}{2h} W_2^2(\rho, \rho_{k+\frac{1}{2}}) \\
    &=\displaystyle \argmin_{\rho \in \mathcal{P}_2(\mathbb{R}^d)}\beta^{-1}\KL(\rho\|\rho_g^*) + \frac{1}{2h} W_2^2(\rho, \rho_{k+\frac{1}{2}}) \,.
\end{cases}
\end{equation}
In this formulation, the first step corresponds to an explicit forward Euler discretization, while the second step is an implicit backward update. Together, these steps give a forward-backward gradient flow in the Wasserstein space.

In terms of particles $\{x^{k+1}\}$, the scheme \eqref{eqn_Passty_density} can be rewritten as
\begin{equation}
\label{eqn_Passty_particle}
\begin{cases}
    x^{k+\frac{1}{2}} = x^k - h\nabla f(x^k), \\
    x^{k+1} = x^{k+\frac{1}{2}} - h\nabla g_h(x^{k+\frac{1}{2}}) - h\beta^{-1}\nabla \log \rho_{k+1}(x^{k+1})\,,
\end{cases}
\end{equation}
where the second step follows from the first-order optimality condition.

However, the second step in this scheme remains implicit, as it involves both $x^{k+1}$ and $\rho_{k+1}$ on the right-hand side. To address this, we consider a semi-implicit version by replacing $x^{k+1}$ with $x^{k+\frac{1}{2}}$, $g$ with $g_h$, and approximating $\rho_{k+1}$ by $K_g^{h,\beta} \rho_{k+\frac{1}{2}}$, resulting in a closed-form update:
\begin{equation}
\label{semi_implicit_g}
x^{k+1} = x^{k+\frac{1}{2}} - h\nabla g_h(x^{k+\frac{1}{2}}) - h\beta^{-1}\nabla \log K_g^{h,\beta} \rho_k(x^{k+\frac{1}{2}})\,.
\end{equation}

We remark that although this update rule is explicit in evaluation, it still constitutes a semi-implicit scheme for two reasons:
\begin{itemize}
    \item The score function is evaluated approximately at $\rho_{k+1}$, as the RWPO provides an approximation to the JKO step.
    \item The explicit gradient descent on the Moreau envelope $g_h$ can also be interpreted as an implicit gradient descent on the original non-smooth function $g$.
\end{itemize}

We will demonstrate the convergence of this splitting scheme under smoothness assumptions in Section~\ref{sec_analysis}.
}

\subsection{Algorithm}

To summarize the derivation in the previous section, the iterative formula for particles $\{x^{k+1}\}$ at the $k+1$ iteration is expressed as
\begin{equation}
\label{split_scheme}
\begin{cases}
    x^{k+\frac{1}{2}} =x^k - h \nabla f(x^k)\,, \\
    x^{k+1} = \operatorname{prox}_g^h( x^{k+\frac{1}{2}}) - h\beta^{-1}\nabla \log K_g^{h,\beta}\rho_{k+\frac{1}{2}}( x^{k+\frac{1}{2}})\,.
\end{cases}
\end{equation}

Next, we derive an explicit and computationally efficient formula for the second step in \eqref{split_scheme}. When \( g(x) = \lambda \|x\|_1 \), the proximal operator is given by the shrinkage (soft-thresholding) operator
\[
  S_{\lambda h}(x) := \prox_{\lambda\|x\|_1}^h(x) = \text{sign}(x) \max\{|x| - \lambda h, 0\}\,.
\]
We then simplify the expression for \( K_{g}^{h,\beta}\rho_{k+\frac{1}{2}} \) defined in \eqref{rho_T_BRWP}.

\blue{We recall the Laplace‑method approximation for a Lipschitz function $\phi$ admitting a unique minimizer (see Lemma C.1 or \cite{tibshirani2025laplace}):
\begin{equation}
\lim_{h\rightarrow 0 }\left|\frac{\int z \exp\left(-\frac{\phi(z)}{h}\right)dz}{\int   \exp\left(-\frac{\phi(z)}{h}\right)dz} -  z^* \right| = 0\,,
\end{equation}
where $z^* = \argmin_{z} \phi(z)$ and given that all integrals are finite. Taking
\[
\phi_y(z) = \frac{\beta}{2} \left( h\lambda \|z\|_1 + \frac{\|z - y\|_2^2}{2} \right),
\]
and noting that $z^* = S_{\lambda h}(y)$, we have
\begin{align*}
&\lim_{h\rightarrow 0 }\left|-\frac{2}{\beta}\frac{d}{dy}\left[\ln\int \exp\left(-\frac{\phi_y(z)}{h} \right)dz\right] -\frac{y-S_{\lambda h}(y)}{h}\right| \\
=&\lim_{h\rightarrow 0 }\left|- \frac{2}{\beta}\frac{d}{dy}\left[\ln\int  \exp\left(-\frac{\phi_y(z)}{h} \right)dz - \frac{\inf_z\phi_y(z)}{h}\right] \right|  = 0\,,
\end{align*}
where the derivative of the Moreau envelope can be computed with the implicit differentiation or explicit derivative of the shrinkage operator. 
Taking the exponential on both sides, we arrive at an approximation to the denominator in \eqref{rho_T_BRWP}
\begin{align}
 \label{Laplace_L1}
&0 = \lim_{h\rightarrow 0}\left|C\int\exp\left(-\frac{\phi_y(z)}{h}\right)dz -\exp\left(-\frac{\phi_y(S_{\lambda h}(y))}{h} \right)\right| \\
= &\lim_{h \rightarrow 0} \left|C\int_{\mathbb{R}^d} \exp\left[-\frac{\beta}{2} \left( \lambda \|z\|_1 + \frac{\|z - y\|_2^2}{2h} \right) \right] dz - \exp\left[ -\frac{\beta}{2} \left( \lambda \|S_{\lambda h}(y)\|_1 + \frac{\| S_{\lambda h}(y) - y \|_2^2}{2h} \right) \right] \right| ,\notag
\end{align}
for some constant $C=C(h,\beta)$ that does not depend on $y$.}



For the numerator in \( K_{g}^{h,\beta} \rho_{k+\frac{1}{2}} \), we assume \( \rho_{k+\frac{1}{2}}(x) \) is given by kernel density estimation based on delta measures at current particle locations
\[
\rho_{k+\frac{1}{2}}(x) = \frac{1}{N} \sum_{j=1}^N \delta_{x_j^{k+\frac{1}{2}}}(x)\,.
\]
Substituting this into \eqref{rho_T_BRWP} and using \eqref{Laplace_L1}, the approximated density at time \( t_{k+1} \) becomes
\begin{align}
\label{rhok_kernel}
\lim_{h \rightarrow 0} \left| K_g^{h,\beta} \rho_{k+\frac{1}{2}}(x)
-  C\exp\left( -\frac{\beta}{2} \lambda \|x\|_1 \right)  \sum_{j=1}^N \exp\left[ U(x, x_j^{k+\frac{1}{2}}) \right] \right| = 0\,,
\end{align}
where the function \( U \) is given by
\begin{equation}
\label{def_U}
U(x, x_j^{k+\frac{1}{2}}) := -\frac{\beta}{2} \left( \frac{ \|x - x_j^{k+\frac{1}{2}}\|_2^2 - \| S_{\lambda h}(x_j^{k+\frac{1}{2}}) - x_j^k \|_2^2 }{2h} - \lambda \| S_{\lambda h}(x_j^{k+\frac{1}{2}}) \|_1 \right).
\end{equation}

\blue{Recalling that \( \nabla \log K_g^{h,\beta} \rho_{k+\frac{1}{2}} = \nabla K_g^{h,\beta} \rho_{k+\frac{1}{2}} / K_g^{h,\beta} \rho_{k+\frac{1}{2}} \), and
\begin{equation}
\label{Subdifferential}
\frac{x - S_{\lambda h}(x)}{h} \in \partial(\lambda \|x\|_1)\,,
\end{equation}
where \( \partial g(x) \) denotes the subdifferential of \( g \) at \( x \)}, we conclude that as \( h \to 0 \),
\begin{align}
\label{approx_score}
\lim_{h \rightarrow 0} \left| \nabla \log K_g^{h,\beta} \rho_{k+\frac{1}{2}}(x)
+ \frac{\beta}{2} \left( \frac{x - S_{\lambda h}(x)}{h} + \frac{ \sum_{j=1}^N (x - x_j^{k+\frac{1}{2}}) \exp(U(x, x_j^{k+\frac{1}{2}})) }{ h \sum_{j=1}^N \exp(U(x, x_j^{k+\frac{1}{2}})) } \right) \right| = 0\,.
\end{align}

To further simplify the notation, we define the matrix operator $A_{i,j}$ and the normalized version $M_{i,j}$  as
\begin{equation}
\label{def_Aij_Mij}
    A_{i,j} =  \exp(U(x_i^{k+\frac{1}{2}},x_j^{k+\frac{1}{2}}))\,,\qquad M_{i,j} = \frac{A_{i,j}}{\sum_{j=1}^N A_{i,j}}\,.
\end{equation}
With this notation, the second step of the iterative scheme \eqref{split_scheme} can be rewritten as  
\begin{equation}
\label{xk_iteration_L1}
    x_i^{k+1} = x_i^{k+\frac{1}{2}} + \frac{1}{2} \left( S_{\lambda h}(x_i^{k+\frac{1}{2}}) - \sum_{j=1}^N M_{i,j} x^{k+\frac{1}{2}}_j\right)\,.
\end{equation}

The above derivation leads to a deterministic sampling algorithm for the composite density function  
\[
\rho^*(x) = \frac{1}{Z} \exp\big(-\beta (f(x) + \lambda\|x\|_1)\big)\,,
\]  
which is described below.

\begin{algorithm}[H]
\caption{Splitting Regularized Wasserstein Proximal Algorithm (BRWP-splitting)}
\begin{algorithmic}[1]
\Require Initial particles $\{ x^0_i\}_{i=1}^N$, step size $h$.
\For{iteration $k = 1,2,\dots$ and each particle $i = 1, \dots, N$}
\State \textbf{Step 1:} Compute the gradient descent with respect to smooth function $f$:
\[
x^{k+\frac{1}{2}}_i = x_i^k - h \nabla f(x_i^k)\,.
\]
\State \textbf{Step 2:} Perform the proximal update on $g$ with the score function
\[
    x_i^{k+1} = x^{k+\frac{1}{2}}_i  + \frac{1}{2} \left( S_{\lambda h}(x^{k+\frac{1}{2}}_i ) - \sum_{j=1}^N M_{i,j} x^{k+\frac{1}{2}}_j\right).
\]
Here, \( M_{i,j} \) is defined as in \eqref{def_Aij_Mij}, replacing \( x^k \) with \( x^{k+\frac{1}{2}} \).
\EndFor
\end{algorithmic}
\label{Alg:splitting}
\end{algorithm}


\blue{For a more general target density function $\rho^*$ as in \eqref{def_rho} containing a nonsmooth function \( g \), the Step 2 in Algorithm \ref{Alg:splitting} is replaced by
\begin{equation}
\label{particle_scheme_general}
    x_i^{k+1} = x_i^{k+\frac{1}{2}} + \frac{1}{2} \left(\prox_g^h(x_i^{k+\frac{1}{2}}) - \sum_{j=1}^N M_{i,j} x^{k+\frac{1}{2}}_j\right),
\end{equation}
where  
\[
A_{i,j} = \exp\left[-\frac{\beta}{2} \left(\frac{\|x_i^{k+\frac{1}{2}} - x_j^{k+\frac{1}{2}}\|_2^2-\|\prox_g^h(x_j^{k+\frac{1}{2}}) - x_j^{k+\frac{1}{2}}\|_2^2}{2h} -g(\prox_g^h(x_j^{k+\frac{1}{2}}))\right) \right]\,,
\]
and $M_{i,j}$ is defined as in \eqref{def_Aij_Mij}. 
Intuitively, looking at \eqref{particle_scheme_general}, the term $1/2 \prox_g^h(x_i^{k+\frac{1}{2}})$ corresponds to a half-step of gradient descent depending on $x_i^{k+\frac{1}{2}}$.}
The first exponent $\|x_i^{k+\frac{1}{2}}-x_j^{k+\frac{1}{2}}\|_2^2$ in \( A_{i,j} \) induces diffusion as a heat kernel, while the last exponent involves $g$, which performs the second half-step of gradient descent via a weighted average of \( x_j^{k+\frac{1}{2}} \). 
Hence, the proposed interacting particle dynamics ensures that the set of points will concentrate in a high-probability region of the target density and will not collapse to the local minimum of the log-density $f+g$.

\subsection{Different choices of kernels for particle interaction}

In this section, we explore alternative formulations of the matrix operator \( M_{i,j} \), previously defined in \eqref{def_Aij_Mij}, based on different density approximations of $\rho_k$ from particles. These alternatives may lead to improved numerical performance in high-dimensional sampling problems. To simplify notation.

\begin{prop}
\label{prop_gaussian_mij}
Suppose the density function $\rho$ is written as a summation of Gaussian kernels  as 
\[
\rho (x) = \frac{1}{N(2\pi\sigma^2)^{d/2}} \sum_{j=1}^{N} \exp\left(-\frac{\|x-x_j \|_2^2}{2\sigma^2}\right)\,,
\]
with bandwidth $\sigma>0$.
Then, for the particle update scheme given by \eqref{xk_iteration_L1}, let $c = 2h/(\sigma^2\beta)$ and $x_{i,\ell} $ be the $\ell$-th component of the particle $x_i$, the matrix operator \( A_{i,j} \) and \( M_{i,j}\) will be
\begin{equation}
    A_{i,j} = \exp\left(-\frac{\|x_j\|_2^2}{2\sigma^2}\right)\prod_{\ell=1}^d\left[S_1(x_{i,\ell},x_{j,\ell}) + S_2(x_{i,\ell}, x_{j,\ell}) + S_3(x_{i,\ell},x_{j,\ell})\right]\,,
\end{equation}
\begin{equation}
M_{i,j} =  \frac{A_{i,j}}{\sum_{j=1}^{N} \left\{ \exp\left(-\frac{\|x_j\|_2^2}{2\sigma^2}\right)\prod_{\ell=1}^d\left[T_1(x_{i,\ell},x_{j,\ell}) + T_2(x_{i,\ell},x_{j,\ell}) + T_3x_{i,\ell},(x_{j,\ell})\right]\right\}}\,,
\end{equation}
where the terms \( T_1 , T_2 , T_3  \) and \( S_1 , S_2 , S_3 \) are given by
{\footnotesize
\[
\begin{cases}
    T_1(x,z) = \sqrt{\frac{4h}{\beta(1+c)}}\int_{\sqrt{\frac{\beta(1+c)}{4h}}\left[\lambda h - \frac{x+cz+\lambda h}{1+c}\right]}^{\infty}\exp(-y^2)dy\exp\left(-\frac{\beta}{4h}\left(\lambda^2h^2-\frac{(x+cz+\lambda h)^2}{1+c}\right)\right),\\
    T_2(x,z)  = \sqrt{\frac{4h}{\beta(1+c)}}\int^{\sqrt{\frac{\beta(1+c)}{4h}}\left[-\lambda h - \frac{x+cz-\lambda h}{(1+c)}\right]}_{-\infty}\exp(-y^2)dy\exp\left(-\frac{\beta}{4h}\left(\lambda^2h^2-\frac{(x+cz-\lambda h)^2}{1+c}\right)\right),\\
    T_3(x,z) = \sqrt{\frac{4h}{c\beta }}\int_{\sqrt{\frac{c\beta}{4h}}\left[-\lambda h-\frac{(x+cz)}{c}\right]}^{\sqrt{\frac{c\beta}{4h}}\left[\lambda h-\frac{(x+cz)}{c}\right]} \exp(-y^2) dy  \exp\left(\frac{\beta}{4h} \frac{(x+cz)^2}{c}  \right),\\
    S_1(x,z) = \frac{\beta}{2}\frac{(x+cz+\lambda h)}{h(1+c)}T_1(x,z)+ \frac{1}{1+c}\exp\left(-\frac{\beta(1+c)}{4h}\left(\lambda h - \frac{x+cz+\lambda h}{1+c}\right)^2 \right)\exp\left(-\frac{\beta}{4h}\left(\lambda^2h^2-\frac{(x+cz+\lambda h)^2}{1+c}\right)\right),\\
     S_2(x,z) = \frac{\beta}{2}\frac{(x+cz-\lambda h)}{h(1+c)}T_2(x,z) -\frac{1}{1+c}\exp\left(-\frac{\beta(1+c)}{4h}\left(-\lambda h - \frac{x+cz-\lambda h}{1+c}\right)^2 \right)\exp\left(-\frac{\beta}{4h}\left(\lambda^2h^2-\frac{(x+cz-\lambda h)^2}{1+c}\right)\right),\\
     S_3(x,z) = \frac{\beta}{2} \frac{(x+cz)}{hc}T_3(x,z)-\frac{1}{c}\left[\exp\left(-\frac{c\beta}{4h}\left(\lambda h - \frac{(x+cz)}{c}\right)^2 \right)  - \exp\left(-\frac{c\beta}{4h}\left(-\lambda h - \frac{(x+cz)}{c}\right)^2 \right)\right]  \exp\left(\frac{\beta(x+cz)^2}{4hc} \right),
\end{cases}
\]
}

Here, the integral of $\exp(-y^2)$ can be obtained by the value of the error function 
\[
\erf(z) = \frac{2}{\sqrt{\pi}} \int_0^z \exp(-y^2) \, dy\,.
\]
\end{prop}
 The derivation of Proposition \ref{prop_gaussian_mij} can be found in supplementary material \ref{sec_appendix_derivation}.
The Gaussian kernel used in \cite{TT_BRWP} has been applied to eliminate asymptotic bias in the discretization of the probability flow ODE when the target distribution is Gaussian. Moreover, Gaussian kernels with adaptively computed bandwidths based on particle variance are also helpful for approximating density functions in high dimensions. For further discussion in this direction, see \cite{wang2022accelerated}.  

Next, by comparing the results in Proposition \ref{prop_gaussian_mij} with the expression in \eqref{def_Aij_Mij}, we observe that the matrix operator $M_{i,j}$ simplifies significantly when the kernel is approximated by delta measures. However, in high-dimensional settings, kernel density estimation with delta measures suffers from the curse of dimensionality, as the number of particles required to maintain a given level of accuracy grows exponentially \cite{emprical_rate}.
To address this issue, we propose an alternative and heuristic method for efficiently approximating the score function while maintaining its representation as a sum of delta measures. Specifically, we assume the density function is approximated bt an auxiliary set of points: 
\begin{equation}
\label{def_rho_k_aux}
    \rho(x) = \frac{1}{N^d}\sum_{j_1,\cdots,j_d=1}^{N}\delta_{\tilde{x}_{j_1,\cdots,j_d}}(x),\quad \tilde{x}_{j} = \tilde{x}_{j_1,\cdots,j_d} = [x_{j_1}(1),\cdots,x_{j_d}(d)]^T.
\end{equation}
Here, \( \rho_k \) can be regarded as approximated by a separable density function. Due to the separability of the \( L_1 \) norm and the shrinkage operator, the proposed matrix operator takes the following form, with its derivation provided in the supplementary material \ref{sec_appendix_derivation}.

\begin{prop}
\label{prop_delta_aux_mij}
Suppose the density function $\rho$ is written as equation \eqref{def_rho_k_aux}, then for the particle update scheme given by \eqref{xk_iteration_L1}, the operator $M_{i,j}$ will be
\begin{equation}
(M_{i,j})_{\ell} = \frac{\exp\left[-\frac{\beta}{2}\left(\frac{(x_{i,\ell}-x_{j,\ell})^2-( S_{\lambda h}(x_{j,\ell})-x_{j,\ell})^2}{2h}-\lambda| S_{\lambda h}(x_{j,\ell})|\right) \right]}{\sum_{j}\exp\left[-\frac{\beta}{2}\left(\frac{(x_{i,\ell}-x_{j,\ell})^2-( S_{\lambda h}(x_{j,\ell})-x_{j,\ell})^2}{2h}-\lambda| S_{\lambda h}(x_{j,\ell})|\right) \right]}\,,
\end{equation}
for $\ell = 1,\cdots, d$ where $x_{i,\ell}$ denotes the $\ell$-th component of the particle $x_i$.
The Step 2 in Algorithm \ref{Alg:splitting} now becomes
\begin{equation}
x_{i}^h = x_i + \frac{1}{2}\left( S_{\lambda h}(x_i)-\sum_{j=1}^N M_{i,j}\cdot x_j \right)\,,
\end{equation}
where $x_i^h$ is the particle at the next time point $h$.
\end{prop}

Our numerical experiments in Section \ref{sec_NE} show that the kernel in Proposition \ref{prop_delta_aux_mij} usually has faster convergence and more accurate estimation of the model variance than the kernel in \eqref{def_Aij_Mij} in high-dimensional sampling problems. 

Moreover, we remark that for more general log-density functions \( g \) and other choices of kernels used to estimate \( \rho_k \) that are not separable, tensor train approaches can be employed. Once the density at time \( t_k \) and the target density \( \rho^* \) are approximated in tensor train format, an analog of Algorithm \eqref{Alg:splitting} remains computationally efficient. For further details, see \cite{TT_BRWP}.

\section{Convergence Analysis}
\label{sec_analysis}

In this section, we analyze the convergence of the proposed Algorithm \ref{Alg:splitting} for sampling from the target distribution. We will separate the verification into two cases. 

\blue{
Firstly, we verify the scenario in which the gradient of the Moreau envelope, \( \nabla g_h \), is Lipschitz continuous with constant \( L_{g_h} \), where \( L_{g_h} \) is either independent of \( h \) or of the order \( \mathcal{O}(h^{-\gamma}) \) for some \( \gamma < 1 \). This setting applies, for example, when \( g_h \) is replaced by \( g_{\epsilon} \) for a  regularization parameter \( \epsilon \) independent of $h$ or $\epsilon = \sqrt{h}$ in Algorithm~\ref{Alg:splitting}, or when \( g \) itself is $L$ smooth.
In this regime, we establish exponential convergence under the assumption that the target distribution satisfies a Poincaré inequality following the idea of proof in \cite{ula_wibisono}. This part serves to illustrate the effectiveness and theoretical soundness of the proposed algorithm under smooth conditions.

Secondly, we address the nonsmooth case where \( g \) may be nonsmooth e.g., \( g(x) = \|x\|_1 \) and the gradient of its Moreau envelope \( \nabla g_h \) can be \( 1/h \)-Lipschitz continuous. In this setting, we establish convergence using tools from convex optimization in Wasserstein space as in \cite{salim_splitting}. We establish the convergence with respect to the Wasserstein-2 distance. }

\subsection{Convergence under smooth assumption}
\label{sec_analysis_smooth}
In this section, we let $\rho_h^* = 1/Z \exp(-\beta (f + g_h))$ where $Z$ is the normalization constant. We also assume that the following assumptions hold
\begin{enumerate}[label=A.\arabic*:]
    \item The function \( f \) is convex and \( L_f\) smooth, meaning that \( \nabla f \) is \( L_f\) Lipschitz continuous. The function $g$ is convex.
    \item The function \(g_h\) is \(L_{g_h}\) smooth.
    \item $\rho_h^*$ satisfies the Poincaré inequality with constant $\alpha_d>0$, i.e., for any bounded smooth function $\psi$,  
    \[
\int_{\mathbb{R}^d} \psi^2 \rho_h^* \,dx - \left(\int_{\mathbb{R}^d} \psi \rho_h^* \,dx\right)^2 \leq \alpha_d \int_{\mathbb{R}^d} \|\nabla \psi\|^2\rho_h^* \,dx\,  .
\]
    \item The score function at time $t$, i.e., $\nabla \log\rho_t$ where $\rho_t$ satisfies the Fokker Plank equation at time $t$ is  \( \beta L_{\rho} \) Lipshitz continuous in $\mathbb{R}^d$. 
\end{enumerate}

We remark that the first and the second conditions ensure the proximal operator of $g$ is single-valued, and the Hessian of $g_h$ is bounded, which allows us to derive the asymptotic expression of the kernel formula \eqref{rho_T_BRWP}.
Regarding the Poincaré inequality in the third assumption, we note that it follows from both the log-Sobolev inequality and the Talagrand inequality. Furthermore, it remains valid even in cases where the log-Sobolev inequality does not apply, such as when \( g \) has a tail of the form \( \|x\|_1 \). Moreover, both the log-Sobolev and Poincaré inequalities are special cases of the Latała–Oleszkiewicz inequality for \( \alpha = 2 \) and \( \alpha = 1 \), respectively. These inequalities characterize concentration properties for densities of the form \( \exp(-\|x\|^{\alpha}) \), as discussed in \cite{latala2000between}. Finally, the last condition is an assumption to ensure the regularity of the score function at time $t$ which is employed in our iteration.

\begin{lemma}
\label{lemma_kernel}
For the approximation to the score function based on the kernel formula \eqref{rho_T_BRWP}, when $g_h$ is $L_{g_h}$-smooth and $h<\min\{1/(L_{g_h}d^2),1/(L_{\rho}d^2)\}$, we have
\begin{equation}
\label{eqn:K_g_FPK}
\nabla\log K_g^{h,\beta}\rho_{0}(x) =  \nabla\log\rho(x,h) + \mathcal{O}(h^2)
\end{equation}
where \( \rho(x,t) \) satisfies the Fokker–Planck equation with the initial $\rho_0$:
\[
\frac{\partial \rho}{\partial t}  = \nabla\cdot(\rho\nabla g_h) + \beta^{-1}\Delta \rho, \quad \rho(x,0) = \rho_{0}(x)\,,
\]
and the coefficient for $\mathcal{O}(h^2)$ terms depends on $L_{g_h}$ and $L_{\rho}$ given in assumptions A.2 and A.4.
 \end{lemma}

\blue{
For the case where $g_h$ is smooth, the analysis has been presented in \cite{han2024convergence}. We only illustrate the validity of the formula through the heat kernel expansion.

\noindent\textbf{Sketch of the proof:}
For smooth $g_h$, let $\rho_g = 1/Z_g \exp(-\beta g)$ and $G_h(x) = (4\pi h)^{-d/2}\exp(-\|x\|^2_2/(4h))$, using Taylor expansion and properties of the heat kernel, we have
\begin{align*}
&K_{g}^h\rho_{0} =  \rho_g^{1/2}\left(\frac{\rho_0}{\rho_g^{1/2}\ast G_{\beta^{-1}h}}\ast G_{\beta^{-1}h}\right) = \rho_g^{1/2}\left(\frac{\rho_0}{\rho_g^{1/2}+h\beta^{-1}\Delta \rho_g^{1/2}}\ast G_{\beta^{-1}h}\right)+\mathcal{O}(h^2)\\
=&\rho_g^{1/2}\left[\left(\rho_g^{-1/2}\rho_0(1-h\beta^{-1}\rho_g^{-1/2} \Delta\rho_g^{1/2})\right)\ast G_{\beta^{-1}h}\right] + \mathcal{O}(h^2)\\
=& \rho_k - h\beta^{-1}\rho_0\rho_g^{-1/2}\Delta\rho_g^{1/2} + h\beta^{-1}\rho_g^{1/2}\Delta\left[\rho_g^{-1/2}\rho_0\right] + \mathcal{O}(h^2)\\
=& \rho_0 + h\beta^{-1}\Delta\rho_0 + h\beta^{-1}\left(-\rho_0 \rho_g^{-1/2}\Delta\rho_g^{1/2} + \rho_0\rho_g^{1/2}\Delta\rho_g^{-1/2} + 2 \rho_g^{1/2}\nabla\rho_g^{-1/2}\cdot \nabla \rho_0 \right) + \mathcal{O}(h^2)\\
=& \rho_0+ h\beta^{-1}\Delta\rho_0 + h \nabla \cdot (\rho_0 \nabla g_h) + \mathcal{O}(h^2)\,. 
\end{align*}

Hence, by recalling that $\nabla\log K_g^{h,\beta}\rho_{0} = \nabla K_g^{h,\beta}\rho_0/K_g^{h,\beta}\rho_0$, we can arrive the desired results. 
}

Recalling the definition of the Moreau envelope of \( g \) in \eqref{def_moreau}, we next recall some key properties of the Moreau envelope \cite{durmus2018efficient}.
\begin{enumerate}
    \item For any \( x \in \mathbb{R}^d \),
    \[
    0 \leq g(x) - g_h(x) \leq L_g^2 h\,.
    \]
    \item \( g_h \) is convex, and the function \( \frac{1}{Z_{g_h}} \exp(-\beta g_h) \) defines a valid probability density function, where  
    \[
    Z_{g_h} = \int_{\mathbb{R}^d} \exp(-\beta g_h(y)) \, dy\,.
    \]
\end{enumerate}
The last statement follows from 
\[
\nabla g_h(x) \in \partial g(\prox_g^h(x))\,.
\]

Moreover, recalling the proposed particle evolution scheme in \eqref{split_scheme}, the first step consists of a gradient descent step with respect to \( f \). By applying the change of variable formula for the probability density function, we obtain
\[
\rho_{k+\frac{1}{2}} = \rho_k + h \nabla\cdot(\rho_k\nabla f) + \mathcal{O}(h^2)\,.
\]
Consequently, after applying the kernel \( K_g^{h,\beta} \) and using the result in Lemma \ref{lemma_kernel}, we have the approximation formula
\begin{equation}
\label{density_k}
K_g^{h,\beta}\rho_{k+\frac{1}{2}} = \rho_k + h\nabla\cdot(\rho_k\nabla (f+g_h)) + h\beta^{-1}\Delta\rho_{k} + \mathcal{O}(h^2)\,.
\end{equation}
Thus, the density function \( K_g^{h,\beta}\rho_{k+\frac{1}{2}} \) obtained from the kernel formula \eqref{rho_T_BRWP} provides a first-order approximation to the evolution of the Fokker–Planck equation with drift term \( \nabla(f+ g_h) \).

Next, the iterative sampling scheme in \eqref{split_scheme} can be rewritten more compactly as
\begin{align}
\label{appendix_particle}
x^{k+1} =& x^k - h\nabla f(x^k) -h\nabla g_h(x^k-h\nabla f(x^k)) - h\beta^{-1}\nabla\log K_g^{h,\beta}\rho_{k+\frac{1}{2}}(x^k-h\nabla f(x^k)) \,.
\end{align}

Our convergence analysis will examine the convergence of the density \(\rho_k\) to \(\rho^*\) in terms of the Rényi divergence \( R_q \) for \( q \in [2, \infty) \). The R\'enyi divergence is defined as
\[
R_q(\mu\|\nu) := \frac{1}{q-1} \log\left(F_q(\mu\|\nu)\right), \, \text{ where }  \,F_q(\mu\|\nu) = \int_{\mathbb{R}^d} \frac{\mu^q}{\nu^{q-1}} \,dx\,.
\]
We define the Rényi information \( G_q(\mu\|\nu) \) as the time derivative of \( F_q(\mu\|\nu) \)
\begin{equation}
\label{def_Renyi_infor}
G_q(\mu\|\nu) = \int_{\mathbb{R}^d} \left(\frac{\mu}{\nu} \right)^q \left\|\nabla\log\frac{\mu}{\nu}\right\|_2^2 \nu \,dx\,.
\end{equation}
A key consequence of the Poincaré inequality is the following relationship regarding the time derivative of the R\'enyi divergence along the Langevin dynamics.

\begin{lemma}[\cite{ula_wibisono}]
\label{lemma_information}
Suppose \( \rho_h^* \) satisfies the Poincaré inequality with constant \( \alpha_d > 0 \). Then, for any \( q \geq 2 \), we have
\begin{equation}
\frac{G_q(\rho\|\rho_h^*)}{F_q(\rho\|\rho_h^*)} \geq \frac{4\alpha_d}{q^2} \left[1 - \exp(-R_q(\rho\|\rho_h^*))\right]\,.
\end{equation}
\end{lemma}

By employing the interpolation argument and establishing bounds for the discretization error, we can prove the convergence of the proposed sampling scheme to the target density as follows.
\begin{theorem}
\label{thm_regu_density}
Let $x^0 \sim \rho_0$ be initial particles and $L = L_f + L_{g_h} + L_\rho$. When $h \leq (\sqrt{2}-1)/L)$ and satisfies the condition in Lemma \ref{lemma_kernel}, we have the following convergence of Algorithm \ref{Alg:splitting} with respect to the Rényi divergence.
\begin{enumerate}
\blue{
    \item For the convergence towards the regularized target density \( \rho_h^* \):
    \begin{equation}
    \begin{aligned}
        R_q(\rho_k\|\rho_h^*) \leq
        \begin{cases}
            R_q(\rho_0\|\rho_h^*) 
            - kh \left( 
                \dfrac{\alpha_d}{q} 
                \left( 1 - \dfrac{2L^2 h^2}{(1 - hL)^2} \right)
                - q L_f^2 (L + L_f)^2 h^2 d 
            \right) + \mathcal{O}(h^3) \\
            \qquad \text{if } R_q(\rho_0\|\rho_h^*) \geq 1; \\[1.5ex]
            
            R_q(\rho_0\|\rho_h^*) 
            \exp\left[
                - kh \dfrac{\alpha_d}{q}
                \left( 1 - \dfrac{2L^2 h^2}{(1 - hL)^2} \right)
            \right] \\
            \quad + \dfrac{q^2 L_f^2 (L + L_f)^2 h^2 d}{\alpha_d} 
            + \mathcal{O}(h^3) \\
            \qquad \text{if } R_q(\rho_0\|\rho_h^*) < 1.
        \end{cases}
    \end{aligned}
    \end{equation}
}
    
    \item For the convergence towards the target density \( \rho^* \):
    \begin{equation}
    \begin{aligned}
        R_q(\rho_k\|\rho^*) \leq
        \begin{cases}
            R_{2q-1}(\rho_0\|\rho_h^*)  
            - t_k \Bigg( 
                \dfrac{\alpha_d}{2q-1}
                \left( 1 - \dfrac{2L^2 h^2}{(1 - hL)^2} \right) \\
                \quad - (2q - 1) L_f^2 (L + L_f)^2 h^2 d 
            \Bigg) 
            + c(q) L_g^2 h + \mathcal{O}(h^3) \\
            \qquad \text{if } R_q(\rho_0\|\rho_h^*) \geq 1; \\[1.5ex]
            
            R_{2q-1}(\rho_0\|\rho_h^*) 
            \exp\left[ 
                - t_k \dfrac{\alpha_d}{2q - 1}
                \left( 1 - \dfrac{2L^2 h^2}{(1 - hL)^2} \right)
            \right] \\
            \quad + \dfrac{(2q - 1)^2 L_f^2 (L + L_f)^2 h^2 d}{\alpha_d}
            + c(q) L_g^2 h + \mathcal{O}(h^3) \\
            \qquad \text{if } R_q(\rho_0\|\rho_h^*) < 1.
        \end{cases}
    \end{aligned}
    \end{equation}
    where \( c(q) = \dfrac{q(2q - 1)}{(2q - 1)^2} \).
\end{enumerate}

\end{theorem}
The proof is provided in the supplementary material \ref{appendix_convergence}. We remark that for the convergence to \( \rho_h^* \), when \( R_q(\rho_0\|\rho_h^*) < 1 \), the asymptotic bias induced by the discretization is of order \( \mathcal{O}(h^2) \), which is smaller than that of sampling methods with Brownian motion, where the bias is of order \( \mathcal{O}(h) \).


\subsection{Convergence without smooth assumption}
\label{sec_analysis_nonsmooth}
\blue{In this section, we establish convergence results without relying on the smoothness assumption (A.2), the Poincaré inequality (A.3), or the smoothness of the score function stated in assumption (A.4) used in the previous section.

To proceed, in addition to assumption (A.1), we introduce the following two additional assumptions:

\begin{enumerate}[label=B.\arabic*:]
    \item The nonsmooth points of \( g \) are contained in \( B(0, R) \), and \( g \) is \( L \)-smooth outside \( B(0, R) \).
    \item The functions \( \rho_k(x) \), \( 1/\rho_k(x) \), \( \|\nabla \log \rho_k(x)\|_2 \), \( 1/\rho_{g_h}^{1/2}(x) \), and \( \|\beta \nabla g_h(x)\|_2 \) are all bounded by a constant \( C_R \) within the region \( B(0, 3R) \).
\end{enumerate}

The second assumption is easily verified for many practical nonsmooth distributions. For instance, when \( g(x) = \|x\|_1 \), we have \( \|\nabla g_h\|_2 \leq 1 \), and
\begin{align}
\int \exp\left(-\frac{g_h(x)}{2}\right) dx &= \left[\int_{-h}^h \exp\left(-\frac{x^2}{4h}\right) dx + 2 \int_h^{\infty} \exp\left(-\frac{x}{2} + \frac{h}{4}\right) dx \right]^d \\
&= \left[\int_{-h}^h \exp\left(-\frac{x^2}{4h}\right) dx + 4 \exp\left(-\frac{h}{4} \right) \right]^d = 4 + \mathcal{O}(h)\,,\notag
\end{align}
which implies that \( 1/\rho_{g_h}^{1/2}(y) \leq \exp(3/(2\sqrt{d}) R)/4 + \mathcal{O}(h) \) for \( y \in B(0, 3R) \), and is therefore bounded. We note that for the case of the \( L_1 \) norm, we can choose \( R = 1 \), so that \( C_R \geq \exp(3/(2\sqrt{d}))/4 \), which is independent of \( h \).

Next, the following lemma establishes the approximation accuracy of the kernel formula $K_g^{h,\beta} \rho_k$ introduced in \eqref{rho_T_BRWP} for nonsmooth functions. The proof is based on a quantitative Laplace approximation for Lipschitz functions, as shown in Lemma \ref{Lemma_laplacian} in the Appendix.

\begin{lemma}
\label{lemma_kernel_nonsmooth}
Under Assumptions (B1) and (B2), for $0 < \delta < 1/2$, $h\leq R^{2/(1-2\delta)}$, and for
\begin{equation}
    h \leq \left[ \frac{4}{\beta} \min\left\{ \left|\ln\left(\frac{2C_R^2 \omega_d \exp(\beta/4)}{1 + C_R^2}\right)\right|,\,\left|\ln\left(\frac{(C_R+2RC_R^2) \omega_d \exp(\beta/4)}{1 + C_R^2(1+2C_{g,1})}\right)\right|,\, (1/2 + d - \delta) \right\} \right]^{\frac{1}{2\delta}}\,,
\end{equation}
where $\omega_d$ is the volume of the unit ball in $\mathbb{R}^d$, $C_{g_1}$ is the first moment of $\rho_g$, we have
\begin{equation}
\left\| \nabla \log K_g^{h,\beta}\rho_{k}(x) - \nabla  \log \rho_{k}(x) \right\|_2 \leq C_K h^{1/2 - \delta} \,,
\end{equation}
with constant $C_K = 4\sqrt{d}(C_R^3+2RC_R^4) + 4C_R^4$.
\end{lemma}

We remark that, for the result of the above lemma, the approximation in Lemma~\ref{lemma_kernel} remains valid for points \( x \) that are sufficiently away from the nonsmooth region, for instance \(\|x\|_2 \geq R + h^{1/2 - \delta} \). In this case, the difference between $K_g^{h,\beta}\rho_k$ and the right-hand side of \eqref{eqn:K_g_FPK} is controlled and remains of order \( \exp(-h^{-2\delta}) \), as can be seen in equation~\eqref{eqn:concentration_heat}, due to the concentration of the heat kernel mass within the ball \( B(0, h^{1/2 - \delta}) \).
Therefore, in regions away from the nonsmooth point, the approximation given by \( K_g^{h,\beta} \) can still be reliably used for the evolution of the Fokker-Planck equation.

For notational simplicity, we rewrite the proposed algorithm as follows
\begin{equation}
\label{particle_x_k}
\begin{cases}
x^{k+\frac{1}{3}} = x^k - h \nabla f(x^k), \\
x^{k+\frac{2}{3}} = x^{k+\frac{1}{3}} - h \nabla g_h(x^{k+\frac{1}{3}})\,,\\
x^{k+1} = x^{k+\frac{2}{3}} - h \nabla \log K \rho_{k+\frac{1}{3}}(x^{k+\frac{1}{3}})\,,
\end{cases}
\end{equation}
where \(x^k \sim \rho_k\). This update corresponds exactly to the evolution in the proposed BRWP-splitting Algorithm \ref{Alg:splitting}, with the update separated into three substeps.

We now analyze the convergence of the sequence $\{x^k\}$ by quantifying the one-step energy decay in the scheme \eqref{particle_x_k}. The analysis follows ideas from \cite{salim_splitting,durmus2019analysis} whose proof can be found in Appendix \ref{appendix_nonsmooth_proof}.

\begin{lemma}
\label{lemma_w2_trhok}
Assuming assumptions A.1, B.1, and B.2. Suppose $h$ satisfies the condition in Lemma \ref{lemma_kernel_nonsmooth}, and that \(hL_f < 1\). Then,
\begin{align}
\label{eqn:lemma_w2}
 2h\beta^{-1}\KL(\rho_{k+\frac{2}{3}}\|\rho^*) 
\leq & (1-h\alpha)W_2^2(\rho_k,\rho^*)-W_2^2(\rho_{k+1},\rho^*)\\
&+  h^{3/2-\delta}\beta^{-1}C_K^2 + h^{3/2-\delta}\beta^{-1} W^2_2(\rho_{k+\frac{2}{3}},\rho^*) + \mathcal{O}(h^2)\,.\notag
\end{align}
\end{lemma}

Based on Lemma \ref{lemma_w2_trhok}, by summing up both sides over index $k$, we can have the following result where the first part follows from the convexity of the KL divergence, and the second part follows from the Talagrand inequality where the proof can be found in Appendix \ref{appendix_nonsmooth_proof}.
\begin{theorem}
\label{thm_KL_nonsmooth}
With the same assumption in Lemma \ref{lemma_w2_trhok}, we have the following.
\begin{enumerate}
    \item 
 Define the averaged distribution \(\trho_k\) as
\[
\trho_k = \frac{1}{k+1} \sum_{j=0}^k \rho_j\,.
\]
Then the KL divergence between the averaged distribution and the target distribution will satisfy
\begin{equation}
\KL(\trho_k\|\rho^*)\leq \frac{\alpha\beta}{\alpha\beta -h^{1/2-\delta}}\left[\frac{\beta}{2t_k}W_2^2(\rho_0,\rho^*)+h^{1/2-\delta}C_K^2\right] + \mathcal{O}(h)
\end{equation}
\item 
When we have the additional assumption that $h\leq (4/(\alpha\beta))^{2/(1-2\delta)}$, 
we have 
\begin{equation}
W_2^2(\rho_{k},\rho^*)\leq (1-h\alpha)^{k}W_2^2(\rho_0,\rho^*) + h^{1/2-\delta}\frac{C_K^2}{\alpha\beta}+ \mathcal{O}(h)\,.
\end{equation}
 \end{enumerate}
Here,  $0<\delta < 1/2$, where $C_K$ is given in Lemma \ref{lemma_kernel_nonsmooth}.
\end{theorem}
 }
 
\section{Generalization to Sampling with TV Regularization}
\label{sec_TV_imaging}

An important practical application of \( L_1 \)-norm regularization is its combination with total variation (TV) regularization for image denoising and restoration \cite{TV_1992}. In this context, we consider sampling from the distribution  
\begin{equation}  
\label{def_V_TV}  
\rho^*(u) = \frac{1}{Z} \exp(-V(u))\,,\quad V(u) = \|\phi - Fu\|_2^2 + \lambda \|D u\|_1\,,  
\end{equation}  
where \( u \in \mathbb{R}^d \) represents the image or signal, \( \phi \in \mathbb{R}^m \) is the noisy observation, and \( F \in \mathbb{R}^{d\times m} \) is a known forward operator with \( m \leq d \). The matrix \( D \in \mathbb{R}^{d\times 2d} \) denotes the discretized gradient operator for two-dimensional images.  
This formulation extends naturally to the more general setting where \( V(u) = f(u) + \|Ku\|_1 \) for an arbitrary function \( f \) and a linear operator \( K \). For clarity, we focus on sampling from \eqref{def_V_TV}. Compared to direct optimization of \( V(u) \), sampling-based algorithms provide a means to quantify uncertainty in the recovered image and facilitate Bayesian inference, which will be demonstrated in Section \ref{sec_NE}.

A common approach to sampling from \( \rho^* \) in \eqref{def_V_TV} is to compute the proximal operator of the TV norm using Chambolle's algorithm \cite{chambolle2004algorithm}, as in \cite{durmus2018efficient}. However, this requires solving an optimization problem at each iteration. Instead, we seek a more deterministic method by combining the BRWP-splitting scheme with proximal splitting techniques in convex optimization \cite{review_splitting}.

\blue{Since the proximal operator of the TV norm lacks a closed-form expression, we introduce an auxiliary variable \( p = D u \in \mathbb{R}^{2d} \) and reformulate the log-density as  
\begin{equation}  
V(u, p) = \|\phi - Fu\|_2^2 + \lambda \|p\|_1 + \gamma \|p - D u\|_1\,, 
\end{equation}  
where \( \gamma > 0 \) is a large regularization parameter enforcing \( p \) to be close to \( D u \). This transforms the sampling problem in \( u \) into a simultaneous sampling task over \( u \) and \( p \).  

We now recall the Condat-Vu primal-dual splitting algorithm \cite{condat_Vu} for optimization problems involving the sum of a smooth term and a nonsmooth term with a linear transformation.
Introduce the linear operator \( L = [I, -D]^T \) for notational convenience and note that the convex conjugate of the \( \ell_1 \)-norm is the indicator function \( \delta_{\|\cdot\|_{\infty} \leq 1} \). Inspired by the method proposed in \cite{habring2024subgradient}, we consider the following stochastic variant of the Condat-Vu scheme:
\begin{equation}
\label{PD_scheme_random}
\begin{cases}
    U^{k+1} = U^{k} - h \nabla \|\phi - FU^k\|_2^2 - h\gamma Y^k + \sqrt{2\beta^{-1}h}\zeta^k_{d}\\
    P^{k+1} = \prox_{\lambda \|\cdot \|_1}(P^k - h \gamma (-D)^T Y^k)+ \sqrt{2\beta^{-1}h}\zeta^k_{2d} \\
    Y^{k+1} = \prox_{\delta\|\cdot \|_{\infty}\leq 1}(Y^k + \tau \gamma L (2 [U^{k+1},P^{k+1}]^T - [U^k, P^k]^T))\,, 
\end{cases}
\end{equation}
where \( \zeta^k_d \) is a \( d \)-dimensional standard Gaussian noise term added to the primal update, and \( h, \tau > 0 \) are the primal and dual step sizes, respectively. We remark that for the standard Condat-Vu scheme, the first two updates are written as one substep, while we decompose them into two, as the \( L_1 \)-norm term only depends on \( p \).

Next, we consider the deterministic version of \eqref{PD_scheme_random} with the help of the score function. Replacing the Brownian motion term with the score function yields the deterministic update:
\begin{equation}
\label{PD_scheme_score}  
\begin{cases}
    u^{k+1} = u^{k} - h \nabla \|\phi - Fu^k\|_2^2 - h\gamma y^k - h\beta^{-1}\nabla \log \rho_{k+1}^u(u^k)\\
    p^{k+1} = \prox_{\lambda \|\cdot \|_1}(p^k - h \gamma (-D)^T y^k) - h\beta^{-1}\nabla \log \rho_{k+1}^p(p^k)  \\
    y^{k+1} = \prox_{\delta\|\cdot \|_{\infty}\leq 1}(y^k + \tau \gamma L (2 [u^{k+1},p^{k+1}]^T - [u^k, p^k]^T))\,, 
\end{cases}
\end{equation}
Here, \( u^{k+1} \sim \rho^u_{k+1} \) and \( p^{k+1} \sim \rho^p_{k+1} \).

Finally, as in Section \ref{sec_algo}, we apply the two-step splitting strategy to update the primal variables and apply the RWPO in \eqref{rho_T_BRWP} to approximate $\rho_{k+1}$ to arrive
\begin{align}  
&\begin{cases}
    u^{k+\frac{1}{2}} = u^k - h\gamma y^k, \\
    u^{k+1} = u^{k+\frac{1}{2}} - h\nabla\|\phi - F u^{k+\frac{1}{2}}\|_2^2 - h\beta^{-1} \nabla\log K_{\|\phi - F\cdot\|_2^2}^h \rho^u_{k+\frac{1}{2}}(u^{k+\frac{1}{2}})\,.
\end{cases} \\
&\begin{cases}
    p^{k+\frac{1}{2}} = p^k - h\gamma (-D y^k), \\
    p^{k+1} = S_{\lambda h}(p^{k+\frac{1}{2}}) - h\beta^{-1} \nabla\log K_{\lambda\|\cdot\|_1}^{h,\beta} \rho^p_{k+\frac{1}{2}}(p^{k+\frac{1}{2}})\,.
\end{cases}
\end{align} 

Moreover, the approximated score functions \( \nabla \log K_{\|\phi - F\cdot\|_2^2}^h \) and \( \nabla \log K_{\lambda\|\cdot\|_1}^{h,\beta} \) are defined analogously to \eqref{approx_score}. The proximal operator \( \prox_{\|\phi - F\cdot\|_2^2}^h \) can be computed explicitly as
\begin{equation}
 \label{L_2_prox}   
    \prox_{\|\phi - F\cdot\|_2^2}^h(v) = (I + hF^T F)^{-1} (v + hF^T\phi) = (I - hF^T F)(v + hF^T\phi) + \mathcal{O}(h^2)\,.
\end{equation}
This splitting scheme decomposes the primal update into two sequential steps:  
(i) a gradient descent step involving the dual variable \( y \), and  
(ii) a gradient descent step for the smooth part and a proximal step for the nonsmooth part with explicit score functions.

The full algorithm, incorporating the dual update and primal splitting, is summarized in Algorithm \ref{alg_TV_sampling}. Numerical experiments are presented in Section \ref{sec_NE}.  
}

\begin{algorithm}
\caption{Sampling Algorithm for Posterior Distribution with TV Regularization}
\label{alg_TV_sampling}
\begin{algorithmic}[1]
\Require Initial particles $\{u^0_i, \,p^0_i, \,y^0_i\}_{i=1}^N$, step size $h$, $\tau$, parameters $\gamma, \lambda$
\For{iteration $k = 1,2,\dots$ and each particle $i = 1, \dots, N$}
        \State {Gradient descent for the inner product term:}
        \[
        {u}^{k+\frac{1}{2}}_i = u^k_i + h\gamma D^T y^k_i\,, \quad
        {p}^{k+\frac{1}{2}}_i = p^k_i - h\gamma y^k_i\,.
        \]

        \State {Semi-implicit discretization of the probability flow ODE for the data fitting term:}
        \[
        u^{k+1}_i = {u}^{k+\frac{1}{2}}_i + \frac{1}{2}\left(u_i^{k+\frac{1}{2}}- hF^T(F{u}^{k+\frac{1}{2}}_i - g) -\sum_{j=1}^Nu_j^{k+\frac{1}{2}}M_{i,j}^u \right)\,,
        \]
        where $M^u_{i,j}$ is defined in \eqref{particle_scheme_general} with $g(v) = \|\phi-Fv\|_2^2$,  $\prox_g^h$ given in \eqref{L_2_prox}, and $x^{k+\frac{1}{2}}$ replaced by $u^{k+\frac{1}{2}}$. 

        \State {Semi-implicit discretization of the probability flow ODE for $L_1$ norm:}
        \[
        p^{k+1}_i = {p}^{k+\frac{1}{2}}_i + \frac{1}{2}\left(S_{h\lambda}(p^{k+\frac{1}{2}}_i) -  \sum_{j=1}^N p^{k+\frac{1}{2}}_j M^p_{i,j} \right)\,,
        \]
     where $M^u_{i,j}$ is defined in \eqref{particle_scheme_general} is defined in \eqref{def_Aij_Mij} with $x^{k+\frac{1}{2}}$ replaced by $p^{k+\frac{1}{2}}$.
 
         \State {Gradient ascent for the dual variable:}
        \[
        y^{k+1}_i = P_{\|\cdot\|_{\infty}\leq 1}\left\{y^k_i + \tau \gamma[I,-D] \begin{bmatrix} 2p^{k+1}_i - p^k_i \\ 2u^{k+1}_i - u^k_i \end{bmatrix}\right\};
        \]
        where $P_{\|\cdot\|_{\infty}\leq 1}$ is the projection to the $L_{\infty}$ ball defined as
        \[
       P_{\|\cdot\|_{\infty}\leq 1}(x_j) = \frac{x_j}{\max\{|x_j|,1\}}\,.
        \]
    \EndFor
\end{algorithmic}
\end{algorithm}


\medskip
\noindent\textbf{Remark:}
\blue{Before we move on to the numerical experiments, we made a remark regarding the connection between the proposed scheme and the transformer structure, which renders a deeper understanding and allows further application of the proposed approach.}
We now recall the interacting particle system formulation for transformers, as discussed in \cite{castin2025unified,geshkovski2023mathematical, sander2022sinkformers}. In a transformer, each data point, represented as a vector, namely a token, is processed iteratively through a series of layers with attention functions. A key component of each layer is the self-attention mechanism, which enables interactions among all tokens.
More specifically, in the simplified single-headed softmax self-attention mechanism, define  \( V \in \mathbb{R}^{d\times d}\) (value), \( Q\in \mathbb{R}^{m\times d} \) (query), and \( K\in \mathbb{R}^{m\times d} \) (key) as learnable matrices, and define the softmax function for $\omega\in\mathbb{R}^N$ as 
\[
\textrm{softmax}(\omega) = \left(\frac{\exp(\omega_i)}{\sum_{\ell=1}^N \exp(\omega_{\ell})}\right)_{1\leq j \leq N}\,.
\]  
The residual transformer structure can be written as a dynamic system\[
\frac{d}{dt}x^{k}  = \sum_{j=1}^N \text{softmax}((Q x_i^k \cdot K x_j^k)_j) V x^k_j\,,
\]
where the softmax function is evaluated at index $j$, which becomes
\[
x^{k+i}  =  x^{k} + h\sum_{j=1}^N \text{softmax}((Q x_i^k \cdot K x_j^k)_j) V x^k_j\,,
\]
 This formulation naturally represents the transformer as an interacting particle system, where the interaction kernel is given by \( Q x_i^k \cdot K x_j^k \). 

Leveraging this perspective, we rewrite the proposed iterative sampling scheme in \eqref{particle_scheme_general} as  
\begin{align}
\label{particle_softmax}
    x_i^{k+1} &= x_i^{k+\frac{1}{2}}+ \frac{1}{2} \left(\prox_g^h(x_i^{k+\frac{1}{2}}) - \sum_{j=1}^N\text{softmax}(U(i,j)) x_j^{k+\frac{1}{2}}\right),\\
   U(i,j) &= -\frac{\beta}{2} \left( \frac{\|x_i^{k+\frac{1}{2}} - x_j^{k+\frac{1}{2}}\|_2^2 - \|\prox_g^h(x_j^{k+\frac{1}{2}}) - x_j^{k+\frac{1}{2}}\|_2^2}{2h} - g(\prox_g^h(x_j^{k+\frac{1}{2}})) \right)\,,\notag
\end{align}
where \( x_i^{k+\frac{1}{2}} = x_i^k - h\nabla f(x_i^k) \).  
Here, the interaction kernel is modified by the new matrix operator \( U \), while the value matrix is replaced by gradient descent updates regarding \( f \). Additionally, the proximal term can be approximated as gradient descent of the regularization term $g$, which integrates the prior information into the dynamics, allowing convergence to the target distribution. Especially, when $g=\lambda\|x\|_1$, the shrinkage operator automatically promotes the sparsity of the variables. Since particle interactions are computed via the softmax function, the system \eqref{particle_softmax} can be efficiently implemented on modern GPUs, making it well-suited for high-dimensional sampling applications.

\section{Numerical Experiments}
\label{sec_NE}

In this section, we numerically verify the performance of the proposed sampling algorithm based on the splitting of the regularized Wasserstein proximal operator (BRWP-splitting, or BRWP for short). Specifically, we use the matrix operator constructed in Proposition \ref{prop_delta_aux_mij} for the first four examples, and the one defined in \eqref{def_Aij_Mij} for the last example to achieve better numerical performance. Numerical experiments include examples of sampling from a mixture distribution, Bayesian logistic regression, image restoration with $L_{1-2}$ TV regularization, uncertainty quantification with Bayesian inference, and Bayesian neural network training. \blue{In particular, the performance of the proposed algorithm will be compared with the Moreau-Yosida Unadjusted Langevin Algorithm (MYULA) \cite{durmus2018efficient} with the regularization parameter $\lambda = 2h$, and the Metropolis-adjusted Proximal Algorithm (PRGO) \cite{RGO_L1} where the appeared restricted Gaussian oracle is sampled by the accelerated gradient method employed in \cite{liang2022proximal}. 
\footnote{The code is in GitHub with the link \url{https://github.com/fq-han/BRWP-splitting}.}  
For the comparison between single-particle-system-based algorithms (such as MYULA and PRGO) and interacting-particle-system-based algorithms (such as BRWP-splitting), we assume that the same number of samples \( N \) is used in each iteration. In other words, each experiment involving MYULA is run with \( N \) independent Markov chains.
}

\subsection{Example 1}
We consider the sampling from a mixture of Gaussian distribution and Laplace distribution, where
\[
\rho^*(x) = \frac{1}{Z}\exp(-(f(x)+\lambda \|x\|_1))\,,\quad \exp(-f(x)) = \sum_{n=1}^M\exp\left(-\frac{(x-y_n)^2}{2\sigma^2}\right)\,,
\]
with $\sigma = 4$ and centers $y_n$ randomly distributed in $[-10,10]^d$. 
To quantify the performance of sampling algorithms, we consider the decay of KL divergence of the one-dimensional marginal distribution, i.e., we plot $\KL(\rho_j\|\rho^*_j)$  where
\[
 \rho_j(x_j) = \int_{\mathbb{R}^{d-1}}\rho(x)dx_1\cdots dx_{j-1}dx_{j+1}\cdots d x_d\,.
\]
The explicit marginal distribution is detailed in the supplementary material. 

We conduct numerical experiments for sampling from the mixture distribution in $d = 20$ and $50$, $\lambda = 0.1$, and $M=4$. Results of the BRWP-splitting are compared with MYULA and PRGO. In Fig.\,\ref{fig_mix_1} and Fig.\,\ref{fig_mix_2}, the decay of KL divergence of the marginal distribution when $j = 1$ and $d$, and the kernel density estimation using a Gaussian kernel with bandwidth $H = 0.1$ from generated samples are plotted where the density is approximated by
$$
\rho(x) = \sum_{j=1}^N \exp\left(-\frac{\|x - x_j\|_2^2}{2H}\right)\,.
$$
\begin{figure}
    \centering
\includegraphics[width=0.24\linewidth]{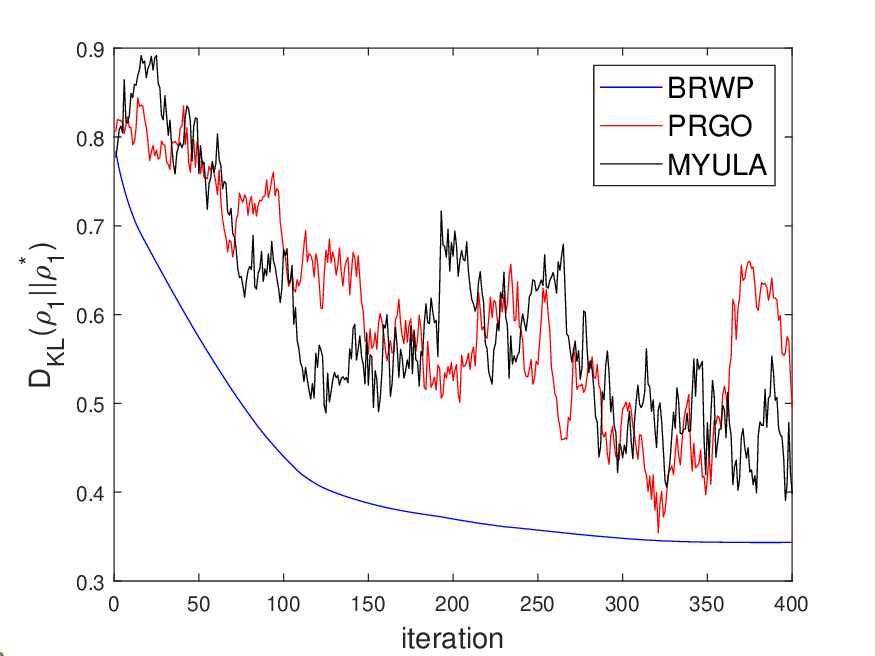}
\includegraphics[width=0.24\linewidth]
{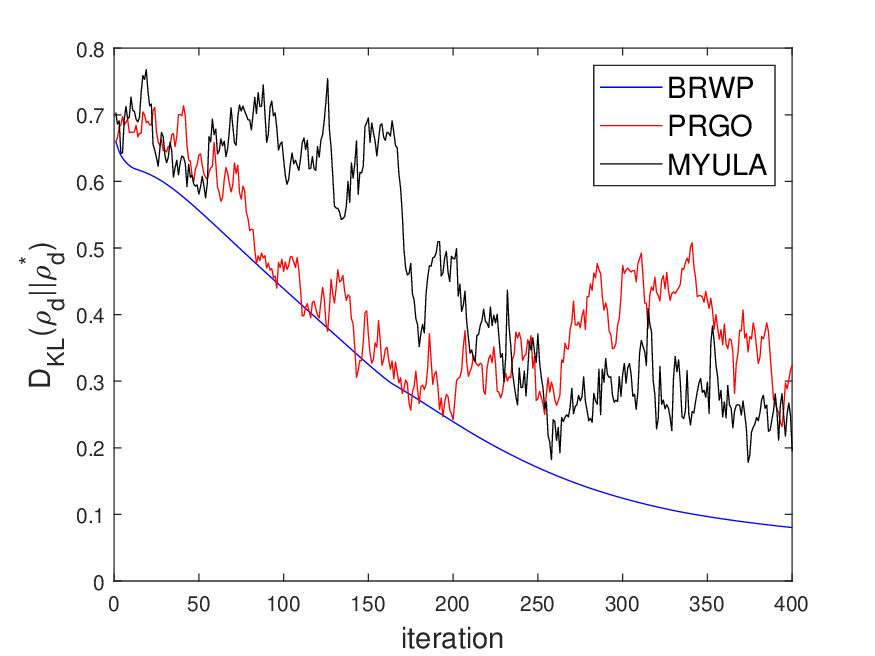}
\includegraphics[width=0.24\linewidth]{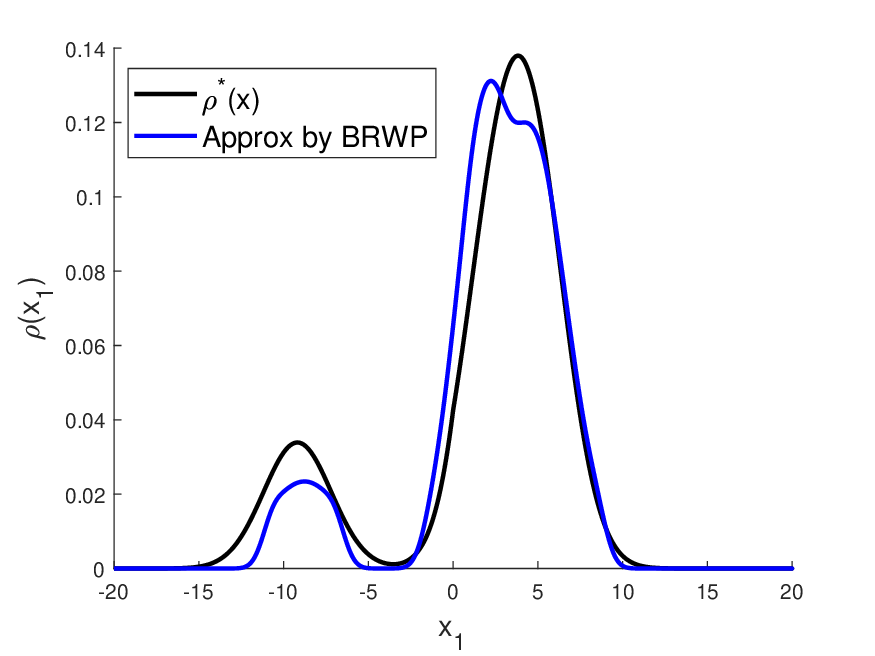}
\includegraphics[width=0.24\linewidth]{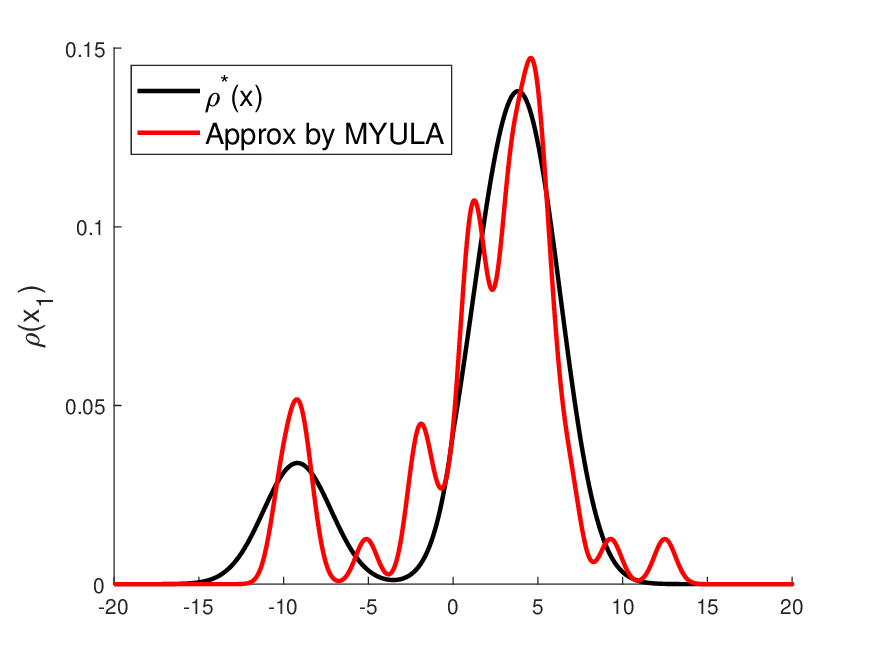}
    \caption{Example 1: Results in $d = 20$, step size $h = 0.02$, and $50$ particles. From left to right: the decay of KL divergence in the first and the last dimension, density approximated by particles generated by BRWP-splitting and MYULA in the first spatial variable.}
    \label{fig_mix_1}
\end{figure}
\begin{figure}
    \centering
\includegraphics[width=0.24\linewidth]{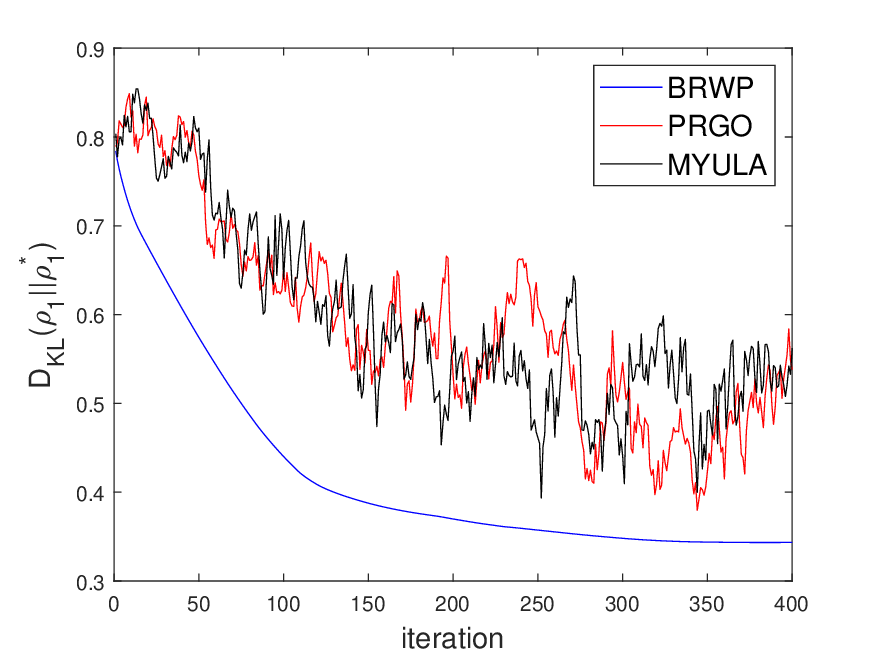}
\includegraphics[width=0.24\linewidth]
{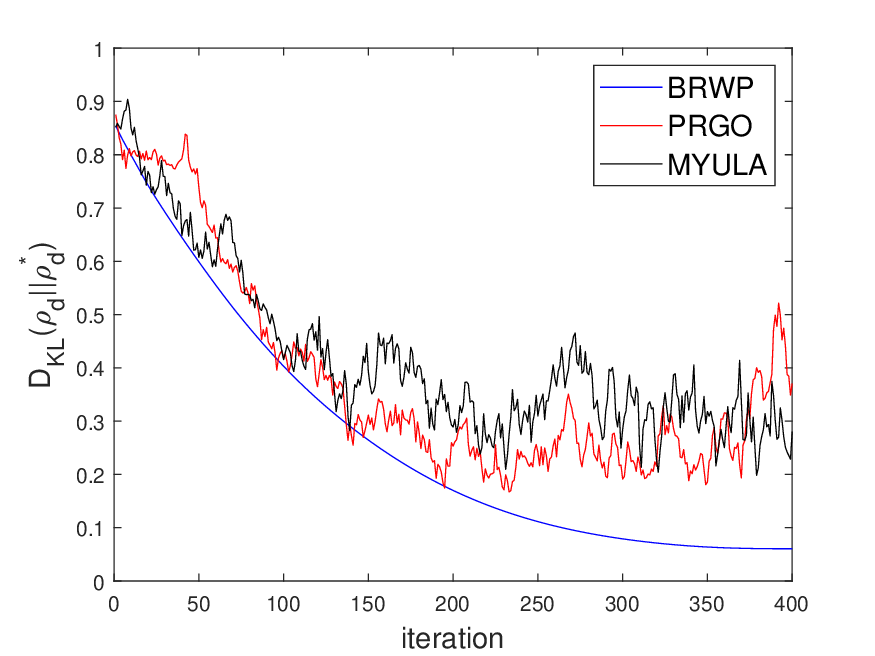}
\includegraphics[width=0.24\linewidth]{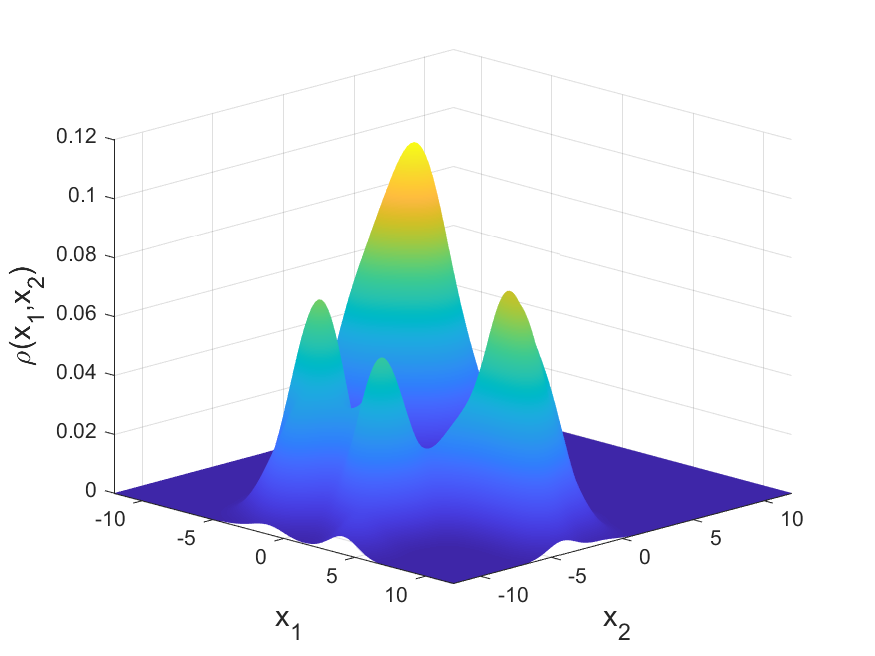}
\includegraphics[width=0.24\linewidth]{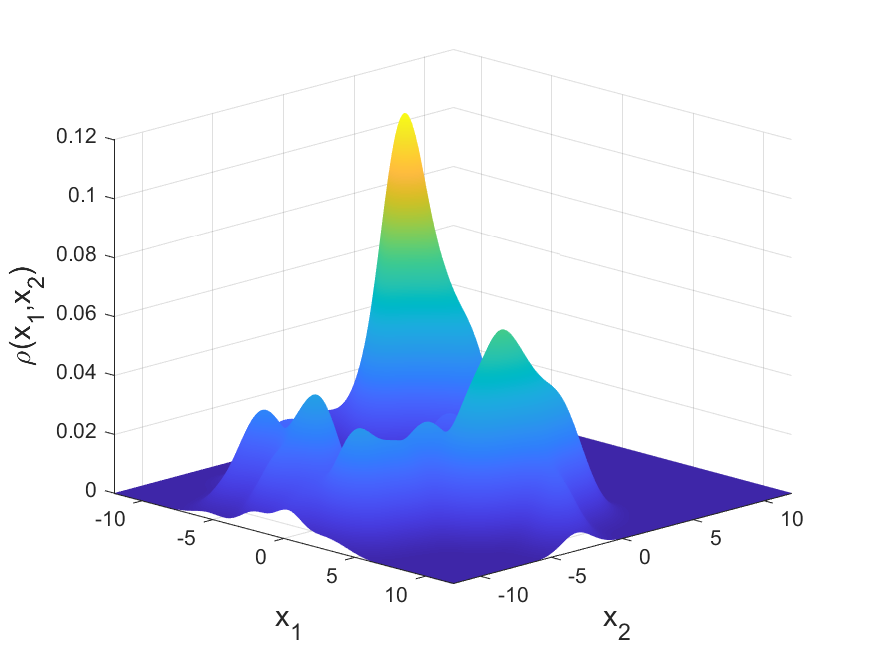}
    \caption{Example 1: Results in $d = 50$, step size $h = 0.02$, and $100$ particles. From left to right: the decay of KL divergence in the first and the last dimension, density approximated by particles generated by BRWP-splitting and MYULA in the first two spatial variables.}
    \label{fig_mix_2}
\end{figure}

Both experiments in Fig.\,\ref{fig_mix_1} and Fig.\,\ref{fig_mix_2} showed that the proposed BRWP-splitting scheme provides a more accurate approximation to the target distribution in terms of KL divergence and the density obtained from kernel density estimation.

\subsection{Example 2}  
The next experiment concerns the Bayesian logistic regression motivated by \cite{dalalyan2017theoretical}. The task is to estimate the unknown parameter $\theta \in \mathbb{R}^d$.
Given binary variable (label) $y = \{0,1\}$ under features (covariate) $x\in \mathbb{R}^n$, the logistic model for $y$ given $x$ can be modeled as
\begin{equation}
    \label{log_model}
    p(y=1|\theta,x) =\frac{\exp(\theta^T x)}{1+\exp(\theta^T x)}\,,
\end{equation}
for some parameter $\theta$ that we try to estimate. 

Suppose now we obtain a set of data pairs $\{(x_i,y_i)\}_{i=1}^n$ where each $y_i$ conditioned on $x_i$ is drawn from a logistic distribution with parameters $\theta^*$. Then using the Bayes rule, we can construct the posterior distribution of parameter $\theta$ in terms of data $\{y_i\}$. 
Denoting $Y = [y_1,\cdots,y_n]$, $X = [x_1,\cdots,x_n]$ and writing $\pi_0(x) = \exp(-\lambda\|x\|_1)$ to be the prior distribution, then the posterior distribution for parameters $\theta$ is computed as
\[
p(\theta|y) = p(y|\theta,x)p_0(\theta) = \frac{1}{Z}\exp\left(Y^TX\theta - \sum_{i=1}^N\log(1+\exp(\theta^T x_i)) - \lambda\|\theta\|_1\right)\,.
\]

For our numerical experiments, $x_i$ is normalized where each component is sampled from Rademacher distribution, i.e., which takes the values $\pm 1$ with probability $\frac{1}{2}$. Given $x_i$, we then draw $y_i$ from the logistic model \eqref{log_model} with $\theta = \theta^*$. The parameter $\theta^*\in\mathbb{R}^d$ contains only $d/4$ non-zero components with value $1$. 
We examine the performance of the algorithm by computing the $L_1$ distance between sample mean $\overline{\theta}$ and the true parameter $\theta^*$
\[
\frac{1}{d}\|\overline{\theta}-\theta^*\|_1\,.
\]
The regularization parameter is chosen as $\lambda = 3d/(2\pi^2)$, and the results are presented in Fig.\,\ref{fig_Bay_log}.
 
\begin{figure}
    \centering
    \includegraphics[width=0.48\linewidth]{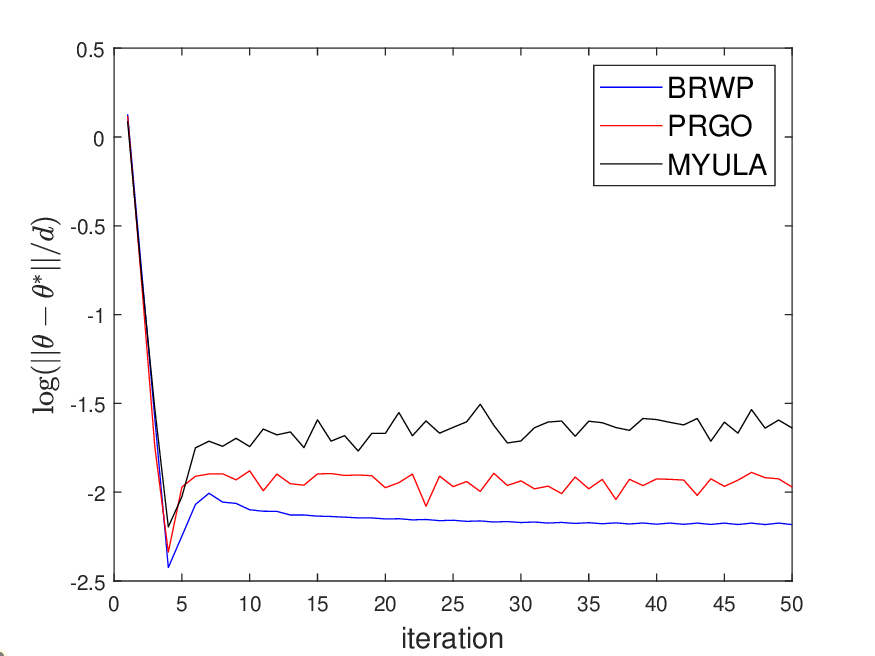}
    \includegraphics[width=0.48\linewidth]{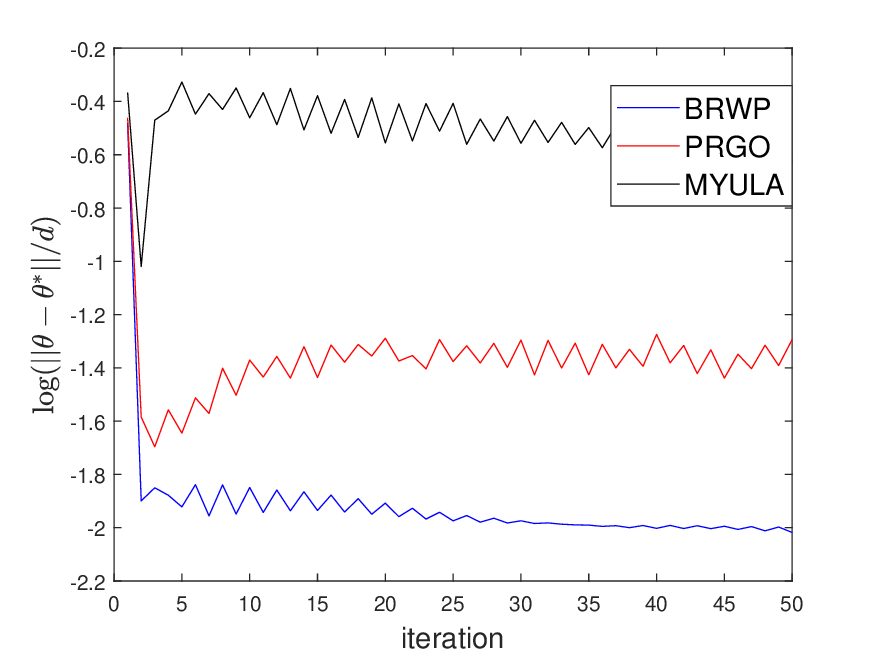}
\label{Bay_log}
\caption{Example 2: Logarithm of relative $L_1$ error $\log\left(\|\overline{\theta}-\theta^*\|_1/d\right)$ in Bayesian logistic regression for 100 particles and $h = 0.05$ with $d=20$ (left) and $d = 50$ (right). }
\label{fig_Bay_log}
\end{figure}

From Fig.\,\ref{fig_Bay_log}, it is clear that the proposed BRWP-splitting method provides a more accurate estimate of the mean parameter in this Bayesian logistic regression.

\subsection{Example 3}
In this example, we apply the proposed sampling algorithm in image denoising with $L_{1-2}$ regularization as proposed in \cite{l12_cs}.

The posterior distribution under consideration is 
\begin{equation}
\rho^*(u) = \frac{1}{Z}\exp\left(-\left(\frac{1}{2}\|Au-y\|_2^2 + \lambda (\|D u\|_1 - \|D u\|_2)\right)\right)\,,
\end{equation}
the first term in the exponent is a data-fitting term and the second term is the difference between $L_1$ and $L_2$ norm with the discrete gradient operator defined in section \ref{sec_TV_imaging} which promotes the sparsity of the image variation. Here, each $u$ corresponds to one single image. 
To tackle this, the log-density is split as
\begin{equation}
f = \|Au-y\|_2^2 - \lambda \|D u\|_2\,,\quad g = \lambda \|D u\|_1\,.
\end{equation}
To handle the second terms with $L_1$-TV norm, we apply the algorithm proposed in Algorithm \ref{alg_TV_sampling}.
We consider the case that $A$ is a noisy measurement operator such that
\[
A = I + \epsilon\,,
\]
where $\epsilon$ is a sparse Gaussian noise with mean $0$, variance $0.1$, that has $3d$ non-zero entries. For the exact image $z_{ex}$, the noisy image $z$ is taken as $ Az_{ex}  + \eta$ where  $\eta$ is a Gaussian noise with mean $0$ and variance $0.2$. The results obtained with $20$ samples and $h = 0.1$ are plotted in Fig.\,\ref{fig_TV_1} and Fig.\,\ref{fig_TV_2}.

\begin{figure}
    \centering
\includegraphics[width=0.24\linewidth]{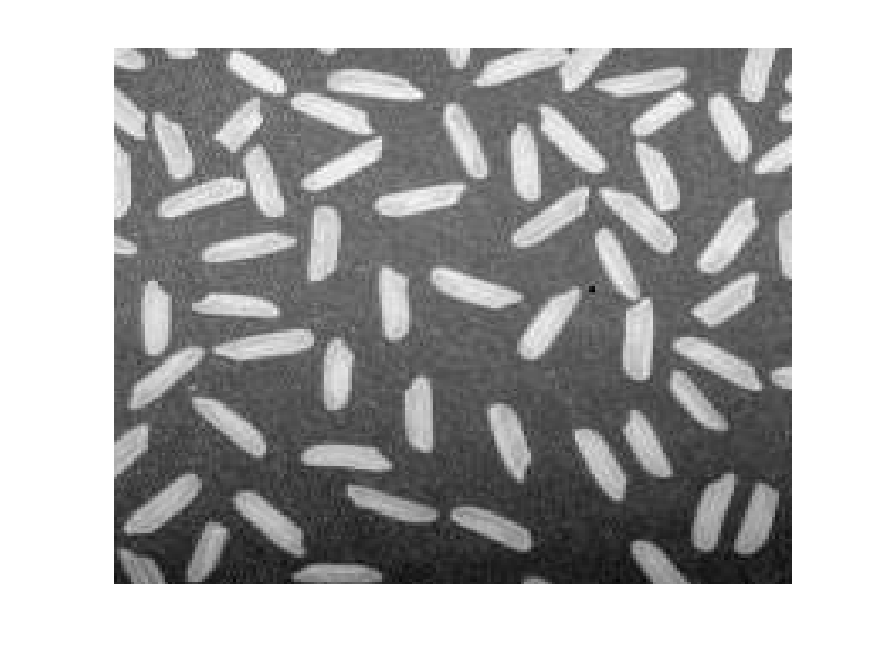}
   \includegraphics[width=0.24\linewidth]{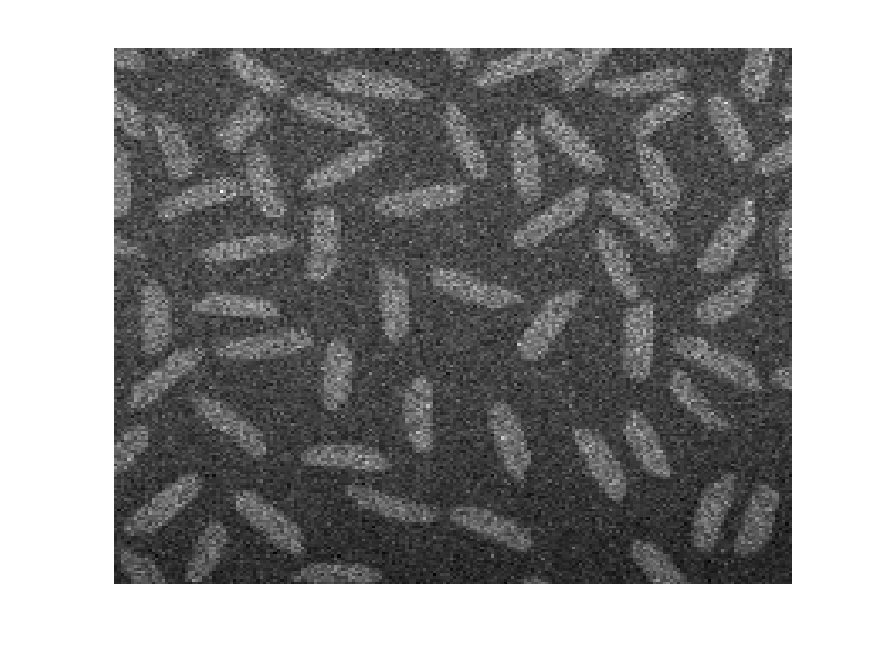}
  \includegraphics[width=0.24\linewidth]{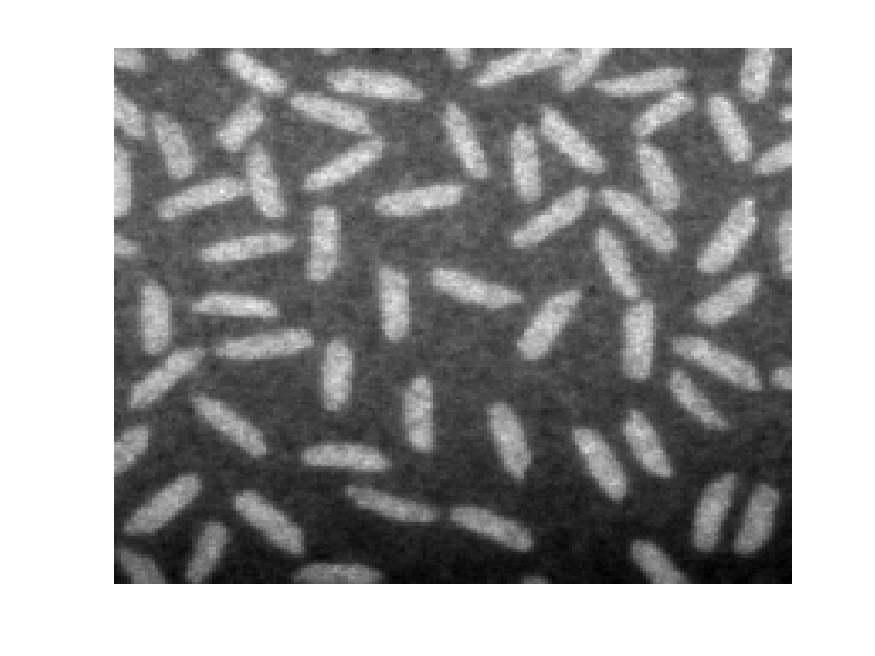}
\includegraphics[width=0.24\linewidth]{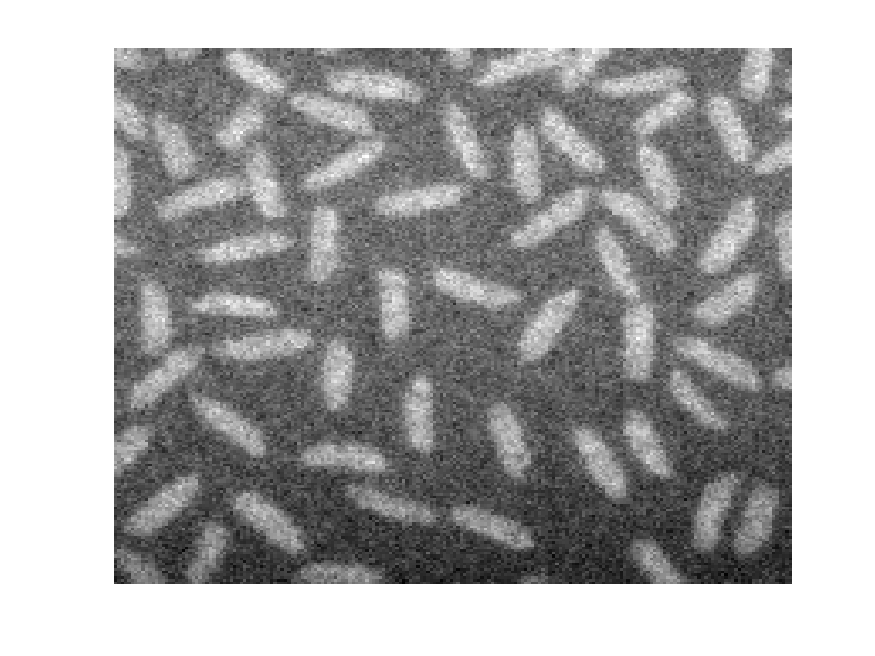}
    \caption{Example 3: Left to right: exact image, noisy image, mean of all samples after 100 iterations by BRWP-splitting and MYULA.}
\label{fig_TV_1}
\end{figure}
\begin{figure}
    \centering
\includegraphics[width=0.24\linewidth]{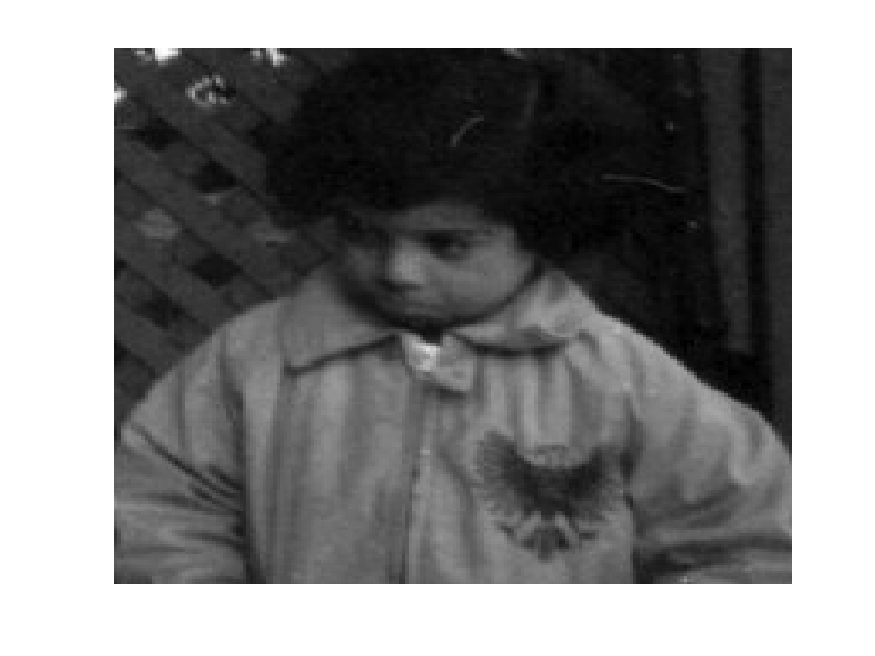}
   \includegraphics[width=0.24\linewidth]{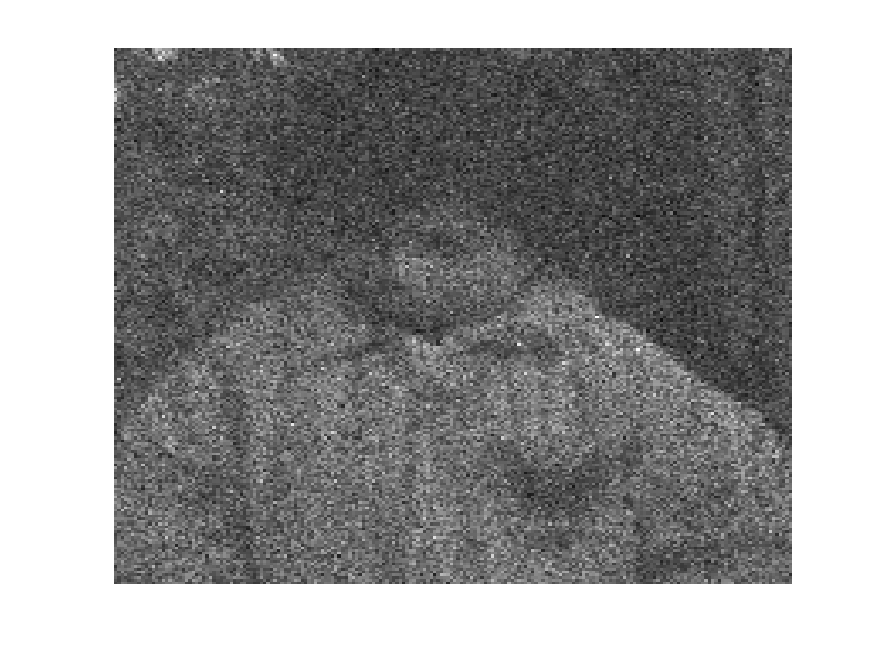}
  \includegraphics[width=0.24\linewidth]{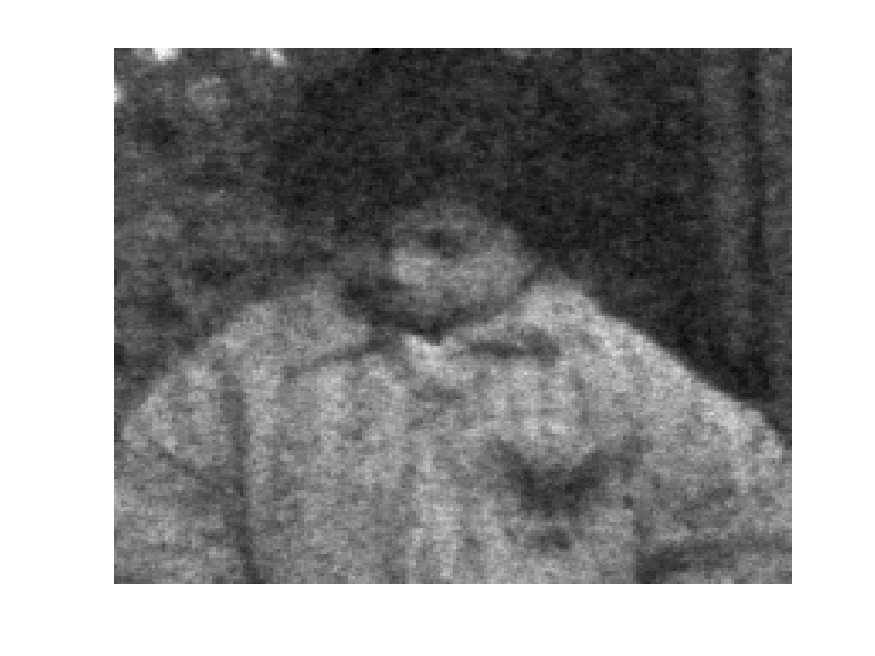}
\includegraphics[width=0.24\linewidth]{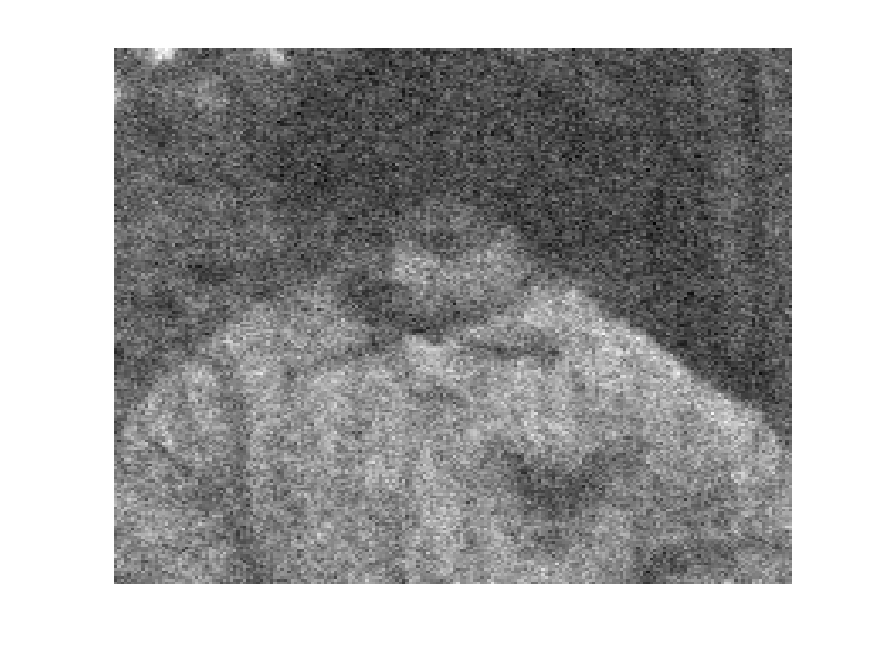}
    \caption{Example 3: Left to right: exact image, noisy image, mean of all samples after 100 iterations by BRWP-splitting and MYULA.}
    \label{fig_TV_2}
\end{figure}
From both Fig.\,\ref{fig_TV_1} and Fig.\,\ref{fig_TV_2}, the proposed sampling method recovers the original image from noisy data properly with $L_{1-2}$ TV regularization. 
 
\subsection{Example 4} In the next example, we examine the application of the proposed sampling algorithm for a compressive sensing application with $L_1$ regularization. The target function for this problem is defined as
\begin{equation}
    \rho^*(x) = \frac{1}{Z}\exp(-\left(\|Ax - z\|_2^2 + \lambda\|x\|_1\right))\,,
\end{equation}
where $x\in \mathbb{R}^d$, \( A \) is a \(m \times d\) circulant blurring matrix with $m = d/4$.

To quantify the uncertainty in the measurement data, we consider the concept of the highest posterior density (HPD). For a given confidence level \(\alpha \in [0, 1]\), the HPD region \(C_\alpha\) is defined as
\[
\int_{C_\alpha} \rho(x) \, dx = 1 - \alpha \,, \quad C_\alpha := \{x\in\mathbb{R}^d : V(x) \leq \eta_\alpha\} \,,
\]
where \(\eta_\alpha\) is a threshold corresponding to the confidence level. The integral can be numerically approximated using samples we get from the BRWP-splitting algorithm. For an arbitrary test image \( \tilde{x} \), by comparing \( V(\tilde{x}) \) with \( \eta_{\alpha} \) for various \( \alpha \), we can assess the confidence that \( \tilde{x} \) belongs to the high-probability region of the posterior distribution. In particular, with the set of particles generated from the BRWP-splitting scheme, the integral is computed numerically as
\[
\int_{C_\alpha} \rho(x) \, dx \approx \frac{\sum_j \mathcal{\chi}_{V(x_j) < \eta_\alpha}}{N} \,,
\]
where \(N\) is the total number of samples, and \(\mathcal{\chi}_{V(x_j) < \eta_\alpha}\) is the indicator function equals to 1 if \(V(x_j) < \eta_\alpha\) and 0 otherwise.

We test the algorithm on a brain MRI image of size \( d = 128^2 \). The measurement model is assumed to be \( Ax + \epsilon \), where \( \epsilon\) represents Gaussian noise with mean 0 and variance 0.2. The reconstruction is estimated using a step size \( h = 0.02 \), with 100 samples and 100 iterations. Additionally, we compute the HPD region threshold and plot the graph of \( \eta_{\alpha} \) versus \( \alpha \), which is estimated using 1000 samples.
\begin{figure}
    \centering
    \includegraphics[width=0.24\linewidth]{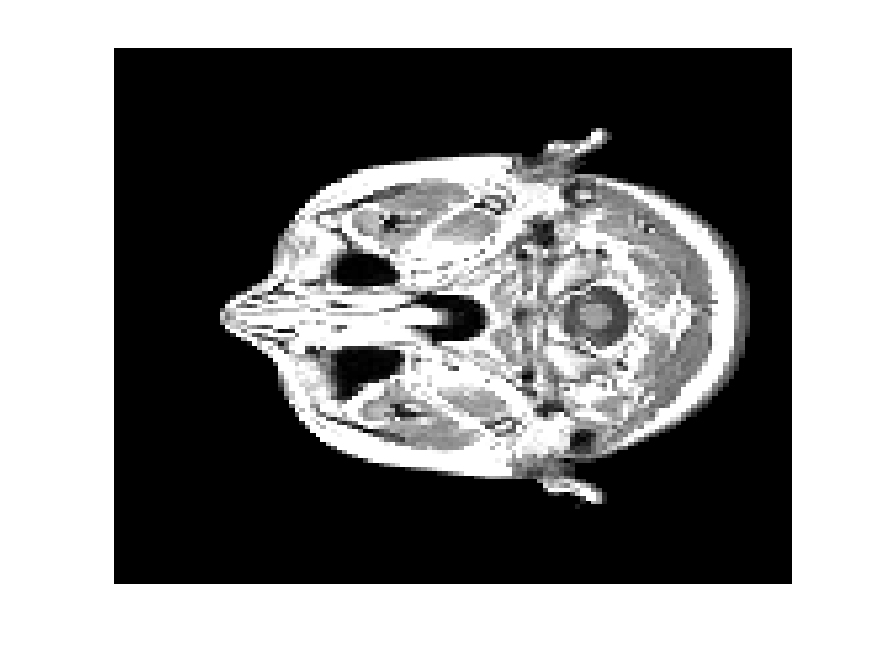}
     \includegraphics[width=0.24\linewidth]{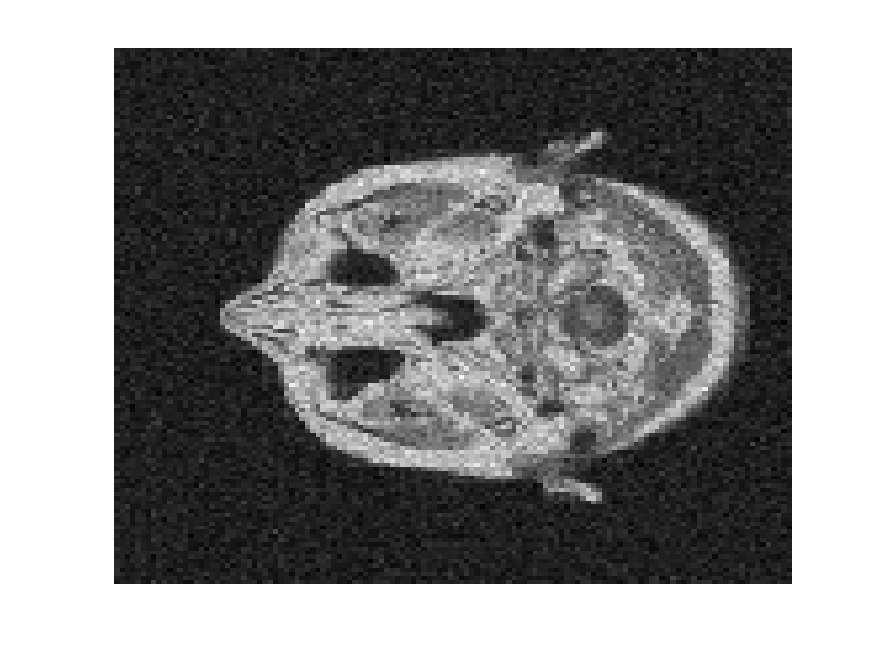}
    \includegraphics[width=0.24\linewidth]{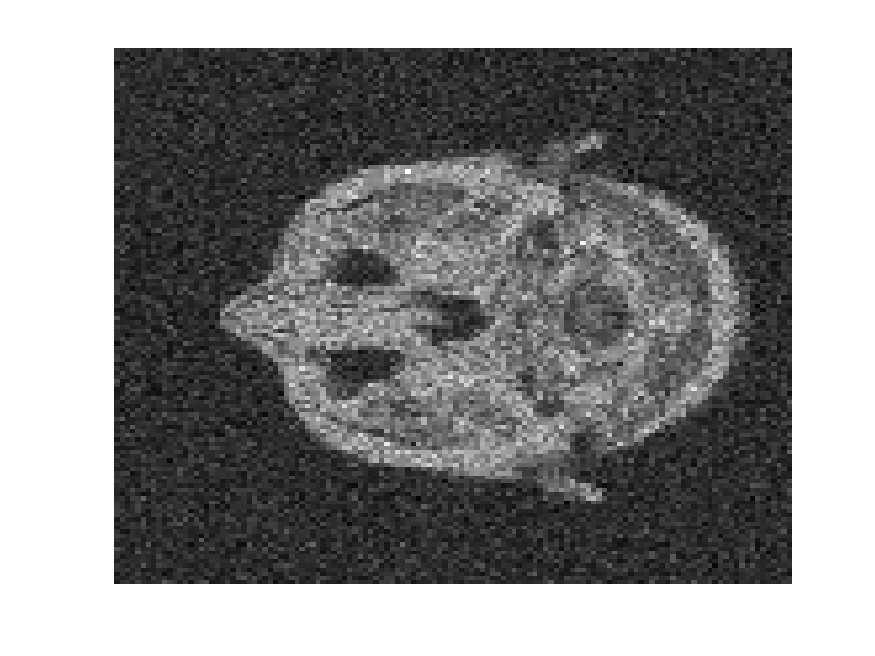}
\includegraphics[width=0.24\linewidth]{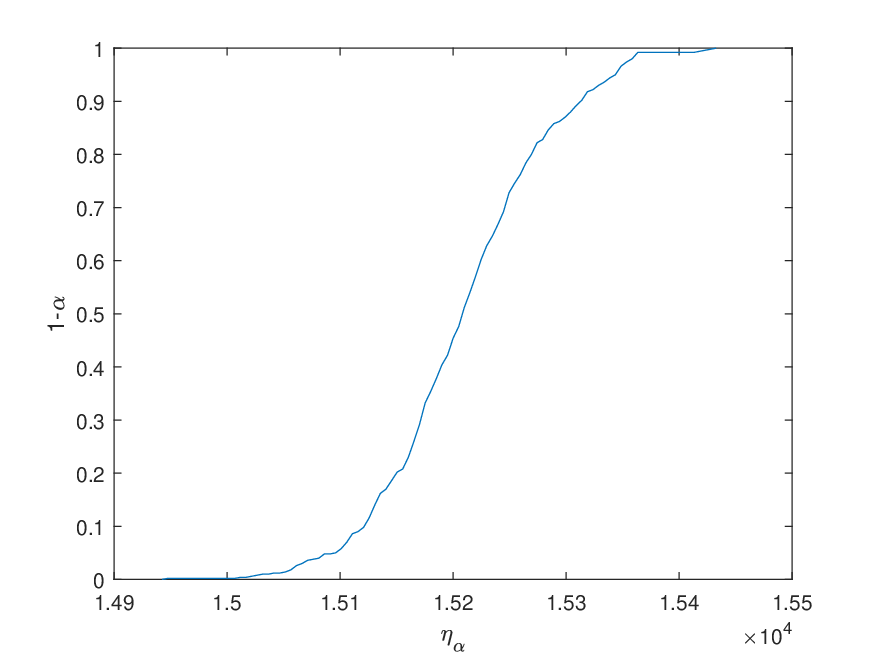}
    \caption{Example 4: From left to right: exact MRI image, reconstructed MRI image with BRWP-splitting, reconstructed MRI image with MYULA, HDP region thresholds $\eta_{\alpha}$.}
    \label{fig_Bay_inf}
\end{figure}
From Fig.\ref{fig_Bay_inf}, we observe that the proposed algorithm yields a better reconstruction compared with MYULA. Furthermore, the sampling approach allows us to compute the HPD region threshold, facilitating practical Bayesian inference analysis.
 
 \subsection{Example 5}In this example, we apply the proposed method to Bayesian neural network training. Specifically, the likelihood function is modeled as an isotropic Gaussian, and the prior distribution is a Laplace prior. 
We consider a two-layer neural network, where each layer consists of 50 hidden units with a ReLU activation function. For each dataset, $90\%$ of the data is used for training, while the remaining $10\%$ is reserved for testing. Each algorithm is simulated using 200 particles over 500 iterations.

We compare the BRWP-splitting against MYULA, the original BRWP (non-splitting, without proximal computation), and SVGD (Stein variational gradient descent). The step size for each method is selected via grid search to achieve the best performance, and it remains consistent across all experiments.

\begin{table}[]
    \centering
    \begin{tabular}{c|c|c|c|c}
        Dataset & BRWP-splitting & BRWP & MYULA & SVGD \\
        \hline
        Boston   & $\mathbf{3.78_{\pm 1.93\times 10^{-1}}}$ & $4.27_{\pm 2.09\times 10^{-2}}$ & $6.29_{\pm 6.00\times 10^{-3}}$ & $4.05_{\pm 6.93\times 10^{-2}}$ \\
        Wine     & $\mathbf{0.53_{\pm 2.54\times 10^{-1}}}$ & $0.61_{\pm 2.47\times 10^{-1}}$ & $0.72_{\pm 1.13\times 10^{-1}}$ & $0.54_{\pm 3.64\times 10^{-1}}$ \\
        Concrete & $\mathbf{3.25_{\pm 1.37\times 10^{-1}}}$ & $4.11_{\pm 1.02\times 10^{-1}}$ & $4.71_{\pm 3.14\times 10^{-1}}$ & $3.32_{\pm 1.47\times 10^{-1}}$ \\
        Kin8nm   & ${0.093_{\pm 7.99\times 10^{-4}}}$ & $0.135_{\pm 2.15\times 10^{-3}}$  & $0.294_{\pm 1.56\times 10^{-3}}$ & $\mathbf{0.092_{\pm 7.93\times 10^{-4}}}$ \\
        Power    & $\mathbf{4.13_{\pm 3.21\times 10^{-2}}}$  &  $5.25_{\pm 8.42\times 10^{-2}}$& $8.49_{\pm 2.87\times 10^{-1}}$ & $4.15_{\pm 1.63\times 10^{-2}}$ \\
        Protein  & $\mathbf{4.23_{\pm 2.17\times 10^{-2}}}$  & $4.74_{\pm 4.32\times 10^{-2}}$ & $5.12_{\pm 7.32\times 10^{-2}}$  & $4.61_{\pm 1.93\times 10^{-2}}$   \\
        Energy   &  $\mathbf{1.54_{\pm 2.37\times 10^{-2}}}$& $3.06_{\pm 6.06\times 10^{-2}}$ &$4.52_{\pm 2.42}$  & $2.00_{\pm 4.13\times 10^{-2}}$  \\
    \end{tabular}
    \vspace{0.2cm}
    \caption{Example 5: Root-mean-square error for different datasets in Bayesian neural network training with $\lambda = 1/d$.}
    \label{table_bayesian}
\end{table}
From Table \ref{table_bayesian}, we observe that, for most datasets tested, the proposed BRWP-splitting approach achieves a lower root-mean-square error compared to the other methods. 
\section{Discussions}

In this work, we propose a sampling algorithm based on splitting methods and regularized Wasserstein proximal operators for sampling from nonsmooth distributions. When the log-density of the prior distribution is the \(L_1\) norm, the scheme is formulated as an interacting particle system incorporating shrinkage operators and the softmax function. The resulting iterative sampling scheme is simple to implement and naturally promotes sparsity. Theoretical convergence of the proposed scheme is established under suitable conditions and the algorithm's efficiency is demonstrated through extensive numerical experiments.

For future directions, we aim to extend our theoretical analysis to investigate the algorithm’s convergence in the finite-particle approximation and explore its applicability beyond log-concave sampling. On the computational side, we seek to enhance efficiency through GPU-based parallel implementations and examine the impact of different kernel choices on the performance. Additionally, as discussed at the end of Section \ref{sec_TV_imaging}, regularized Wasserstein proximal operators share a close structural connection with transformer architectures, motivating our interest in analyzing the self-attention mechanism through the lens of interacting particle systems. More importantly, building on the proposed algorithm, we plan to develop tailored transformer models for learning sparse data distributions, which are only known by samples.

\noindent\textbf{Acknowledgement}: 
F. Han's work is partially supported by AFOSR YIP award No. FA9550 -23-10087. F. Han and S. Osher's work is partially supported by ONR N00014-20-1-2787, NSF-2208272, STROBE NSF-1554564, and NSF 2345256.
W. Li's work is supported by AFOSR YIP award No. FA9550-23-10087, NSF RTG: 2038080, and NSF DMS-2245097.
\begin{appendix}
\section{Derivation in Section \ref{sec_algo}}
\label{sec_appendix_derivation}
\begin{proof}[Proof of proposition \ref{prop_gaussian_mij}]
   For the case $\rho$ is approximated with Gaussian kernel, writing $x_{\ell}$ as the $\ell$-th component of $x\in \mathbb{R}^d$, we note \eqref{rho_T_BRWP} becomes
\begin{align*}
K_g^{h,\beta}\rho(x) &= \frac{1}{N(2\pi\sigma^2)^{d/2}}\exp\left(-\frac{\beta}{2}\lambda\|x\|_1\right)\cdot\\
&\sum_{j=1}^N\prod_{\ell=1}^d\int_{\mathbb{R}}\exp\left[-\frac{\beta}{2}\left(\frac{(x_{\ell}-y_{\ell})^2-(y_{\ell}- S_{\lambda h}(y_{\ell}))^2}{2h} - \lambda| S_{\lambda h}(y_{\ell})|\right)-\frac{(y_{\ell}-x_{j,\ell})^2}{2\sigma^2} \right] dy_{\ell}\,.
\end{align*}

Hence, obtaining the closed-form formula reduces to evaluating a one-dimensional exponential integral. This integral can be decomposed into three parts: \( [\lambda h, \infty) \), \( (-\lambda h, \lambda h) \), and \( (-\infty, -\lambda h] \), following the definition of the shrinking operator \( S_{\lambda h}(y) \).  Defining \( c = 2h / (\sigma^2 \beta) \), denoting $\psi(x,x_j) = \exp\left(-\frac{\beta}{2}\left(\frac{x^2+cx_j^2}{2h}\right)\right)$ to simplify notation, and omitting the index \( \ell \) for simplicity, the integral over \( [\lambda h, \infty) \) is given by  
\begin{align*}
    &\psi(x,x_j)\int_{\lambda h}^{\infty}\exp\left(-\frac{\beta}{4h}\left[(1+c)y^2 - 2y(x+cx_j+\lambda h) \right]\right) dy\exp\left(-\frac{\beta\lambda^2 h}{4} \right)\\
    =&\psi(x,x_j)\sqrt{\frac{4h}{\beta(1+c)}} \\&\int_{\sqrt{\frac{\beta(1+c)}{4h}}\left[\lambda h - \frac{x+cx_j+\lambda h}{1+c}\right]}^{\infty}\exp(-y^2)dy\exp\left(-\frac{\beta}{4h}\left(\lambda^2h^2-\frac{(x+cx_j+\lambda h)^2}{1+c}\right)\right)\,.
\end{align*}
Similarly, the integral on $(-\infty,-\lambda h]$ will be
\begin{align*}
    &\psi(x,x_j)\int_{-\infty}^{-\lambda h}\exp\left(-\frac{\beta}{4h}\left[(1+c)y^2 - 2y(x+cx_j-\lambda h) \right]\right) dy\exp\left(-\frac{\beta\lambda^2 h}{4} \right)\\
    =&\psi(x,x_j)\sqrt{\frac{4h}{\beta(1+c)}}\\&\int^{\sqrt{\frac{\beta(1+c)}{4h}}\left[-\lambda h - \frac{x+cx_j-\lambda h}{(1+c)}\right]}_{-\infty}\exp(-y^2)dy\exp\left(-\frac{\beta}{4h}\left(\lambda^2h^2-\frac{(x+cx_j-\lambda h)^2}{1+c}\right)\right)\,.
\end{align*}

Finally the integral on $[-\lambda h, \lambda h]$ can be computed as
\begin{align*}
    &\psi(x,x_j)\int_{-\lambda h}^{\lambda h}\exp\left(-\frac{\beta}{4h}\left(cy^2 - 2y(x+cx_j)\right)\right)dy\\
    = &\psi(x,x_j)\int_{-\lambda h}^{\lambda h}\exp\left(-\frac{c\beta}{4h}\left[y-\frac{(x+cx_j)}{c}\right]^2 \right) dy \exp\left(\frac{\beta}{4h} \frac{(x+cx_j)^2}{c}  \right)\\
    =&\psi(x,x_j)\sqrt{\frac{4h}{c\beta }}\int_{\sqrt{\frac{c\beta}{4h}}\left[-\lambda h-\frac{(x+cx_j)}{c}\right]}^{\sqrt{\frac{c\beta}{4h}}\left[\lambda h-\frac{(x+cx_j)}{c}\right]} \exp(-y^2) dy \exp\left(\frac{\beta}{4h} \frac{(x+cx_j)^2}{c}  \right)\,.
\end{align*}

Next, to compute the score function, we need to evaluate $\nabla K_g^{h,\beta}\rho_{k}$ on each sub-integral. For the integral on $[\lambda h,\infty)$, omitting the $\psi$ term, direct computation implies the following
\begin{align*}
   & \nabla\left\{\sqrt{\frac{4h}{\beta(1+c)}}\int_{\sqrt{\frac{\beta(1+c)}{4h}}\left[\lambda h - \frac{x+cx_j+\lambda h}{1+c}\right]}^{\infty}\exp(-y^2)dy\exp\left(-\frac{\beta}{4h}\left(\lambda^2h^2-\frac{(x+cx_j+\lambda h)^2}{1+c}\right)\right)\right\}\\
    =&\Bigg\{\sqrt{\frac{\beta}{h(1+c)}} (x+cx_j+\lambda h)\int_{\sqrt{\frac{\beta(1+c)}{4h}}\left[\lambda h - \frac{x+cx_j+\lambda h}{1+c}\right]}^{\infty}\exp(-y^2)dy\\
    &+\exp\left[-\frac{\beta(1+c)}{4h}\left(\lambda h - \frac{x+cx_j+\lambda h}{1+c}\right)^2 \right]\Bigg\}\cdot\frac{\exp\left(-\frac{\beta}{4h}\left(\lambda^2h^2-\frac{(x+cx_j+\lambda h)^2}{1+c}\right)\right)}{1+c}\,.
    \end{align*}
The gradient for the integral on $(-\infty,-\lambda h]$ can be evaluated similarly to the above by replacing $x+cx_j+\lambda h$ with $x+c_j-\lambda h$ and change of signs.
 Finally, the gradient for the integral on $(-\lambda h, \lambda h)$ can be evaluated as
 \begin{align*}
   & \nabla\left\{\sqrt{\frac{4h}{c\beta }}\int_{\sqrt{\frac{c\beta}{4h}}\left[-\lambda h-\frac{(x+cx_j)}{c}\right]}^{\sqrt{\frac{c\beta}{4h}}\left[\lambda h-\frac{(x+cx_j)}{c}\right]} \exp(-y^2) dy \exp\left(\frac{\beta}{4h} \frac{(x+cx_j)^2}{c}  \right)\right\}\\
    =&\exp\left(\frac{\beta}{4h}\frac{(x+cx_j)^2}{c} \right)\cdot \Bigg\{\sqrt{\frac{4h}{c\beta}}\frac{\beta}{2h} \frac{(x+cx_j)}{c}\int_{\sqrt{\frac{c\beta}{4h}}\left[-\lambda h-\frac{(x+cx_j)}{c}\right]}^{\sqrt{\frac{c\beta}{4h}}\left[\lambda h-\frac{(x+cx_j)}{c}\right]} \exp(-y^2) dy  \\
    & -\frac{1}{c}\left[\exp\left(-\frac{c\beta}{4h}\left(\lambda h - \frac{x+cx_j}{c}\right)^2 \right)  - \exp\left(-\frac{c\beta}{4h}\left(-\lambda h - \frac{x+cx_j}{c}\right)^2 \right)\right]\Bigg\}   \,.
 \end{align*}
Combining the above gives the desired result. 
\end{proof}

\begin{proof}[Proof of proposition \ref{prop_delta_aux_mij}]
  For the sum of $A_{i,j}$ defined in \eqref{def_Aij_Mij}, we have
\begin{align*}
&\sum_{j=1}^{N^d} A_{i,j} =  \sum_{j=1}^{N^d}\exp\left(-\frac{\beta}{2}\left(\frac{\|x_i-\tx_j\|_2^2}{2h}\right)\right)\exp\left(\frac{\beta}{2}\left(\frac{\| S_{\lambda h}(\tx_j)-\tx_j\|_2^2}{2h}+\lambda\| S_{\lambda h}(\tx_j)\|_1\right) \right)\\
=& \sum_{j=1}^{N^d}\prod_{\ell=1}^d \exp\left(-\frac{\beta}{2}\left(\frac{(x_{i,\ell}-\tx_{j,\ell})^2}{2h}\right)\right)\exp\left(\frac{\beta}{2}\left(\frac{( S_{\lambda h}(\tx_{j,\ell})-\tx_{j,\ell})^2}{2h}+\lambda| S_{\lambda h}(\tx_{j,\ell})|\right) \right)\\
=& \sum_{j_1=1}^{N}\cdots \sum_{j_{d-1}=1}^N\prod_{\ell=1}^{d-1}\sum_{j_d=1}^N \exp\left(-\frac{\beta}{2}\left(\frac{(x_{i,\ell}-x_{j_{\ell},\ell}}{2h}\right)\right)\cdot \\&\hspace{6cm}\exp\left(\frac{\beta}{2}\left(\frac{( S_{\lambda h}(x_{j_{\ell},\ell})-x_{j_{\ell},\ell})^2}{2h}+\lambda| S_{\lambda h}(x_{j_{\ell},\ell})|\right) \right) \,.
\end{align*}
Then, we can rearrange the terms to get
\[
\sum_{j=1}^{N^d} A_{i,j}= \prod_{\ell=1}^d\sum_{j=1}^N \exp\left(-\frac{\beta}{2}\left(\frac{(x_{i,\ell}-x_{j,\ell})^2}{2h}\right)\right)\exp\left(\frac{\beta}{2}\left(\frac{( S_{\lambda h}(x_{j,\ell})-x_{j,\ell})^2}{2h}+\lambda| S_{\lambda h}(x_{j,\ell})|\right) \right)\,,
\]
which is the desired formula in the proposition.
\end{proof}

\section{Postponed Proof in Section \ref{sec_analysis_smooth}}
\label{appendix_convergence}

Our proof of the convergence of R\'enyi divergence will rely on the interpolation argument by considering the continuity equation of \eqref{appendix_particle}  in time \( t \in [kh, (k+1)h]\). The particle at time $t$ is written as
\begin{align}
\label{x_t_particle}
 &x_{t} - x_{kh} \\
 = & - (t-kh)\left[\nabla f(x_{kh}) +  \nabla g_h(x_{kh}-  h\nabla f(x_{kh}))+\beta^{-1}\nabla\log K_g^{ (t-kh)}\rho_{k+\frac{1}{2}}(x_{kh}-h\nabla f(x_{kh}))\right]\notag\\
=& - (t-kh)\left[\nabla f(x_t) +\nabla g_h(x_t) +\beta^{-1}\nabla\log\rho_t(x_t) +\Lambda(x_{t},x_{kh})\right]\,,\notag
\end{align}
where
\begin{align*}
\Lambda(x_{t},x_{kh}) := & -\beta^{-1}\nabla \log \frac{\rho_t}{\rho_h^*}(x_t) + \beta^{-1}\nabla \log \frac{\rho_t}{\rho_h^*}(x_{kh} - h \nabla f(x_{kh})) \\
 & + \nabla f(x_{kh}) - \nabla f(x_{kh} - h \nabla f(x_{kh})) \\
 &  - \beta^{-1}\nabla \log \rho_t(x_{kh} - h\nabla f(x_{kh})) + \beta^{-1}\nabla \log  K_g^{t-kh}\rho_{k+\frac{1}{2}}(x_{kh} - h\nabla f(x_{kh}))\,.
\end{align*}
We note that when $t = (k+1)h$, we have $x_t = x^{k+1}$, i.e., the location of the particle in the next time step. 

Then the Fokker Planck equation corresponds to \eqref{x_t_particle} for $t\in [kh,(k+1)h]$ will be
\begin{equation}
\label{FPE_interpolation}
    \frac{\partial \rho_t }{\partial t} (x_t) = \beta^{-1}\nabla \cdot \left( \rho_t(x_t) \nabla \log \frac{\rho_t}{\rho_h^*}(x_t)  \right) + \nabla \cdot \left( \rho_t(x_t) \Lambda(x_{t},x_{kh})  \right).
\end{equation}

We now state the following lemma on the time derivative of R\'enyi divergence along \eqref{FPE_interpolation}.

\begin{lemma}
\label{lemma_time_derivative}
For $t\in[kh,(k+1)h]$, the time derivative of the R{\'e}nyi divergence between $\rho_t$ along \eqref{FPE_interpolation} and $\rho_h^*$ satisfies
\begin{equation}
    \frac{\partial}{\partial t} R_q(\rho_t \| \rho_h^*) \leq -\frac{q}{2} \frac{G_q(\rho_t \| \rho_h^*)}{F_q(\rho_t \| \rho_h^*)} + \frac{q}{2 F_q(\rho_t \| \rho_h^*)} \int_{\mathbb{R}^d} \|\Lambda(x_{t},x_{kh})\|_2^2 \left( \frac{\rho_t}{\rho_h^*} \right)^q \rho_h^* \, dx_t\,.
\end{equation}
\end{lemma}

\begin{proof}
By the definition of the R{\'e}nyi divergence, we have
\begin{align*}
&\frac{\partial}{\partial t} R_q(\rho_t \| \rho_h^*) = \frac{q}{q-1} \frac{\int_{\mathbb{R}^d} \left( \frac{\rho_t}{\rho_h^*} \right)^{q-1} \partial_t \rho_t \, dx_t}{F_q(\rho_t \| \rho_h^*)} \\
= &\frac{q}{(q-1) F_q(\rho_t \| \rho_h^*)} \int_{\mathbb{R}^d} \left( \frac{\rho_t}{\rho_h^*} \right)^{q-1} \nabla \cdot \left[ \left( \rho_t\nabla\log\frac{\rho_t}{\rho_h^*}  \right) + \Lambda(x_{t},x_{kh}) \rho_t \right] dx_t \\
=& -\frac{q}{F_q(\rho_t \| \rho_h^*)} \int_{\mathbb{R}^d} \left( \frac{\rho_t}{\rho_h^*} \right)^{q-2} \nabla \frac{\rho_t}{\rho_h^*} \cdot \left[ \left( \rho_t\nabla\log\frac{\rho_t}{\rho_h^*}  \right) + \Lambda(x_{t},x_{kh}) \rho_t \right] dx_t \\
= &-\frac{q}{F_q(\rho_t \| \rho_h^*)} \left[\int_{\mathbb{R}^d} \left\| \nabla \frac{\rho_t}{\rho_h^*} \right\|_2^2 \left( \frac{\rho_t}{\rho_h^*} \right)^{q-2} \rho_h^*  \, dx_t + \int_{\mathbb{R}^d} \left( \frac{\rho_t}{\rho_h^*} \right)^{q-1} \nabla \frac{\rho_t}{\rho_h^*}\cdot \Lambda(x_{t},x_{kh}) \rho_h^*  \, dx_t\right]\,.
\end{align*}
The first term is precisely the R{\'e}nyi information term defined in \eqref{def_Renyi_infor}, and the second term represents the discretization error, which we need to bound. The second term can be further simplified as follows:
\begin{align*}
&\int_{\mathbb{R}^d} \left( \frac{\rho_t}{\rho_h^*} \right)^{q-1} \nabla \frac{\rho_t}{\rho_h^*} \cdot\Lambda(x_{t},x_{kh}) \rho_h^* \, dx_t = \int_{\mathbb{R}^d}   \nabla \frac{\rho_t}{\rho_h^*} \cdot\left[\Lambda(x_{t},x_{kh})\frac{\rho_t}{\rho_h^*}\right]\left( \frac{\rho_t}{\rho_h^*} \right)^{q-2} \rho_h^* \, dx_t\\
 \geq &-\frac{1}{2} \int_{\mathbb{R}^d} \left\| \nabla \frac{\rho_t}{\rho_h^*} \right\|_2^2 \left( \frac{\rho_t}{\rho_h^*} \right)^{q-2} \rho_h^* \, dx_t - \frac{1}{2} \int_{\mathbb{R}^d} \|\Lambda(x_{t},x_{kh})\|_2^2 \left( \frac{\rho_t}{\rho_h^*} \right)^{q} \rho_h^* \, dx_t\,.
\end{align*}
The final result is obtained by combining the above relations.
\end{proof}

Next, we will bound the discretization error using the Lipschitz continuity of the score function and \((f + g_h)\), i.e., assumptions A.1, A.2, A.4.

\begin{lemma}
\label{lemma_dis_error}
The discretization error term $\Lambda(x_t,x_{kh})$ satisfies
\begin{equation}
\int_{\mathbb{R}^d}\|\Lambda(x_{t},x_{kh})\|_2^2\left(\frac{\rho_t}{\rho_h^*}\right)^{q}\rho_h^*dx_t\leq \frac{2L^2h^2}{(1-hL)^2}G_q(\rho_t\|\rho_h^*) + 2L_f^2(L+L_f)^2h^2 d F_q(\rho_t\|\rho_h^*) + \mathcal{O}(h^3)\,,
\end{equation}
where $L = L_f + L_{g_h} + L_{\rho}$. 
\end{lemma}

\begin{proof}
Firstly, using the gradient Lipschitz condition on $f$, $g_h$, and $\log\rho_t$, and also the approximation result in Lemma \ref{lemma_kernel}, we can bound the discretization error as
\begin{align}
\label{Lambda_error}
\|\Lambda(x_{t},x_{kh})\|_2 \leq &\, L\|x_t-x_{kh}+h\nabla f(x_{kh})\|_2 + L_fh\|\nabla f(x_{kh}) \|_2 + \mathcal{O}(h^2) \notag\\
\leq &\, L \|x_t-x_{kh}\|_2 + (L+L_f)h\|\nabla f(x_{kh}) \|_2 + \mathcal{O}(h^2) \\
\leq &\, L \|x_t-x_{kh}\|_2 + h(L+L_f)L_f \sqrt{d} + \mathcal{O}(h^2)\,.\notag
\end{align}

For the first term of \eqref{Lambda_error}, by the formula for $x_t$ in \eqref{x_t_particle}, we obtain
\begin{align*}
\left\|x_t-x_{kh}\right\|_2 \leq &\, h\left\|\nabla\log\frac{\rho_t}{\rho_h^*}(x_{kh})\right\|_2 + \mathcal{O}(h^2) \\
\leq &\, h\left\|\nabla\log\frac{\rho_t}{\rho_h^*}(x_{t})\right\|_2 + h\left\|\nabla\log\frac{\rho_t}{\rho_h^*}(x_{t})-\nabla\log\frac{\rho_t}{\rho_h^*}(x_{kh})\right\|_2 + \mathcal{O}(h^2) \\
\leq &\, h\left\|\nabla\log\frac{\rho_t}{\rho_h^*}(x_{t})\right\|_2 + Lh\left\|x_t-x_{kh}\right\|_2 + \mathcal{O}(h^2)\,.
\end{align*}
The above leads to 
\[
\|x_t-x_{kh}\|_2 \leq \frac{h}{1-hL}\left\|\nabla\log\frac{\rho_t}{\rho_h^*}(x_t)\right\| + \mathcal{O}(h^2)\,.
\]
Substituting this back into \(\Lambda(x_{t},x_{kh})\), we get
\begin{align*}
&\int_{\mathbb{R}^d} \|\Lambda(x_{t},x_{kh})\|_2^2\left(\frac{\rho_t}{\rho_h^*}\right)^{q}\rho_h^* dx_t \\
\leq &\, \frac{2L^2h^2}{(1-hL)^2}\int_{\mathbb{R}^d}\left\|\nabla\log\frac{\rho_t}{\rho_h^*}\right\|_2^2\left(\frac{\rho_t}{\rho_h^*}\right)^{q}\rho_h^*dx_t   + 2L_f^2(L+L_f)^2h^2d\int_{\mathbb{R}^d} \left(\frac{\rho_t}{\rho_h^*}\right)^{q}\rho_h^*dx_t + \mathcal{O}(h^3)\\
= &\, \frac{2L^2h^2}{(1-hL)^2}G_q(\rho_t\|\rho_h^*) + 2L_f^2(L+L_f)^2h^2 d F_q(\rho_t\|\rho_h^*) + \mathcal{O}(h^3)\,.
\end{align*}
\end{proof}

Next, we are ready to prove Theorem \ref{thm_regu_density}.

\begin{proof}[Proof of Theorem \ref{thm_regu_density} part (1)]
By combining Lemma \ref{lemma_time_derivative} and \ref{lemma_dis_error}, we have
\begin{align*}
    &\frac{\partial }{\partial t} R_q(\rho_t\|\rho_h^*) =  \, -\frac{q}{2}\frac{G_q(\rho_t\|\rho_h^*)}{F_q(\rho_t\|\rho_h^*)} + \frac{q}{2F_q(\rho_t\|\rho_h^*)}\int_{\mathbb{R}^d}\|\Lambda(x_{t},x_{kh})\|_2^2\left(\frac{\rho_t}{\rho_h^*}\right)^{q}\rho_h^*dx_t \\
    \leq &\, \frac{q}{2}\left(-1+\frac{2L^2h^2}{(1-hL)^2}\right)\frac{G_q(\rho_t\|\rho_h^*)}{F_q(\rho_t\|\rho_h^*)} + q L_f^2(L+L_f)^2h^2d + \mathcal{O}(h^3)\,.
\end{align*}
Using the result in Lemma \ref{lemma_information}, i.e., when \(\rho_h^*\) satisfies the Poincaré inequality with constant \(\alpha_d\) (assumption A.3), we have
\[
\frac{G_q(\rho\|\rho_h^*)}{F_q(\rho\|\rho_h^*)} \geq \frac{4\alpha_d}{q^2} \left(1 - \exp(-R_q(\rho\|\rho_h^*))\right)\,.
\]
Hence, we arrive at
\[
  \frac{\partial }{\partial t} R_q(\rho_t\|\rho_h^*) \leq \frac{2\alpha_d}{q}\left(1 - \exp(-R_q(\rho_t\|\rho_h^*))\right)\left(-1+\frac{2L^2h^2}{(1-hL)^2}\right) + q L_f^2(L+L_f)^2h^2d + \mathcal{O}(h^3)\,,
\]
when \(h \leq  (\sqrt{2}-1)/L\).

Writing \(\rho_{k} = \rho_{kh}\). Then when \(R_q(\rho_0\|\rho_h^*) \geq 1\), it follows \(1 - \exp(-R_q(\rho_k\|\rho_h^*))\geq \frac{1}{2}\). In this case, we can derive the linear convergence given by
\[
R_q(\rho_k\|\rho_h^*) \leq R_q(\rho_0\|\rho_h^*) - kh\left(\frac{ \alpha_d}{q}\left(1-\frac{2L^2h^2}{(1-hL)^2}\right) - q L_f^2(L+L_f)^2h^2d\right) + \mathcal{O}(h^3)\,.
\]
For the case \(R_q(\rho_0\|\rho_h^*) < 1\),  we note that
\[
1 - \exp(-R_q(\rho_0\|\rho_h^*)) \geq R_q(\rho_0\|\rho_h^*) - \frac{R_q(\rho_0\|\rho_h^*)^2}{2} \geq \frac{1}{2}R_q(\rho_0\|\rho_h^*)\,.
\]
In this scenario, by integration with respect to $t$ from $0$ to $kh$, we have
\[
R_q(\rho_k\|\rho_h^*) \leq R_q(\rho_0\|\rho_h^*)\exp\left[-kh\frac{\alpha_d}{q} \left(1- \frac{2L^2h^2}{(1-hL)^2}\right)\right] + \frac{q^2L_f^2(L+L_f)^2h^2 d}{\alpha_d} + \mathcal{O}(h^3)\,.
\]
\end{proof}

\begin{proof}[Proof of Theorem \ref{thm_regu_density} part (2)]
The bound of  R\'enyi divergence between $\rho^*$ and $\rho_h^*$ can be derived using properties of Moreau envelope and Taylor expansion of $\log$ function, which leads to
\[
R_q(\rho^*\|\rho_h^*) = \frac{1}{q-1} \log \left( \int_{\mathbb{R}^d} \left( \frac{\rho^*}{\rho_h^*} \right)^q \rho_h^* \, dx \right) \leq \frac{q L_g^2 h}{q-1} + \mathcal{O}(h^2)\,.
\]
Additionally, we recall the following decomposition theorem for R\'enyi divergence  
\[
R_q(\rho_k\|\rho^*) \leq \left( \frac{q - \frac{1}{2}}{q-1} \right) R_{2q}(\rho^*\|\rho_h^*) + R_{2q-1}(\rho_k\|\rho_h^*)\,.
\]
Plug in the above two relations into part (1) of Theorem \ref{thm_regu_density}, and the desired result can be proved.
\end{proof}

\section{Postponed Proof in Section \ref{sec_analysis_nonsmooth}}
\label{appendix_nonsmooth_proof}

In the following, we establish the convergence of the sampling algorithm for nonsmooth convex functions $g$ under the assumptions B.1, B.2. Specifically, we provide a proof for  Lemma~\ref{lemma_kernel_nonsmooth}. To do so, we rely on the following lemma, which offers a quantitative Laplace approximation for Lipschitz functions. This result is a slight generalization of Proposition 4.5 in \cite{cbo_convergence}.

\begin{lemma}
\label{Lemma_laplacian}
Let $h, r, \tilde{r} > 0$. Let $\phi$ be a convex function with minimizer $y^*$. Suppose $f$ is an $L$-Lipschitz function in $B_r(y^*)$ for $r>0$. For a probability density function $\rho$, we have
\begin{align}
\label{eqn_laplacian_Lip}
 &\left| f(y^*) - \int f(y)\frac{\exp(-\phi(y)/h) \rho(y)}{\int \exp(-\phi(z)/h)\rho(z) dz} dy\right|
 \\ \leq &Lr + \frac{\exp(-(\phi(y_r^c)-\phi(y_{\tilde{r}}))/h)}{\int_{B_{\tilde{r}}(y^*)}\rho(z)dz} \int |f(y)-f(y^*)|\rho(y) dy\,,\notag
\end{align}
with $\phi(y_{\tilde{r}}) = \sup_{y\in B_{\tilde{r}}(y^*)} \phi(y)$ and $\phi(y_r^c) = \inf_{y\in B_r(y^*)^c} \phi(y)$, where $B_r(y^*)^c$ denotes the complement of $B_r(y^*)$ in $\mathbb{R}^d$.
\end{lemma}

\begin{proof}
For fixed $r$, we separate the integration into two parts as
\begin{align}
&\left|f(y^*) -  \int f(y)\frac{\exp(-\phi(y)/h) \rho(y)}{\int \exp(-\phi(z)/h)\rho(z) dz} dy\right|\\
= &\int_{B_r(y^*)} |f(y^*)-f(y)|\frac{\exp(-\phi(y)/h) \rho(y)}{\int \exp(-\phi(z)/h)\rho(z) dz} dy\notag\\
&+ \int_{B_r(v^*)^c} |f(y^*)-f(y)|\frac{\exp(-\phi(y)/h) \rho(y)}{\int \exp(-\phi(z)/h)\rho(z) dz } dy\,.\notag
\end{align}

Then for the first term, we can use the Lipsthiz condition on $f$, which implies it will be bounded above by $Lr$. 

The upper bound for the second term can be derived from Markov's inequality that
\[
\int_{\mathbb{R}^d} \exp\left(-\frac{\phi(z)}{h}\right)\rho(z) dz \geq \exp\left(-\frac{\phi(y_{\tilde{r}})}{h}\right) \int_{B_{\tilde{r}}(y^*)}\rho(z) dz\,,
\]
where $\phi(y_{\tilde{r}}) = \sup_{y\in B_{\tilde{r}}(y^*)}\phi(y)$.  

Hence, denote $\phi(y_r^c) = \inf_{y\in B_r(y^*)^c}\phi(y)$ and plug in the above Markov's inequality, we can have
\begin{align}
&\int_{B_r(y^*)^c} |f(y^*)-f(y)|\frac{\exp(-\phi(y)/h) \rho(y)}{\int \exp(-\phi(z)/h)\rho(z) dz } dy \\
\leq& \frac{\int_{B_r(y*)^c}|f(y^*)-f(y)|\rho(y)dy\exp(-(\phi(y_r^c)-\phi(y_{\tilde{r}}))/h)}{\int_{B_{\tilde{r}}(y^*)}\rho(z)dz}\notag\\
\leq & \frac{\exp(-(\phi(y_r^c)-\phi(y_{\tilde{r}}))/h)}{\int_{B_{\tilde{r}}(y^*)}\rho(z)dz}\int_{B_r(y^*)^c}|f(y)-f(y^*)|\rho(y)dy\,.\notag
\end{align}
\end{proof}

Next, we are ready to prove the Lemma \ref{lemma_kernel_nonsmooth}.

\medskip

\noindent
\textbf{Proof of Lemma \ref{lemma_kernel_nonsmooth}}
By the definition in \eqref{Reg_Was}, we would like to show
\begin{equation}
\left\|\frac{\nabla K_g^{h,\beta}\rho_k}{K_g^{h,\beta}\rho_k}(x) - \nabla \log \rho_k(x) \right\|\leq C_Kh^{1/2-\delta}\,,
\end{equation}
for small $0<\delta<1/2$ and $C_K>0$ that is independent of $h$ and only depends on $d$ and $C_R$ in the assumption B.2.

Firstly,  we consider $h > 0$ that satisfies
$$
h^{1/2 - \delta}  < R\,,
$$
which implies $h < R^{2/(1-2\delta)}$.

Recall that by a change of variable
\begin{align}
\label{eqn:concentration_heat}
&\frac{1}{(4\pi h/\beta)^{d/2}} \int_{B(0,\, h^{1/2 - \delta})^c} \exp\left(-\beta\frac{\|x\|_2^2}{4h}\right)\, dx = \int_{|z|\geq \sqrt{\beta/4}h^{-2\delta}} \exp(-\|z\|_2^2)dz\\
\leq& \pi^{d/2}\exp\left(-\beta\frac{h^{-2\delta}}{4}\right)\,,\notag
\end{align}
which decays to zero exponentially fast as $h \to 0$. This implies that the mass of the heat kernel is concentrated inside the ball $B(0,\, h^{1/2 - \delta})$.

In particular, when $x \in B(0, 2R)^c$, we have
$$
B(x,\, h^{1/2 - \delta}) \cap B(0, R ) = \emptyset.
$$
As a result, $\rho_{g}$ is smooth in $B(x,\, h^{1/2 - \delta})$, which ensures that the expansion carried out in the smooth case remains valid in this region with an error term of order $\exp(-\beta h^{-2\delta}/4)$.

Next, when $ x \in B(0,2R)$, for some constants $C_1$ $C_2$ that are independent of $h$, we would like to show
\begin{equation}
\label{C_1_def}
 \left|\rho_{g_h}^{1/2}(x)\int \frac{\rho_k(y)\exp(-\beta\|x-y\|_2^2/(4 h) )}{\int \rho_{g_h}^{1/2}(z)\exp(-\beta\|y-z\|_2^2/(4h)dz} dy -  \rho_k(x)\right|  \leq   C_1h^{1/2-\delta}\,,
\end{equation}
and 
\begin{equation}
\left\|\nabla\rho_{g_h}^{1/2}(x)\int \frac{\rho_k(y)\exp(-\beta\|x-y\|_2^2/(4 h) )}{\int \rho_{g_h}^{1/2}(z)\exp(-\beta\|y-z\|_2^2/(4h)dz} dy -\nabla \rho_k(x) \right\| \leq  C_2h^{1/2-\delta}\,.
\end{equation}

From Lemma \ref{Lemma_laplacian},  we take $f(y) = \rho_k(y)/\rho_g^{1/2}(y)$, $\phi(y) = C_h+\beta\|x-y\|_2^2/4$ where $C_h = d/2\log(4\pi h/\beta)$ as the normalization constant for the heat kernel that will be concealed later, $\rho(y) = \rho_{g_h}^{1/2}(y)$, $\tilde{r} = \sqrt{h}$, and $r = h^{1/2-\delta}$. We see immediately that $y^* = x$, $\phi(y_r^c) -\phi(y_{\tilde{r}})=  \beta h^{1-2\delta}- \beta h$. 

Then we note that the Lipchitz constant $L$ of $\rho_k/\rho_g^{1/2}$ satisfies
\[
L\leq \sup_{y \in B(0,3R)}\left\|\nabla
\frac{\rho_k}{\rho_{g_h}^{1/2}}(y)\right\|\leq \sup_{y\in B(0,3R)}\left\|\frac{\rho_k}{\rho_{g_h}^{1/2}}(y)\left(-\nabla\log\rho_k(y)+\frac{\beta}{2} \nabla g_h(y)\right)\right\| \,,
\]
where this term is bounded by $ 2C_R^3$ by assumption B.2.

Moreover, we know the numerator of the second term in \eqref{eqn_laplacian_Lip} is of constant order since
\begin{equation}
\label{Lemma_rho_K_f}
    \int |f(y)-f(y^*)|\rho_{g_h}^{1/2}(y)dy = 
\int \left|\frac{\rho(y)}{\rho_{g_h}^{1/2}(y)} -\frac{\rho(x)}{\rho_{g_h}^{1/2}(x)} \right|\rho_{g_h}^{1/2}(y) dy \leq1 + \rho(x)/\rho^{1/2}(x) \leq 1+C_R^2\,.
\end{equation}

Next, by assumption B.2 on the bound of $1/\rho_{g_h}$, we know that 
\[
\int_{B_{\tilde{r}}(y^*)}\rho_{g_h}^{1/2}(z)dz \geq 1/C_R h^d\omega_d\,,
\]
where $\omega_d$ is the volume of $d$ dimensional ball.
And finally, we arrive
\[
\frac{\exp(-(\phi(y_r^c)-\phi(y_{\tilde{r}}))/h)}{\int_{B_{\tilde{r}}(y^*)}\rho_{g_h}^{1/2}(z)dz} \int |f(y)-f(y^*)|\rho_{g_h}^{1/2}(y) dy \leq \exp\left(-\beta\frac{h^{-2\delta}}{4}\right)\exp\left(\frac{\beta}{4}\right)\frac{C_R(1+C_R^2)}{h^d \omega_d} \,.
\]
 The above helps us to conclude
\[
 \left|\rho_{g_h}^{1/2}(x)\int \frac{\rho_k(y)\exp(-\beta\|x-y\|_2^2/(4 h) )}{\int \rho_{g_h}^{1/2}(z)\exp(-\beta\|y-z\|_2^2/(4h)dz} dy- \rho_k(x)\right| \leq  2C_R^3 h^{1/2-\delta}+ C_3\frac{\exp\left(-\beta \frac{h^{-2\delta}}{4}\right)}{h^d}\,,
\]
where $C_3 = \ \exp\left(\frac{\beta}{4}\right)\frac{C_R(1+C_R^2)}{\omega_d} $. Then, when we have $h^{-1/2+\delta -d} \exp(-\beta h^{-2\delta}/4)\leq 2C_R^3/C_3$, we will have $C_1 = 4C_R^3$ in \eqref{C_1_def}. One relaxed and explicit upper bound for $h$ is  
\[h \leq \left[4/\beta\min\{|\ln(2C_R^3/C_3)|,\,(1/2+d-\delta)\}\right]^{1/(2\delta)}\,,
\]
which is obtained by taking the logarithm on both side of the above condition.

Finally, for the gradient term $\nabla K_{g_h}^{h}\rho_{k+\frac{1}{2}}$, compared with the previous case, the only difference is to pick $f(y) = (x-y)\rho_k(y)/\rho_{g_h}^{1/2}(y)$. By the chain rule, the Lipschitz constant of $f$ in this case will be $C_R^2+2RC_R^3$. Hence, we have
\begin{equation}
\left\| \nabla K_g^{h,\beta}\rho_{k} - \nabla  \rho_{k}(x)\right\|_2\leq  4\sqrt{d} (C_R^2+2RC_R^3)\,,
\end{equation}
for the $h$ that satisfies the previous condition. Combining the above, we have $C_K = 4\sqrt{d}(C_R^3+2RC_R^4) + 4C_R^4$ as desired. 
  

\medskip

\noindent\textbf{Proof of Lemma \ref{lemma_w2_trhok}}
In this proof, we first recall the definitions of the entropy term and the energy term from \eqref{def_energy}, and note that the KL divergence between \(\rho\) and \(\rho^*\) can be expressed as
\[
\beta^{-1}\KL(\rho\|\rho^*) = \mathcal{H}(\rho) + \mathcal{E}_{f}(\rho) + \mathcal{E}_g(\rho) \,.
\]

\textbf{Step 1:}
We would like to derive
\begin{equation}
2h\left[ \mathcal{E}_f(\rho_{k+\frac{1}{3}})-\mathcal{E}_f(\rho^*)\right]\leq (1-h\alpha)W_2^2(\rho_k,\rho^*)-W_2^2(\rho_{k+\frac{1}{3}},\rho^*)\,.
\end{equation}

By the definition of $x^{k+\frac{1}{3}}$, we have
\begin{align}
\label{step1_2norm}
&\|x^{k+\frac{1}{3}}-x^*\|_2^2 = \|x^{k}-x^*\|_2^2 + h^2\|\nabla f(x^k)\|_2^2 + 2h\langle \nabla f(x^k),  x^*-x^k\rangle \\
\leq & (1-h\alpha)\|x^{k}-x^*\|_2^2 + h^2\|\nabla f(x^k)\|_2^2  + 2h(f(x^*)-f(x^k))
\,.\notag
\end{align}
where we used that for an $\alpha$ strongly convex function
\[
f(x^*)-f(x^k)\geq \langle \nabla f(x^k), x^*-x^k\rangle + \frac{\alpha}{2}\|x^k-x^*\|_2^2\,.
\]

Then, recalling the $L-$smoothness assumption of $f$ and noting that
\[
f(x^{k+\frac{1}{3}})\leq f(x^{k}) - h\|\nabla f(x^k)\|_2^2  + \frac{h^2L_f}{2}\|\nabla f(x^k)\|_2^2 = f(x^k) - h\left( 1- \frac{hL_f}{2}\right)\|\nabla f(x^k)\|_2^2\,,
\]
which implies
\[
 h\|\nabla f(x^k)\|_2^2 \leq  2(f(x^{k})- f(x^{k+\frac{1}{3}}))
\]
when $h < 1/L_f$. 

Substituting this into \eqref{step1_2norm}, we arrive
\[
\|x^{k+\frac{1}{3}}-x^*\|_2^2  \leq (1-h\alpha)\|x^k-x^*\|_2^2 + 2h(f(x^*)- f(x^{k+\frac{1}{3}}))\,.
\]
Taking the expectation of all the random variables, as the right-hand side gives an upper bound for all possible couplings, we have
\begin{align*}
&W_2^2(\rho_{k+\frac{1}{3}},\rho*) \leq (1-h\alpha)\mathbb{E}_{x_k,x^*}(\|x_k-x^*\|_2^2) + 2h(\mathcal{E}_f(\rho^*)-\mathcal{E}_f(\rho_{k+\frac{1}{3}}))\\
\leq & (1-h\alpha)W_2^2(\rho_k,\rho^*) + 2h(\mathcal{E}_f(\rho^*)-\mathcal{E}_f(\rho_{k+\frac{1}{3}}))\,.
\end{align*}

\textbf{Step 2:}
We would like to show
\begin{equation}
2h\left[\mathcal{E}_f(\rho_{k+\frac{2}{3}})-\mathcal{E}_f(\rho_{k+\frac{1}{3}})\right] \leq  Ch^2 \,,
\end{equation}
for some constant $C$ independent of $h$.

By the convexity and $L_f$-smoothness of $f$, we note
\[
f(x^{k+\frac{2}{3}})-f(x^{k+\frac{1}{3}}) + \langle \nabla f(x^{k+\frac{1}{3}}),x^{k+\frac{1}{3}}-x^{k+\frac{2}{3}}\rangle \leq \frac{L_f}{2}\|x^{k+\frac{1}{3}}-x^{k+\frac{2}{3}}\|_2^2 \,.
\]
Taking the expectation on both sides, we obtain
\begin{align*}
&\left[\mathcal{E}_f(\rho_{k+\frac{2}{3}})-\mathcal{E}_f(\rho_{k+\frac{1}{3}})\right] \leq  h\int \nabla f\cdot \nabla  g_h \rho_{k+\frac{1}{3}} dx + \frac{L_f}{2}h^2\int \|\nabla g_h\|_2^2\rho_{k+\frac{1}{3}} dx\,.
\end{align*}
Then by the assumption that $f$ is L-smooth and assumptions B.1  and B.2 of $g_h$, we have $\|\nabla f(x)\|\leq |x| L_f + C_f$ and $\|\nabla g_h\| \leq |x|L + C_g$ for some constant $C_f$ $C_g$ independent of $h$. Hence, we can arrive at the desired conclusion since $\rho_k$ has the finite second-order moment. 

\textbf{Step 3:}
In this step, we will use the convexity of the nonsmooth term and the property of the Moreau envelope to derive
\begin{equation}
2h\left[ \mathcal{E}_g(\rho_{k+\frac{2}{3}})-\mathcal{E}_g(\rho^*)\right]\leq W_2^2(\rho_{k+\frac{1}{3}},\rho^*)-W_2^2(\rho_{k+\frac{2}{3}},\rho^*)\,.
\end{equation}

To see that, by the definition of iteration steps, we have
\begin{align}
\label{g_energy}
&\|x^{k+\frac{2}{3}}-y^*\|_2^2 = \|x^{k+\frac{1}{3}}-y^*\|_2^2 + h^2 \|\nabla g_h(x^{k+\frac{1}{3}})\|_2^2 - 2h \langle \nabla g_h(x^{k+\frac{1}{3}}), x^{k+\frac{1}{3}}-y^*\rangle \\
=&\|x^{k+\frac{1}{3}}-y^*\|_2^2 - h^2 \|\nabla g_h(x^{k+\frac{1}{3}})\|_2^2 - 2h \langle \nabla g_h(x^{k+\frac{1}{3}}), x^{k+\frac{2}{3}}-y^*\rangle\notag\,.
\end{align}
Then, by the property of the Moreau envelope, we know that
\[
\nabla g_h(x^{k+\frac{1}{3}}) \in \partial g(\prox_g^h(x^{k+\frac{1}{3}})) = \partial g(x^{k+\frac{2}{3}})\,.
\]
In this case, by the convexity of $g$, we have
\[
\langle \nabla g_h(x^{k+\frac{1}{3}}), x^{k+\frac{2}{3}}-y^*\rangle \geq g(x^{k+\frac{2}{3}})- g(y^*)\,.
\]
Plug this into \eqref{g_energy}, we arrive
\[
\|x^{k+\frac{2}{3}}-y^*\|_2^2 \leq \|x^{k+\frac{1}{3}}-y^*\|_2^2   - 2h(g(x^{k+\frac{2}{3}})- g(y^*))\notag\,.
\]
Taking the expectation on both sides gives the desired result. 
 
\noindent\textbf{Step 4:}
In this step, we will derive the third step of iteration \eqref{particle_x_k} leads to
\begin{equation}
  2h(\mathcal{H}(\rho_{k+\frac{2}{3}})-\mathcal{H}(\rho^*)) \leq (1+h^{3/2-\delta}\beta^{-1})W_2^2(\rho_{k+\frac{2}{3}},\rho^*)-W_2^2(\rho_{k+1},\rho^*) + h^{3/2-\delta}\beta^{-1}C_K^2 + \mathcal{O}(h^2)\,.
\end{equation}
To derive this, by the iterative relationship, we have
\begin{align}
\label{eqn_mid_step_H}
&\|x^{k+1}-x^*\|_2^2 \\
= &\|x^{k+\frac{2}{3}}-x^*\|_2^2 + h^2 \beta^{-2}\|\nabla\log K\rho_{k+\frac{1}{3}}(x^{k+\frac{1}{3}})\|_2^2 - 2h\beta^{-1}\langle \nabla\log K\rho_{k+\frac{1}{3}}(x^{k+\frac{1}{3}}),x^{k+\frac{2}{3}}-x^*\rangle\notag\\
=& \|x^{k+\frac{2}{3}}-x^*\|_2^2 + h^2 \beta^{-2}\|\nabla \log K\rho_{k+\frac{1}{3}}(x^{k+\frac{1}{3}})\|_2^2 - 2h \beta^{-1}\langle \nabla\log\rho_{k+\frac{1}{3}}(x^{k+\frac{1}{3}}),x^{k+\frac{1}{3}}-x^*\rangle\notag\\
&+2h^2  \beta^{-1}\langle \nabla\log\rho_{k+\frac{1}{3}}(x^{k+\frac{1}{3}}), \nabla g_h(x^{k+\frac{1}{3}})\rangle\notag\\
&+2h  \beta^{-1}\langle (\nabla\log\rho_{k+\frac{1}{3}}-\nabla\log K\rho_{k+\frac{1}{3}})(x^{k+\frac{1}{3}}), x^{k+\frac{2}{3}}-x^*\rangle \,.\notag
\end{align}
After taking the expectation of both sides, for the second-to-last term, integration by parts implies
\begin{align*}
&\int \nabla\log\rho_{k+\frac{1}{3}}  \nabla g_h\rho_{k+\frac{1}{3}} dy   = \int \rho_{k+\frac{1}{3}}\Delta g_h dy\\
=&\int_{B(0,R)} \rho_{k+\frac{1}{3}}\Delta g_h dy + \int_{B(0,R)^c}\rho_{k+\frac{1}{3}}\Delta g_h dy \leq C_d(h^{d-1} + L_{g_h})\,,
\end{align*}
for some constant $C_d$ that is independent of $h$ where we used $\rho_{k+\frac{1}{3}}\leq C_R$ in $B(0,R)$ and $g_h$ is $L$-smooth outside the $B(0,R)$. 

Moreover, by the Cauchy-Schwarz inequality, the last term in \eqref{eqn_mid_step_H} is bounded by
\begin{align*}
&\langle (\nabla\log\rho_{k+\frac{1}{3}}-\nabla\log K\rho_{k+\frac{1}{3}})(x^{k+\frac{1}{3}}), x^{k+\frac{2}{3}}-x^*\rangle\\
\leq & h^{1/2-\delta}C_K \|x^{k+\frac{2}{3}}-x^*\|_2 \leq h^{1/2-\delta}\frac{(C_K^2+\|x^{k+\frac{2}{3}}-x^*\|_2^2)}{2}\,.
\end{align*}

We next will employ the fact that the heat equation generates the Wasserstein gradient flow of the entropy term $\mathcal{H}$ to bound $\langle \nabla\log \rho_{k+\frac{1}{3}}(x^{k+\frac{1}{3}}),x^{k+\frac{1}{3}}-x^*\rangle $.
Let \( T: \mathbb{R}^d \to \mathbb{R}^d \) be the optimal transport map pushing \( \rho_{k+\frac{1}{3}} \) forward to \( \rho^* \), i.e., \( T_\# \rho_{k+\frac{1}{3}} = \rho^* \).
Define the interpolation \( X_t(x) = (1 - t)x + t T(x) \), and set \( \rho_t = (X_t)_\# \rho_{k+\frac{1}{3}} \), which yields \( \rho_0 = \rho_{k+\frac{1}{3}} \), \( \rho_1 = \rho^* \), and \( (\rho_t)_{t \in [0,1]} \) forms the constant-speed geodesic in the Wasserstein space. The displacement convexity of the entropy \( \mathcal{H}(\rho) = \beta^{-1}\int \rho \log \rho \, dx \) implies
\begin{equation}
\mathcal{H}(\rho_{k+\frac{1}{3}}) - \mathcal{H}(\rho^*) \leq -\left.\frac{d}{dt} \mathcal{H}(\rho_t)\right|_{t=0}.
\end{equation}
To compute the derivative explicitly, use the change-of-variable formula: \( \rho_t(X_t) = \frac{\rho_{k+\frac{1}{3}}}{\det \nabla X_t} \), so
\[
\mathcal{H}(\rho_t) = \beta^{-1}\int \rho_t(y) \log \rho_t(y) \, dy = \beta^{-1}\int \rho_{k+\frac{1}{3}}(x) \log \left( \frac{\rho_{k+\frac{1}{3}}(x)}{\det \nabla X_t(x)} \right) dx.
\]
Differentiating at \( t = 0 \), we get
\[
\left. \frac{d}{dt} \mathcal{H}(\rho_t) \right|_{t=0} = -\beta^{-1}\int \rho_{k+\frac{1}{3}}(x) \left.\frac{d}{dt} \log \det \nabla X_t(x)\right|_{t=0} dx.
\]
Using Jacobi's formula and \( \nabla X_t(x) = (1 - t) I + t \nabla T(x) \), we compute
\begin{equation*}
\left.\frac{d}{dt} \log \det \nabla X_t(x)\right|_{t=0} = \operatorname{Tr}(\nabla T(x) - I) = \nabla \cdot (T(x) - x),
\end{equation*}
so the final result is
\begin{equation}
\mathcal{H}(\rho_{k+\frac{1}{3}})-\mathcal{H}(\rho^*) \leq -\beta^{-1}\int \nabla\log \rho_{k+\frac{1}{3}}(x)(T(x)-x) \rho_{k+\frac{1}{3}}(x)dx\,.
\end{equation}

Hence, we finally have the relation
\begin{align*}
 &W_2^2(\rho_{k+1},\rho^*) + 2h(\mathcal{H}(\rho_{k+\frac{1}{3}})-\mathcal{H}(\rho^*)) \\
 \leq &(1+h^{3/2-\delta}\beta^{-1})W_2^2(\rho_{k+\frac{2}{3}},\rho^*) + h^{3/2-\delta}\beta^{-1}C_K^2 +  \mathcal{O}(h^2)\,,
\end{align*} 
where the desired one is obtained by combining the above and $x^{k+\frac{2}{3}} = x^{k+\frac{1}{3}} - h\nabla g_h(x^{k+\frac{1}{3}})$.

\textbf{Step 5:}
Summing up the results in steps 1 to 5 to have
\begin{align}
&2h\beta^{-1}\KL(\rho_{k+\frac{2}{3}}\|\rho^*)\\
\leq&  (1-h\alpha)W_2^2(\rho_k,\rho^*)-W_2^2(\rho_{k+1},\rho^*) +  h^{3/2-\delta}\beta^{-1}C_K^2 + h^{3/2-\delta}\beta^{-1} W^2_2(\rho_{k+\frac{2}{3}},\rho^*) + \mathcal{O}(h^2)\,.\notag
\end{align}

\medskip
 
\noindent\textbf{Proof for Theorem \ref{thm_KL_nonsmooth}}
By the Assumption A.1., we recall the Talagrand inequality that for strongly log-concave $\rho^*$, we have
\begin{equation}
W_2^2(\rho,\rho^*) \leq \frac{2}{\alpha}\beta^{-1}\KL(\rho,\rho^*)\,.
\end{equation}
 
\noindent\textbf{Proof for part (1).}
Summing the inequality from Lemma \ref{lemma_w2_trhok} over \(j = 0, \ldots, k\) yields
\begin{align*}
&2h\sum_{j=0}^k \KL(\rho_{j+\frac{2}{3}} \| \rho^*) \\
\leq& \beta W_2^2(\rho_0, \rho^*) +   2(k+1)h^{3/2-\delta} C_K^2 +h^{3/2-\delta} \sum_{j=1}^k W_2^2(\rho_{j+\frac{2}{3}},\rho^*) + \mathcal{O}(h^2)\\
\leq & \beta W_2^2(\rho_0, \rho^*) +   (k+1)h^{3/2-\delta}  C_K^2 +\frac{2h^{3/2-\delta}}{\alpha}\beta^{-1}\sum_{j=1}^k\KL(\rho_{j+\frac{2}{3}}\|\rho^*) + \mathcal{O}(h^2)\,.
\end{align*}
Defining the averaged density
\[
\trho_{k} = \frac{1}{k+1} \sum_{j=0}^k \rho_j,
\]
and applying the convexity of the KL divergence with respect to $\rho_j$, we obtain
\[
\KL(\trho_{k} \| \rho^*) \leq \frac{1}{k+1} \sum_{j=0}^k \KL(\rho_j \| \rho^*).
\]
Combining these inequalities and letting $t_k = h(k+1)$ give the following result
\begin{equation}
\KL(\trho_k\|\rho^*)\leq \frac{\alpha\beta}{\alpha\beta -h^{1/2-\delta}}\left[\frac{\beta}{2t_k}W_2^2(\rho_0\|\rho^*)+h^{1/2-\delta}C_K^2\right] + \mathcal{O}(h)
\end{equation}

\noindent\textbf{Proof for part (2).}
We rewrite the inequality \eqref{eqn:lemma_w2} as
\begin{align}
&W_2^2(\rho_{k+1},\rho^*)\\
\leq &(1-h\alpha)W_2^2(\rho_k,\rho^*) - h\beta^{-1}(2\KL(\rho_{k+\frac{2}{3}}\|\rho^*)-h^{1/2-\delta}W_2^2(\rho_{k+\frac{2}{3}},\rho^*))+h^{3/2-\delta}\beta^{-1}C_K^2+ \mathcal{O}(h)\notag\\
\leq &(1-h\alpha)W_2^2(\rho_k,\rho^*) - h\beta^{-1}\left(\frac{4}{\alpha\beta}-h^{1/2-\delta}\right)W_2^2(\rho_{k+\frac{2}{3}},\rho^*)+h^{3/2-\delta}\beta^{-1}C_K^2+ \mathcal{O}(h)\notag\\
\leq &(1-h\alpha)W_2^2(\rho_k,\rho^*) + h^{3/2-\delta}\beta^{-1}C_K^2+ \mathcal{O}(h)\,,\notag
\end{align}
where we have used the assumption that 
$h\leq (4/(\alpha\beta))^{2/(1-2\delta)}$.
Iterating over $k$, we arrive
\begin{equation}
W_2^2(\rho_{k},\rho^*)\leq (1-h\alpha)^{k}W_2^2(\rho_0,\rho^*) + h^{1/2-\delta}\frac{C_K^2}{\alpha\beta}+ \mathcal{O}(h)\,.
\end{equation}

\section{Details about Numerical Experiments}

\textbf{Evaluation of marginal distribution in Example 1.}
We can integrate the mixture of Gaussian and Laplace models exactly. If $\Sigma_i^{-1} = 1/(2\sigma^2) I_d$, the integral is given by
\begin{align*}
   & \int_{\mathbb{R}^d} \rho^*(x) \, dx \\= & \sum_{n=1}^N \prod_{j=1}^d \int_{\mathbb{R}} \exp\left( -\frac{(x_j - y_{n,j})^2}{2\sigma_n^2} - \lambda |x_j| \right) dx_j \\
    = & \sum_{n=1}^N \prod_{j=1}^d \left[ \int_{-(y_{n,j} - \lambda \sigma_n^2)}^{\infty} \exp\left( -\frac{z_j^2}{2\sigma_n^2} \right) dz_j \exp\left( -\frac{y_{n,j}^2 - (y_{n,j} - \lambda \sigma_n^2)^2}{2\sigma_n^2} \right) \right. \\
    &\quad + \left. \int_{-\infty}^{-(y_{n,j} + \lambda \sigma_n^2)} \exp\left( -\frac{z_j^2}{2\sigma_n^2} \right) dz_j \exp\left( -\frac{y_{n,j}^2 - (y_{n,j} + \lambda \sigma_n^2)^2}{2\sigma_n^2} \right) \right] \\
    = & \frac{1}{(\sqrt{2}\sigma_n)^d} \sum_{n=1}^N \prod_{j=1}^d \left[ \int_{-\frac{(y_{n,j} - \lambda \sigma_n^2)}{\sqrt{2}\sigma_n}}^{\infty} \exp\left( -z_j^2 \right) dz_j \exp\left( -\frac{y_{n,j}^2 - (y_{n,j} - \lambda \sigma_n^2)^2}{2\sigma_n^2} \right) \right. \\
    &\quad + \left. \int_{-\infty}^{-\frac{(y_{n,j} + \lambda \sigma_n^2)}{\sqrt{2}\sigma_n}} \exp\left( -z_j^2 \right) dz_j \exp\left( -\frac{y_{n,j}^2 - (y_{n,j} + \lambda \sigma_n^2)^2}{2\sigma_n^2} \right) \right]\,.
\end{align*}
The above computation provides the normalization constant \(Z\). By replacing the integration over \(\mathbb{R}^d\) with an integration over \(\mathbb{R}^{d-1}\), we obtain the formula for the marginal distribution \(\rho_1^*\).

\end{appendix}
\bibliographystyle{plain}
\bibliography{ref}

\begin{thebibliography}{10}

\bibitem{ambrosio2008gradient}
L.~Ambrosio, N.~Gigli, and G.~Savar{\'e}.
\newblock {\em Gradient flows: in metric spaces and in the space of probability measures}.
\newblock Springer Science \& Business Media, 2008.

\bibitem{benko2024langevin}
M.~Benko, I.~Chlebicka, J.~Endal, and B.~Miasojedow.
\newblock {L}angevin {M}onte {C}arlo beyond lipschitz gradient continuity.
\newblock {\em arXiv preprint arXiv:2412.09698}, 2024.

\bibitem{LMC_JKO_splitting}
E.~Bernton.
\newblock {L}angevin {M}onte {C}arlo and {JKO} splitting.
\newblock In {\em Conference on Learning Theory}, pages 1777--1798. PMLR, 2018.

\bibitem{castin2025unified}
V.~Castin, P.~Ablin, J.~A. Carrillo, and G.~Peyr{\'e}.
\newblock A unified perspective on the dynamics of deep transformers.
\newblock {\em arXiv preprint arXiv:2501.18322}, 2025.

\bibitem{chambolle2004algorithm}
A.~Chambolle.
\newblock An algorithm for total variation minimization and applications.
\newblock {\em Journal of Mathematical Imaging and Vision}, 20:89--97, 2004.

\bibitem{chen2023improved}
H.~Chen, H.~Lee, and J.~Lu.
\newblock Improved analysis of score-based generative modeling: User-friendly bounds under minimal smoothness assumptions.
\newblock In {\em International Conference on Machine Learning}, pages 4735--4763. PMLR, 2023.

\bibitem{chen2022improved_proximal}
Y.~Chen, S.~Chewi, A.~Salim, and A.~Wibisono.
\newblock Improved analysis for a proximal algorithm for sampling.
\newblock In {\em Conference on Learning Theory}, pages 2984--3014. PMLR, 2022.

\bibitem{chu2022application}
J.~Chu, N.~A. Sun, W.~Hu, X.~Chen, N.~Yi, and Y.~Shen.
\newblock The application of {B}ayesian methods in cancer prognosis and prediction.
\newblock {\em Cancer Genomics \& Proteomics}, 19(1):1--11, 2022.

\bibitem{condat_Vu}
L.~Condat.
\newblock A primal--dual splitting method for convex optimization involving lipschitzian, proximable and linear composite terms.
\newblock {\em Journal of optimization theory and applications}, 158(2):460--479, 2013.

\bibitem{review_splitting}
L.~Condat, D.~Kitahara, A.~Contreras, and A.~Hirabayashi.
\newblock Proximal splitting algorithms for convex optimization: A tour of recent advances, with new twists.
\newblock {\em SIAM Review}, 65(2):375--435, 2023.

\bibitem{blob_sampling}
K.~Craig, K.~Elamvazhuthi, M.~Haberland, and O.~Turanova.
\newblock A blob method for inhomogeneous diffusion with applications to multi-agent control and sampling.
\newblock {\em Mathematics of Computation}, 92(344):2575--2654, 2023.

\bibitem{dalalyan2017theoretical}
A.~S. Dalalyan.
\newblock Theoretical guarantees for approximate sampling from smooth and log-concave densities.
\newblock {\em Journal of the Royal Statistical Society Series B: Statistical Methodology}, 79(3):651--676, 2017.

\bibitem{durmus2019analysis}
A.~Durmus, S.~Majewski, and B.~Miasojedow.
\newblock Analysis of {L}angevin {M}onte {C}arlo via convex optimization.
\newblock {\em Journal of Machine Learning Research}, 20(73):1--46, 2019.

\bibitem{durmus2018efficient}
A.~Durmus, E.~Moulines, and M.~Pereyra.
\newblock Efficient {B}ayesian computation by proximal {M}arkov chain {M}onte {C}arlo: when {L}angevin meets {M}oreau.
\newblock {\em SIAM Journal on Imaging Sciences}, 11(1):473--506, 2018.

\bibitem{cbo_convergence}
M.~Fornasier, T.~Klock, and K.~Riedl.
\newblock Consensus-based optimization methods converge globally.
\newblock {\em SIAM Journal on Optimization}, 34(3):2973--3004, 2024.

\bibitem{emprical_rate}
N.~Fournier and A.~Guillin.
\newblock On the rate of convergence in {W}asserstein distance of the empirical measure.
\newblock {\em Probability Theory and Related Fields}, 162(3):707--738, 2015.

\bibitem{geshkovski2023mathematical}
B.~Geshkovski, C.~Letrouit, Y.~Polyanskiy, and P.~Rigollet.
\newblock A mathematical perspective on transformers.
\newblock {\em arXiv preprint arXiv:2312.10794}, 2023.

\bibitem{habring2024subgradient}
A.~Habring, M.~Holler, and T.~Pock.
\newblock Subgradient {L}angevin methods for sampling from nonsmooth potentials.
\newblock {\em SIAM Journal on Mathematics of Data Science}, 6(4):897--925, 2024.

\bibitem{han2024convergence}
F.~Han, S.~Osher, and W.~Li.
\newblock Convergence of noise-free sampling algorithms with regularized {W}asserstein proximals.
\newblock {\em arXiv preprint arXiv:2409.01567}, 2024.

\bibitem{TT_BRWP}
F.~Han, S.~Osher, and W.~Li.
\newblock Tensor train based sampling algorithms for approximating regularized {W}asserstein proximal operators.
\newblock {\em arXiv preprint arXiv:2401.13125}, 2024.

\bibitem{jko1998}
R.~Jordan, D.~Kinderlehrer, and F.~Otto.
\newblock The variational formulation of the {F}okker-{P}lanck equation.
\newblock {\em SIAM Journal on Mathematical Analysis}, 29(1):1--17, 1998.

\bibitem{latala2000between}
R.~Lata{\l}a and K.~Oleszkiewicz.
\newblock Between {S}obolev and {P}oincar{\'e}.
\newblock In {\em Geometric Aspects of Functional Analysis: Israel Seminar 1996--2000}, pages 147--168. Springer, 2000.

\bibitem{pock_sampling}
T.~T.-K. Lau, H.~Liu, and T.~Pock.
\newblock Non-log-concave and nonsmooth sampling via {L}angevin {M}onte {C}arlo algorithms.
\newblock In {\em INdAM Workshop: Advanced Techniques in Optimization for Machine learning and Imaging}, pages 83--149. Springer, 2022.

\bibitem{Lee2021RGO}
Y.~T. Lee, R.~Shen, and K.~Tian.
\newblock Structured logconcave sampling with a restricted {G}aussian oracle.
\newblock In {\em Proceedings of Thirty Fourth Conference on Learning Theory}, pages 2993--3050. PMLR, 2021.

\bibitem{kernel_proximal}
W.~Li, S.~Liu, and S.~Osher.
\newblock A kernel formula for regularized {W}asserstein proximal operators.
\newblock {\em Research in the Mathematical Sciences}, 10(4):43, 2023.

\bibitem{liang2022proximal}
J.~Liang and Y.~Chen.
\newblock A proximal algorithm for sampling from non-smooth potentials.
\newblock In {\em 2022 Winter Simulation Conference (WSC)}, pages 3229--3240. IEEE, December 2022.

\bibitem{svgd_2016}
Q.~Liu and D.~Wang.
\newblock Stein variational gradient descent: A general purpose {B}ayesian inference algorithm.
\newblock {\em Advances in Neural Information Processing Systems}, 29, 2016.

\bibitem{RGO_L1}
W.~Mou, N.~Flammarion, M.~J. Wainwright, and P.~L. Bartlett.
\newblock An efficient sampling algorithm for non-smooth composite potentials.
\newblock {\em Journal of Machine Learning Research}, 23(233):1--50, 2022.

\bibitem{psycho_bay_lasso}
J.~Pan, E.~H. Ip, and L.~Dub{\'e}.
\newblock An alternative to post hoc model modification in confirmatory factor analysis: The {B}ayesian {L}asso.
\newblock {\em Psychological Methods}, 22(4):687, 2017.

\bibitem{bayesian_lasso_2008}
T.~Park and G.~Casella.
\newblock The {B}ayesian {L}asso.
\newblock {\em Journal of the American Statistical Association}, 103(482):681--686, 2008.

\bibitem{pereyra2016proximal}
M.~Pereyra.
\newblock Proximal {M}arkov chain {M}onte {C}arlo algorithms.
\newblock {\em Statistics and Computing}, 26:745--760, 2016.

\bibitem{TV_1992}
L.~I. Rudin, S.~Osher, and E.~Fatemi.
\newblock Nonlinear total variation based noise removal algorithms.
\newblock {\em Physica D: Nonlinear Phenomena}, 60(1-4):259--268, 1992.

\bibitem{salim_splitting}
A.~Salim, D.~Kovalev, and P.~Richt{\'a}rik.
\newblock Stochastic proximal {L}angevin algorithm: Potential splitting and nonasymptotic rates.
\newblock {\em Advances in Neural Information Processing Systems}, 32, 2019.

\bibitem{sander2022sinkformers}
M.~Sander, P.~Ablin, M.~Blondel, and G.~Peyr{\'e}.
\newblock Sinkformers: Transformers with doubly stochastic attention.
\newblock In {\em International Conference on Artificial Intelligence and Statistics}, pages 3515--3530. PMLR, 2022.

\bibitem{BRWP_2023}
H.~Y. Tan, S.~Osher, and W.~Li.
\newblock Noise-free sampling algorithms via regularized {W}asserstein proximals.
\newblock {\em Research in the Mathematical Sciences}, 11(4):65, 2024.

\bibitem{tibshirani2025laplace}
R.~Tibshirani, S.~Fung, H.~Heaton, and S.~Osher.
\newblock {L}aplace meets {M}oreau: Smooth approximation to infimal convolutions using {L}aplace's method.
\newblock {\em Journal of Machine Learning Research}, 26(72):1--36, 2025.

\bibitem{ula_wibisono}
S.~Vempala and A.~Wibisono.
\newblock Rapid convergence of the unadjusted {L}angevin algorithm: Isoperimetry suffices.
\newblock {\em Advances in Neural Information Processing Systems}, 32, 2019.

\bibitem{priors_bayesian_NN}
M.~Vladimirova, J.~Verbeek, P.~Mesejo, and J.~Arbel.
\newblock Understanding priors in {B}ayesian neural networks at the unit level.
\newblock In {\em International Conference on Machine Learning}, pages 6458--6467. PMLR, 2019.

\bibitem{wang2022accelerated}
Y.~Wang and W.~Li.
\newblock Accelerated information gradient flow.
\newblock {\em Journal of Scientific Computing}, 90:1--47, 2022.

\bibitem{pula_wibisono}
A.~Wibisono.
\newblock Proximal {L}angevin algorithm: Rapid convergence under isoperimetry.
\newblock {\em arXiv preprint arXiv:1911.01469}, 2019.

\bibitem{l12_cs}
P.~Yin, Y.~Lou, Q.~He, and J.~Xin.
\newblock Minimization of $\ell_{1-2}$ for compressed sensing.
\newblock {\em SIAM Journal on Scientific Computing}, 37(1):A536--A563, 2015.

\end{thebibliography}

\end{document}